\documentclass[sigplan,screen]{acmart}
\newif\iflongversion
\longversionfalse

\usepackage{amsmath}
\usepackage{xfrac}

\usepackage{paralist}

\usepackage[prefixflatinterpret,prefixflatbracketmodalinterpret,prefixflatsetinterpret,sidenotecalculus,longseqcontext,boldmquant]{logic}
\usepackage[prefixflatinterpret,prefixflatbracketmodalinterpret,differentialdL,simplenames]{dL}

\usepackage{math}

\usepackage{verbatim}

\iflongversion
	\usepackage[conf={normal}]{proof-at-the-end}

\else
	\usepackage{proof-at-the-end}	
\fi

\setcopyright{none}
\copyrightyear{2022}
\acmYear{2022}
\acmDOI{}

\acmConference[]{}{}{}
\acmBooktitle{}
\acmPrice{}
\acmISBN{}

\newcommand{\ud}{\mathrm{d}}

\newcommand{\Vc}{\mathcal{V}}

\newcommand{\Sc}{\mathcal{S}}
\newcommand{\Pc}{\mathcal{P}}

\newcommand{\Rc}{\mathcal{R}}

\newcommand{\Rb}{\mathbb{R}}

\newcommand{\Nb}{\mathbb{N}}

\newcommand{\Rf}{\mathfrak{R}}

\newcommand{\Lp}{\left (}
\newcommand{\Rp}{\right )}

\ifdefined\lenvelope\ifdefined\envelope\else
	\newcommand*{\envelope}[1]{\lenvelope{#1}\renvelope}
\fi\fi

\usepackage{prettyref}
\newcommand{\rref}[2][]{\prettyref{#2}}
\newrefformat{sec}{Section\,\ref{#1}}
\newrefformat{def}{Def.\,\ref{#1}}
\newrefformat{thm}{Theorem\,\ref{#1}}
\newrefformat{prop}{Proposition\,\ref{#1}}
\newrefformat{lem}{Lemma\,\ref{#1}}
\newrefformat{cor}{Corollary\,\ref{#1}}
\newrefformat{ex}{Example\,\ref{#1}}
\newrefformat{tab}{Table\,\ref{#1}}
\newrefformat{fig}{Fig.\,\ref{#1}}
\newrefformat{app}{Appendix\,\ref{#1}}

\usepackage{IEEEtrantools}

\newcommand{\mutex}{\texorpdfstring{\ensuremath{\mu}\xspace}{Mu\xspace}}

\author{Noah Abou El Wafa}
\email{nabouelw@andrew.cmu.edu}
\orcid{0000-0002-3987-9919}

\affiliation{
  \institution{Carnegie Mellon University}
  \streetaddress{5000 Forbes Avenue}
  \city{Pittsburgh}
  \state{Pennsylvania}
  \country{USA}
  \postcode{15213}
}

\author{Andr\'e Platzer}
\email{aplatzer@cs.cmu.edu}
\orcid{0000-0001-7238-5710}

\affiliation{
  \institution{Carnegie Mellon University}
  \streetaddress{5000 Forbes Avenue}
  \city{Pittsburgh}
  \state{Pennsylvania}
  \country{USA}
  \postcode{15213}
}

\title{First-Order Game Logic and Modal \mutex-Calculus}

\newcommand{\dLmu}{\ensuremath{\ud\mathsf{L}\mu{}}\xspace}
\newcommand{\dLmusubst}{{\ud\mathsf{L}\mu}}

\renewcommand{\domain}[1]{|#1|}

\renewcommand{\subst}[3][]{#1\tfrac{#3}{#2}}
\renewcommand{\modif}[3]{#1\tfrac{#3}{#2}}
\renewcommand{\models}{\vDash}

\renewcommand{\ltrue}{\top}%
\renewcommand{\lfalse}{\bot}%

\renewcommand*{\infers}[1][]{\vdash_{\scriptscriptstyle#1}}

\newcommand{\transposepv}[3]{{#1}^{#3}_{#2}}
\newcommand{\transposestate}[3]{{#1}^{#3}_{#2}}
\newcommand{\transposeterm}[3]{{#1}^{#3}_{#2}}
\newcommand{\transposefml}[3]{{#1}^{#3}_{#2}}
\newcommand{\transposeact}[3]{{#1}^{#3}_{#2}}

\newcommand{\Act}{\mathrm{Act}}
\newcommand{\structA}{\mathfrak{A}}

\newcommand{\Lmu}[1][\Lc]{\ensuremath{#1\mu{}}\xspace}
\newcommand{\Lmucalc}[1][\Lc]{\ensuremath{#1\mu}\xspace}
\newcommand{\LmuV}[1][\Lc]{\ensuremath{#1\mu{[\Vc]}}\xspace}

\newcommand{\GL}[1][\Lc]{\ensuremath{\mathrm{G}#1{}}\xspace}
\newcommand{\GLV}[1][\Lc]{\ensuremath{\mathrm{G}#1[\Vc]}\xspace}
\newcommand{\GLcalc}[1][\Lc]{\ensuremath{\mathrm{G}#1}\xspace}

\newcommand{\dem}{\mathrm{d}}
\newcommand{\demrep}{\times}
\renewcommand*{\pdual}[1]{#1^\dem}

\newcommand*{\envelopemu}[1]{\envelope{#1}_{\scriptscriptstyle\mu}}
\newcommand*{\envelopegl}[1]{\envelope{#1}_{\scriptscriptstyle\mathrm{GL}}}

\newcommand{\substgl}[1][\Lc]{_{_{\mathrm{G}#1}}}

\DeclareMathOperator{\FOL}{FO}
\DeclareMathOperator{\FOLL}{Lit}
\newcommand{\Xc}{\mathcal{X}}

\newcommand{\Lc}{\mathrm{L}}
\newcommand{\Lcf}{\mathcal{L}}

\newcommand{\stdlit}{p}
\newcommand{\stdtr}{a}
\newcommand{\stdterm}{\theta}
\newcommand{\stdasst}{\Omega}
\newcommand{\stdstate}{\omega}%
\newcommand{\newstate}{\nu}%

\newif\ifonectrlvar
\onectrlvartrue

\ifonectrlvar
\newcommand{\ctrlvar}[1]{\mathfrak{c}}
\newcommand{\ctrlval}[1]{\mathfrak{c}_{#1}}

\else
\newcommand{\ctrlvar}[1]{x_{#1}}
\newcommand{\ctrlval}[1]{1}

\fi

\newcommand{\setcval}[1]{\ctrlvar{#1}:=\ctrlval{#1}}
\newcommand{\testcval}[1]{\ptest{\ctrlvar{#1}=\ctrlval{#1}}}

\newcommand{\neqcval}[1]{{\ctrlvar{#1}\neq\ctrlval{#1}}}
\newcommand{\eqcval}[1]{{\ctrlvar{#1}=\ctrlval{#1}}}

\newcommand{\testcvalnot}[1]{\ptest{\ctrlvar{#1}\neq\ctrlval{#1}}}

\newcommand{\rank}{\mathrm{rk}}

\definecolor{darkishgray}{rgb}{.35,.35,.35}
\definecolor{vvblue}{rgb}{.14,.21,.868}%
\renewcommand{\linferenceRuleNameSeparation}{~~~~}
\renewcommand*{\irrulename}[1]{\text{\textcolor{darkishgray}{#1}}}

\newcommand{\ovariable}{object variable\xspace} %
\newcommand{\ovariables}{object variables\xspace}

\newcommand{\sovariable}{variable\xspace} %

\newcommand{\Ovariables}{Object variables\xspace}
\newcommand{\OVariables}{Object Variables\xspace}

\newcommand{\pvariable}{propositional variable\xspace} %
\newcommand{\pvariables}{propositional variables\xspace}

\newcommand{\Pvariable}{Propositional variable\xspace}
\newcommand{\Pvariables}{Propositional variables\xspace}
\newcommand{\PVariables}{Propositional Variables\xspace}

\newcommand{\assignment}{valuation\xspace}
\newcommand{\assignments}{valuations\xspace}

\begin{document}

\begin{abstract}
This paper investigates first-order game logic and first-order modal $\mu$-calculus, which extend their propositional modal logic counterparts with first-order modalities of interpreted effects such as variable assignments.
Unlike in the propositional case, both logics are shown to have the same expressive power and their proof calculi to have the same deductive power.
Both calculi are also mutually relatively complete.

In the presence of differential equations, corollaries obtain usable and complete translations between differential game logic, a logic for the deductive verification of hybrid games, and the differential $\mu$-calculus, the modal $\mu$-calculus for hybrid systems.
The differential $\mu$-calculus is complete with respect to first-order fixpoint logic
and differential game logic is complete with respect to its ODE-free fragment.
\end{abstract}


\maketitle

\section{Introduction}

Modal $\mu$-calculus \cite{ScottBakker69,DBLP:journals/tcs/Kozen83} adds to propositional modal logic
 (with modalities as in the formula $\ddiamond{\stdtr}{\varphi}$)
the least fixpoint operator $\mu$, where $\lfp{X}{\varphi}$ is true in the smallest $X$ such that $X=\varphi(X)$.
Dually, the greatest fixpoint $\gfp{X}{\varphi}$ is the largest such $X$.
Modal $\mu$-calculus is notoriously hard to read but gives powerful engines, e.g., for CTL and CTL$^*$ \cite{DBLP:journals/tcs/Dam94}, because many computations correspond to alternating fixpoints.

Game logic \cite{DBLP:conf/focs/Parikh83} augments propositional dynamic logic PDL \cite{DBLP:journals/jcss/FischerL79,Segerberg77} with a duality operator $\pdual{}$ that switches between the players of a two player game, where formula \(\ddiamond{\gamma}{\varphi}\) means that the Angel player has a winning strategy in game $\gamma$ to make formula $\varphi$ true.
Game logic separates existence of winning strategies for games from strategy constructions as justifications in proofs.
Game logic formulas are easy to read because of their direct operational intuition of game play.

Three decades later, Parikh's problem has been solved: Game logic is less expressive than the modal $\mu$-calculus \cite{DBLP:journals/mst/BerwangerGL07}, because it embeds into the two-variable fragment of $\mu$-calculus whose variable hierarchy is strict.
Completeness of the axiomatization for game logic was shown recently \cite{DBLP:conf/lics/EnqvistHKMV19} based on cut-free completeness for the modal $\mu$-calculus~\cite{DBLP:conf/lics/AfshariL17}.

While these results about the propositional modal logic setting are exciting, this paper goes beyond the propositional case of abstract actions $a, b, c$ of unknown effect and considers first-order modalities with interpreted effects (such as assignments $\pupdate{\pumod{x}{\stdterm}}$ to object variables).
First-order modalities like in first-order dynamic logic are crucial for representing programs \cite{DBLP:conf/focs/Pratt76,Harel_1979,Harel_et_al_2000} and dynamical systems \cite{Platzer18}.

This paper shows that Parikh's problem has the opposite answer in the first-order case: first-order game logic and first-order modal $\mu$-calculus have the same expressive power, their calculi have the same deductive power and are mutually relatively complete.
Consequently the variable hierarchy of the modal $\mu$-calculus over first-order structures collapses at the second stage.
Beyond theoretical appeal, these results show that it is possible to have the best of both worlds, the readability and clear intuition of game logic
and the syntactical simplicity of the modal $\mu$-calculus.

The difficulty when proving properties about game logic are its more complicated game modalities.
First-order modal $\mu$-calculus in contrast
only has atomic modalities, which can be dealt with more easily.
An instance of this phenomenon is 
proving the completeness of a theory of the modal $\mu$-calculus,
relative to a fragment with fewer kinds of atomic modalities.
For the modal $\mu$-calculus this can be done schematically.
In view of the equi-expressivity, relative completeness for extensions
of first-order game logic can also be established schematically
by eliminating one atomic modality at a time, without dealing with games of complicated structure.

When extended with modalities for differential equations, the resulting differential $\mu$-calculus for hybrid systems is compared to differential game logic \dGL for hybrid games \cite{DBLP:journals/tocl/Platzer15} and shown to have the same expressive power and (their proof calculi to have) the same deductive power and are mutually relatively complete.
Using the expressive power of the first-order modal $\mu$-calculus,
an axiom is presented that equivalently characterizes differential equation modalities as a greatest fixpoint.
The schematic relative completeness result yields the equi-expressiveness and relative completeness of the differential $\mu$-calculus 
and its differential-equation-free fragment.
This carries over to differential game logic.

The contributions of this paper are threefold.
Firstly, a proof shows equi-expressivity of the modal $\mu$-calculus and game logic in the first-order case.
This translation is natural and useful 
for meta\-logical investigations of game logic on first-order structures and interpreted variants.
Secondly, translations between sound proof calculi for the first-order modal $\mu$-calculus and first-order game logic are shown to respect provability.
Thirdly, two interpreted variants of first-order game logic and first-order modal $\mu$-calculus for modeling and reasoning about hybrid dynamics with differential equations are considered: differential game logic and differential $\mu$-calculus.
Via the differential $\mu$-calculus and a characterization of continuous reachability
as a greatest fixpoint, the discrete relative completeness theorem for differential dynamic logic \cite{DBLP:conf/lics/Platzer12b} is extended to differential game logic.

\section{Preliminaries}
\label{sec:preliminaries}

This section recalls basic facts about fixpoints of monotone functions \cite{Harel_et_al_2000}.
A function $\Gamma:\Pc(X)\to\Pc(X)$ on the power set $\Pc(X)$ of $X$ is \emph{monotone}
if $\Gamma(A)\subseteq \Gamma(B)$ for sets $A\subseteq B$.
A set $A$ is called a \emph{pre-fixpoint} of $\Gamma$ iff $\Gamma(A)\subseteq A$
and a \emph{post-fixpoint} of $\Gamma$ iff $A\subseteq \Gamma(A)$.
If $\Gamma(A) = A$, then $A$ is a \emph{fixpoint} of $\Gamma$.
A pre-fixpoint $A$ is the \emph{least pre-fixpoint} of $\Gamma$ iff there is no proper subset of
$A$ which is a pre-fixpoint of $\Gamma$.
Dually a post-fixpoint $A$ is the \emph{greatest post-fixpoint} of $\Gamma$ iff there is no proper superset of $A$ which is a post-fixpoint of $\Gamma$.

If $\Gamma:\Pc(X) \to\Pc(X)$ is monotone, then
$\lfp{X}{\Gamma} = \capfold\{D\subseteq X : \Gamma(D)\subseteq D\}$
is the \emph{unique} least pre-fixpoint
and $\gfp{X}{\Gamma} =\cupfold\{D\subseteq X : D\subseteq \Gamma(D)\}$
the \emph{unique} greatest post-fixpoint.
Least pre-fixpoints and greatest post-fixpoint are fixpoints.

\section{Structures, Syntax and Semantics}
\label{sec:structures}

This paper is concerned with the modal $\mu$-calculus
and game logic, when interpreted on the usual first-order domains of quantification.
As usual a first-order signature refers to a sequence
of constant symbols, function symbols and predicate symbols.
Transition symbols $\Act$ are used for modalities. 

\begin{definition}
	A \emph{signature} $\Lc$ is a pair $(\Lcf,\Act)$ consisting of a first-order signature $\Lcf$ and
	a set of transition symbols~$\Act$.
\end{definition}

The logics in this paper involve two kinds of variables.
\emph{\Ovariables}, typeset in lowercase, denote objects in the domain as in first-order logic. Their values depend on the state, which determines the values of the \ovariables.
Once and for all fix an infinite set $\Xc$ of \ovariables.

\emph{\Pvariables}, typeset in uppercase, denote truth values. Their truth values depend on the state.
Thus, their truth value may change as the values of \ovariables changes.
For example, $x=1\lbisubjunct X$ may hold in
some state and some \assignment of the \pvariables.
\Pvariables will be used to form fixpoint constructs.
Once and for all fix a collection $\Vc$ of \pvariables.
An \ovariable in $\Xc$ is also just called \sovariable while elements of $\Vc$ are always referred to as \pvariables.

Write $\FOL_{\Lcf}$
for the collection of first-order formulas with equality, in the signature $\Lcf$ with \ovariables in the fixed set $\Xc$.
The set of literals $\FOLL_{\Lcf}$ consists of all atomic first-order formulas in $\FOL_{\Lcf}$ and their negations.

Given a first-order $\Lcf$-structure $\structA$, a \emph{state} $\stdstate$ is a function $\stdstate: \Xc\to\domain{\structA}$ where $\domain{\structA}$ is the domain of quantification of $\structA$.
The set of states is denoted $\Sc$ when the structure is clear from context.
Given a structure $\structA$ and a state $\stdstate\in \Sc$,
the semantics of an $\Lcf$-term $\stdterm$ is defined as their value $\stdstate\envelope{\stdterm}\in\domain{\structA}$.
Write $\structA,\stdstate\models \varphi$ when the $\Lcf$-formula $\varphi$ is true in $\structA$ 
and the state $\stdstate$.

It is frequently convenient to have a notation for the state obtained from another
state by modifying the value of a single \ovariable of that state.
The state $\modif{\stdstate}{x}{e}$ is the \emph{modification} of state $\stdstate$ at \ovariable $x\in\Xc$ to $e \in \domain{\structA}$ and coincides with state $\stdstate$ on $\Xc\setminus\{x\}$ except that \(\modif{\stdstate}{x}{e}(x)=e\).

\begin{definition}
	Let $\Lc= (\Lcf, \Act)$ be a signature.
	An \emph{$\Lc$-structure} $\structA$ is a first-order $\Lcf$-structure with domain $\domain{\structA}=A$ (interpreting constant symbols as objects, function symbols as functions, predicate symbols as relations)
	with an accessibility relation $\structA\envelope{\stdtr}\subseteq \Sc\times \Sc$ on states $\Sc$ for each $\stdtr\in \Act$.
\end{definition}
The relation $\structA\envelope{\stdtr}$ describes which final state $\newstate$ is reachable by $\stdtr$ from which initial state $\stdstate$, written \((\stdstate,\newstate)\in\structA\envelope{\stdtr}\). 

\subsection{Syntax and Semantics of First-Order Modal \mutex-Calculus}

Let $\Lc = (\Lcf,\Act)$ be a signature.
To avoid technicalities with polarities in $\mu$-calculus, every \pvariable $X\in\Vc$ is assumed to have a fresh \pvariable $\overline{X}\notin \Vc$ that will be used for logical complements.
Further $\overline{\overline{X}}$ is identified with $X$.
This identification is frequently used implicitly.
Let $\overline{\Vc} = \Vc\cup\{\overline{X} : X\in \Vc\}$.

\begin{definition} \label{def:Lmu-syntax}
The formulas of (the logic called) \emph{first-order modal $\mu$-calculus \LmuV} are defined by the grammar:
	\begin{IEEEeqnarray*}{rClr}
		\varphi &::=& \stdlit \mid X \mid \varphi_1\lor\varphi_2 \mid \varphi_1\land \varphi_2 \mid \ddiamond{\stdtr}{\varphi} \mid \dbox{\stdtr}{\varphi} \mid \lfp{X}{\psi} \mid \gfp{X}{\psi}
	\end{IEEEeqnarray*}
where\footnote{Only first-order \emph{literals} are allowed for $\stdlit$ to avoid ambiguous parsing.}
$\stdlit \in \FOLL_{\Lcf}$,
$X\in \overline\Vc$, $\stdtr\in \Act_\Lc$
and $\psi$ is an \Lmu-formula not mentioning $\overline{X}$ (to ensure the existence of fixpoints).
\end{definition}

The formula \(\dbox{\stdtr}{\varphi}\) expresses that $\varphi$ holds in all states reachable by $\stdtr$ from the current state.
Dually, \(\ddiamond{\stdtr}{\varphi}\) expresses that $\varphi$ holds in at least one such state.
The formula \(\lfp{X}{\varphi}\) is true in the states forming the least fixpoint $X$ of $X=\varphi(X)$.
Dually, \(\gfp{X}{\varphi}\) is true in the greatest fixpoint $X$ of $X=\varphi(X)$.

An occurrence of $X\in\overline\Vc$ is \emph{free} in $\varphi$ iff it does not occur in 
the scope of a least or greatest fixpoint operator binding~$X$.
Write \Lmu for the set of \LmuV{}-formulas without free \pvariables.

All modalities and fixpoint operators bind short.
For example, \(\ddiamond{\stdtr}{\varphi} \lor \psi\) is \((\ddiamond{\stdtr}{\varphi}) \lor \psi\) and \(\lfp{X}{\varphi} \lor \psi\) is \((\lfp{X}{\varphi}) \lor \psi\).
Note that barred \pvariables can be bound as well.
For example one may write $\lfp{\overline{X}}{(\stdlit\lor\overline{X})}$.
However  $\lfp{\overline{X}}{(\stdlit \lor {X})}$ is not a well-formed formula, 
since $X$ may not be mentioned in the scope of a fixpoint operator binding $\overline{X}$.

The semantics of an \LmuV-formula is the set of all states in $\Sc$ in which the formula is true.
This depends on the truth value assigned to the \pvariables,
which in turn depends on the state and changes in the scope of fixpoint operators.
To make this dependency explicit, the semantics 
uses a \emph{\assignment}
i.e.\ a function $\stdasst: \Vc \to \Pc(\Sc)$.
Given a \assignment $\stdasst:\Vc\to \Pc(\Sc)$, an 
\ovariable $X\in \overline{\Vc}$ and a set $E\subseteq \Sc$, the \emph{modified \assignment} $\modif{\stdasst}{X}{E}:\Vc\to\Pc(\Sc)$ 
is
	\[\modif{\stdasst}{X}{E}(Z) = \begin{cases}
		E&\text{if $Z = X, X\in \Vc$}\\
		\Sc\setminus\stdasst(E) &\text{if $Z = \overline{X}$, $X\in\overline\Vc\setminus\Vc$}\\
		\stdasst(Z) &\text{otherwise.}
	\end{cases}\]
Although $\stdasst$ is only defined for variables in $\Vc$, it is extended to a map on $\overline{\Vc}$ 
by ${\stdasst}(\overline{X})=\Sc\setminus\stdasst(X)$.

\begin{definition} \label{def:Lmu-semantics}
For an $\Lc$-structure $\structA$ and a \assignment $\stdasst:\Vc\to \Sc$
the \emph{denotational semantics} for \LmuV-formula $\varphi$ is recursively defined
as a subset $\stdasst\envelopemu{\varphi}\subseteq \Sc$ as follows:
\begin{enumerate}
	\item $\stdasst\envelopemu{\stdlit} =\{\stdstate \in \Sc : \structA,\stdstate \models \stdlit\}$ (for $\stdlit \in \FOLL_{\Lcf}$)
	\item $\stdasst\envelopemu{X} = \stdasst(X)$ (for $X\in \overline\Vc$)
	 \item $\stdasst\envelopemu{\varphi_1\lor\varphi_2} = \stdasst\envelopemu{\varphi_1} \cup \stdasst\envelopemu{\varphi_2}$
	 \item $\stdasst\envelopemu{\varphi_1\land\varphi_2} = \stdasst\envelopemu{\varphi_1} \cap \stdasst\envelopemu{\varphi_2}$
	 \item $\stdasst\envelopemu{\ddiamond{\stdtr}{\varphi}} =\{\stdstate\in \Sc {\with} \mexists{(\stdstate,\newstate)\in \structA\envelope{\stdtr}} \newstate\in \stdasst\envelopemu{\varphi}\}$
	 \item $\stdasst\envelopemu{\dbox{\stdtr}{\varphi}} = \{\stdstate\in \Sc {\with}\mforall{(\stdstate,\newstate)\in \structA\envelope{\stdtr}} \newstate\in \stdasst\envelopemu{\varphi}\}$
	 \item $\stdasst\envelopemu{\lfp{X}{\varphi}} = \capfold\{D\subseteq \Sc : \stdasst\tfrac{D}{X}\envelopemu{\varphi} \subseteq D\}$
	 \item $\stdasst\envelopemu{\gfp{X}{\varphi}}= \cupfold\{D\subseteq \Sc : D\subseteq \stdasst\tfrac{D}{X}\envelopemu{\varphi} \}$
\end{enumerate}
\end{definition}

\noindent
The semantics depends not only on \assignments $\stdasst$,
but also on the $\Lc$-structure $\structA$.
The synonymous notation $(\structA,\stdasst)\envelopemu{\varphi}$ is used when that dependency is important.
An \LmuV-formula is \emph{valid} iff $(\structA,\stdasst)\envelopemu{\varphi}=\Sc$ for all $\Lc$-structures $\structA$ and all \assignments $\stdasst$.
If $\varphi$ is an \Lmu-formula, i.e. it does not contain free \pvariables,
the semantics is independent of $\stdasst$ and written $\structA\envelopemu{\varphi}$.

Let $\varphi,\psi$ be \Lmu-formulas and $X\in\overline\Vc$ a \pvariable.
The formula $\subst[\varphi]{X}{\psi}$ obtained 
from $\varphi$ by replacing all free occurrences of $X$ with $\psi$ is defined as usual (\rref{app:substitutionpvars}).

Negation is \emph{not} part of the grammar of \Lmu, but negation $\overline{\varphi}$ is definable for any \Lmu formula $\varphi$ as an abbreviation:
\begin{IEEEeqnarray*}{rCllrCll}
	\overline{\stdlit} &\equiv& 
	\lnot \stdlit &(\stdlit\in\FOLL_{\Lcf})~
	& 
	\overline{\overline{X}} &\equiv& {X} \quad&(X\in\overline{\Vc})
	\\
	\overline{\varphi_1\lor\varphi_2} &\equiv& 
	\overline{\varphi_1}\land\overline{\varphi_2}
	&&
	\overline{\varphi_1\land\varphi_2} &\equiv& 
	\overline{\varphi_1}\lor\overline{\varphi_2}
	\\
	\overline{\ddiamond{\stdtr}{\varphi}} &\equiv& \dbox{\stdtr}{\overline{\varphi}}
	&&
	\overline{\dbox{\stdtr}{\varphi}} &\equiv& \ddiamond{\stdtr}{\overline{\varphi}}
	\\
	\overline{\lfp{X}{\varphi}} &\equiv& 
	\gfp{\overline X}{\overline{\varphi}}
	&&
	\overline{\gfp{X}{\varphi}} &\equiv& 
	\lfp{\overline X}{\overline{\varphi}}
\end{IEEEeqnarray*}
Note $\gfp{\overline X}{\overline{\varphi}}$ is still a well-formed \Lmu-formula: $X$ does not occur in $\overline{\varphi}$, because $\overline{X}$ does not occur in $\varphi$ since $X$ is bound in $\lfp{X}{\varphi}$.

Indeed, the semantics of $\overline{\varphi}$ corresponds to negation:%

\begin{theoremEnd}{proposition}\label{prop:negation}
	For any \LmuV-formula $\varphi$ and \assignment~$\stdasst$
	$$\stdasst\envelopemu{\overline{\varphi}}=\Sc\setminus {\stdasst}\envelopemu{{\varphi}}.$$
\end{theoremEnd}

\begin{proofEnd}
	The proof is by structural induction on the formula.

\begin{inparaitem}[\noindent- \emph{Case:}]
\item $\varphi \in \FOLL_\Lcf$: the statement is immediate by the definition of the semantics of first-order formulas.

\item $\varphi$ is a \pvariable $X\in \overline{\Vc}$ then $\stdasst\envelopemu{\overline{X}}=\stdasst(\overline{X}) = \Sc\setminus \stdasst(X)
	= \Sc\setminus\stdasst\envelopemu{X}$.

\item disjunction and conjunction are immediate by the induction hypothesis and De Morgan's laws for sets.

\item \(\ddiamond{\stdtr}{\varphi}\) of diamond modalities:
	\begin{IEEEeqnarray*}{rCl}
		\stdasst\envelopemu{\overline{\ddiamond{\stdtr}{\varphi}}}
		&=&\stdasst\envelopemu{{\dbox{\stdtr}{\overline{\varphi}}}}\\
		&=&\{\stdstate\in \Sc : \mforall{(\stdstate,\newstate)\in \structA\envelope{\stdtr}} \newstate\in \stdasst\envelopemu{\overline{\varphi}}\}\\
		&=&\{\stdstate\in \Sc : \mforall{(\stdstate,\newstate)\in \structA\envelope{\stdtr}} \newstate\in \Sc\setminus\stdasst\envelopemu{{\varphi}}\}\\
		&=&\Sc \setminus\{\stdstate\in \Sc : \mexists{(\stdstate,\newstate)\in \structA\envelope{\stdtr}} \newstate\in \stdasst\envelopemu{{\varphi}}\}\\
		&=& \Sc\setminus\stdasst\envelopemu{{\ddiamond{\stdtr}{\varphi}}}
	\end{IEEEeqnarray*}

\item analogously for box modalities.

\item \(\lfp{X}{\varphi}\) of the least fixpoint operator:
	\begin{IEEEeqnarray*}{rCl}
		\stdasst\envelopemu{\overline{\lfp{X}{ \varphi}}}
		&=&\stdasst\envelopemu{\gfp{\overline{X}}{\overline{\varphi}}}\\
		&=&\bigcup\{D\subseteq \Sc : D\subseteq\stdasst\tfrac{D}{\overline{X}}\envelopemu{\overline{\varphi}} \}\\
		&=&\bigcup\{D\subseteq \Sc : D\subseteq\Sc\setminus\stdasst\tfrac{\Sc \setminus D}{{X}}\envelopemu{{\varphi}} \}\\
		&=&\Sc\setminus\bigcap\{D\subseteq \Sc : \stdasst\tfrac{ D}{{X}}\envelopemu{{\varphi}} \subseteq D\}\\
		&=& \Sc\setminus\stdasst\envelopemu{{\lfp{X}{\varphi}}}
	\end{IEEEeqnarray*}
	The third equality is by definition of the modified \assignment and the induction hypothesis.

\item analogously for the $\nu$-operator.
\end{inparaitem}
\end{proofEnd}

As usual $\varphi_1 \limply \varphi_2$ is short for $\overline{\varphi_1}\lor\varphi_2$
and $\varphi_1 \lbisubjunct \varphi_2$ short for $(\varphi_1 \limply \varphi_2) \land (\varphi_2 \limply \varphi_1)$.
Also $\ltrue$ is \(p\lor\overline{p}\) and $\lfalse$ is \(p\land\overline{p}\).

Knaster-Tarski's fixpoint theorem \cite{DBLP:journals/pjm/Tarski55} guarantees the existence of fixpoints because any formula $\varphi$ that occurs in \(\lfp{X}{\varphi}\) or \(\gfp{X}{\varphi}\) has a semantics that is monotone in $X$.

\begin{theoremEnd}{lemma}[\LmuV monotonicity] \label{lem:monomu}
	The map $D\mapsto \modif{\stdasst}{X}{D}\envelopemu{\varphi}$ is monotone for all
	$X\in \overline{\Vc}$, all \assignments $\stdasst$
	and all \LmuV-formulas~$\varphi$ not mentioning $\overline{X}$.
\end{theoremEnd}
\begin{proofEnd}
	The proof is by structural induction on the formula $\varphi$.
	
\begin{inparaitem}[\noindent- \emph{Case:}]
\item if $\varphi$ is in $\FOLL_\Lcf$ or a \pvariable
	other than $X$ and $\overline{X}$,
	then the map is constant and monotonicity is trivial.

\item
	if $\varphi$ is $X$, the map is $D\mapsto \stdasst\tfrac{D}{X}\envelopemu{\varphi} = D$.

\item
	For disjunctions consider $D_1\subseteq D_2$ and note
	\begin{IEEEeqnarray*}{rCl}
		\stdasst\tfrac{D_1}{X}\envelopemu{\varphi_1\lor\varphi_2}
		&=& \stdasst\tfrac{D_1}{X}\envelopemu{\varphi_1}
		\cup \stdasst\tfrac{D_1}{X}\envelopemu{\varphi_2}\\
		&\subseteq&
		\stdasst\tfrac{D_2}{X}\envelopemu{\varphi_1}
		\cup \stdasst\tfrac{D_2}{X}\envelopemu{\varphi_2}\\
		&=& \stdasst\tfrac{D_2}{X}\envelopemu{\varphi_1\lor\varphi_2}.
	\end{IEEEeqnarray*}

\item Similarly for conjunctions.

\item
	For diamond modalities \(\ddiamond{\stdtr}{\varphi}\) consider $D_1\subseteq D_2$. Then
	\begin{IEEEeqnarray*}{rCl}
		\stdasst\tfrac{D_1}{X}\envelopemu{\ddiamond{\stdtr}{\varphi}}
		&=& \{\stdstate\in \Sc : \mexists{(\stdstate,\newstate)\in \structA\envelope{\stdtr}} \newstate\in \stdasst\tfrac{D_1}{X}\envelopemu{\varphi}\}\\
		&\subseteq& \{\stdstate\in \Sc : \mexists {(\stdstate,\newstate)\in \structA\envelope{\stdtr}} \newstate\in \stdasst\tfrac{D_2}{X}\envelopemu{\varphi}\}\\
		&=& \stdasst\tfrac{D_2}{X}\envelopemu{\ddiamond{\stdtr}{\varphi}}
	\end{IEEEeqnarray*}

\item Similarly for box modalities.

\item
	For least fixpoint formulas $\lfp{Y}{\varphi}$ there are
	two cases. If $Y\in\{X,\overline{X}\}$ the map is constant so there is nothing to show.

	If $Y\notin\{X,\overline{X}\}$,
	consider $D_1\subseteq D_2$ and note
	\begin{IEEEeqnarray*}{rCl}
		\stdasst\tfrac{D_1}{X}\envelopemu{\lfp{Y}{\varphi}}
		&=& \bigcap\{E\subseteq \Sc : \stdasst\tfrac{D_1}{X}\tfrac{E}{Y}\envelopemu{\varphi} \subseteq E\}\\
		&\subseteq& \bigcap\{E\subseteq \Sc : \stdasst\tfrac{D_2}{X}\tfrac{E}{Y}\envelopemu{\varphi} \subseteq E\}\\
		&=& \stdasst\tfrac{D_2}{X}\envelopemu{\lfp{Y}{\varphi}}
	\end{IEEEeqnarray*}
\item Analogously for greatest fixpoint formulas.
\end{inparaitem}
\end{proofEnd}

\rref{lem:monomu} crucially needed that
$\overline{X}$ does not occur in $\varphi$.
For example
$D\mapsto \modif{\stdasst}{X}{D}\envelopemu{\overline X}=\Sc\setminus D$
is not monotone.

\subsection{Syntax and Semantics of First-Order Game Logic}

\begin{definition} \label{def:GL-syntax}
Let $\Lc = (\Lcf,\Act)$ be a signature.
The formulas of \emph{first-order game logic} \GLV with \pvariables in $\Vc$ are defined by the grammar:
	\begin{IEEEeqnarray*}{rCl}
		\varphi &::=&  \stdlit \mid X \mid  \lnot \varphi \mid \varphi_1\lor \varphi_2 \mid \ddiamond{\gamma}{\varphi}\\
 		\gamma &::=& \stdtr \mid{} \ptest{\varphi} \mid \pchoice{\gamma_1}{\gamma_2} \mid \gamma_1;\gamma_2 \mid \prepeat{\gamma}\mid \pdual{\gamma}
	\end{IEEEeqnarray*}
where $\stdlit\in \FOLL_{\Lcf}$, $X\in \Vc$, and $\stdtr\in \Act$.
\(\GL\) is the fragment of first-order game logic
without \pvariables.
\end{definition}

In game logic negation is written $\lnot$.
To align with $\mu$-calculus $\overline{X}\in\overline\Vc$ is also used but syntactically identified with $\lnot X$.
Usually game logic is introduced without \pvariables.
The modification with \pvariables was introduced to facilitate the inductive translation between the first-order modal $\mu$-calculus
into first-order game logic.

The formula \(\ddiamond{\gamma}{\varphi}\) expresses that player Angel has a winning strategy in the game $\gamma$ to reach one of the states in which $\varphi$ is true.
The test game $\ptest{\varphi}$ is lost prematurely by Angel unless formula $\varphi$ is true in the current state.
The choice game $\pchoice{\gamma_1}{\gamma_2}$ allows Angel to choose between playing $\gamma_1$ or $\gamma_2$.
The sequential game $\gamma_1;\gamma_2$ plays $\gamma_2$ after $\gamma_1$ unless a player lost prematurely during the game $\gamma_1$.
The repetition game $\prepeat{\gamma}$ allows Angel to decide after each round of $\gamma$ whether she wants to stop or repeat.
The dual game $\pdual{\gamma}$ flips the roles of the players Angel and Demon by passing control with all choices and responsibilities of passing tests to the opponent.

Demonic choice \(\dchoice{\gamma_1}{\gamma_2}\) is definable by \(\pdual{(\dchoice{\pdual{\gamma_1}}{\pdual{\gamma_2}})}\).
Demonic repetition \(\drepeat{\gamma}\) is definable by
\(\pdual{(\prepeat{(\pdual{\gamma})})}\).
The formula \(\dbox{\gamma}{\varphi}\) expressing that player Demon has a winning strategy in the game $\gamma$ to reach one of the states satisfying $\varphi$ is definable as \(\ddiamond{\pdual{\gamma}}{\varphi}\).
Propositional connectives $\land,\limply,\lbisubjunct$ are definable.

\begin{definition} \label{def:GL-semantics}
For an $\Lc$-structure $\structA$ and \assignment $\stdasst: \Vc\to \Pc(\Sc)$ the \emph{denotation} of a \GLV-formula $\varphi$ is defined as a
subset \(\stdasst\envelopegl{\varphi}\subseteq \Sc\)
and the \emph{denotational semantics} of game $\gamma$ is defined (by simultaneous induction) as a function \(\stdasst\envelope{\gamma}:\Pc(\Sc)\to \Pc(\Sc)\).
Here \(\stdasst\envelopegl{\varphi}\) is the set of states in which $\varphi$ is true and $\stdasst\envelope{\gamma}(D)\subseteq\Sc$ is the set of states from which Angel has a winning strategy in game $\gamma$ to reach $D$.
The notation $\stdasst\envelope{\gamma}^D$ is synonymous with $\stdasst\envelope{\gamma}(D)$.

\begin{enumerate}
	\item $\stdasst\envelopegl{\stdlit} =\{\stdstate \in \Sc : \structA,\stdstate \models \stdlit\}$ (for $\stdlit \in \FOLL_{\Lcf}$)
	\item $\stdasst\envelopegl{X} = \stdasst(X)$ (for $X\in \Vc$)
	\item $\stdasst\envelopegl{\lnot \varphi} = \Sc\setminus{\stdasst}\envelopegl{\varphi}$ 
	\item $\stdasst\envelopegl{\varphi_1\lor\varphi_2} = \stdasst\envelopegl{\varphi_1} \cup \stdasst\envelopegl{\varphi_2}$
	\item $\stdasst\envelopegl{\ddiamond{\gamma}{\varphi}} =
	\stdasst\envelope{\gamma}(\stdasst\envelopegl{ \varphi})$
\end{enumerate}
For games $\gamma$ and $D\subseteq \Sc$ define:
\begin{enumerate}
	\item $\stdasst\envelope{\stdtr}^D = \{\stdstate\in \Sc : \mexists{(\stdstate,\newstate)\in \structA\envelope{\stdtr}} \newstate\in D\}$
	($\stdtr{\in}\Act$)
	\item $\stdasst\envelope{\ptest{\varphi}}^D = \stdasst\envelopegl{\varphi}\cap D$
	\item $\stdasst\envelope{\pchoice{\gamma_1}{\gamma_2}}^D =\stdasst\envelope{\gamma_1}^D \cup \stdasst\envelope{\gamma_2}^D$
	\item $\stdasst\envelope{\gamma_1;\gamma_2}^D = \stdasst\envelope{\gamma_1}(\stdasst\envelope{\gamma_2}^D)$
	\item $\stdasst\envelope{\prepeat{\gamma}}^D =\capfold\{Z\subseteq\Sc : D\cup \stdasst\envelope{\gamma}^Z \subseteq Z\}$
	\item $\stdasst\envelope{\pdual{\gamma}}^D = \Sc\setminus \stdasst\envelope{\gamma}^{\Sc\setminus D}.$
\end{enumerate}
\end{definition}

The synonymous notation $(\structA,\stdasst)\envelopegl{\varphi}$ is used when the dependency on the structure is important.
If $\varphi$ is a \GL-formula, %
the semantics is independent of $\stdasst$ and written $\structA\envelope{\varphi}\substgl$.
A \GLV-formula $\varphi$ is \emph{valid} iff $(\structA,\stdasst)\envelopegl{\varphi}=\Sc$ for all $\Lc$-structures $\structA$ and \assignments~$\stdasst$.

The demonic choice operator corresponds to intersection semantically: $\stdasst\envelope{\dchoice{\gamma_1}{\gamma_2}}^D = \stdasst\envelope{\gamma_1}^D
\cap\stdasst\envelope{\gamma_2}^D$.
The semantics of repetition \(\stdasst\envelope{\prepeat{\gamma}}^D\)
is the least fixpoint of the map \(Z\mapsto D\cup \stdasst\envelope{\gamma}^Z\).
The semantics of the definable \(\stdasst\envelope{{\gamma}^\demrep}^D\)
is the greatest fixpoint of the map 
\(Z\mapsto D \cap \stdasst\envelope{\gamma}^Z\).

The \GL counterpart to the monotonicity \rref{lem:monomu} is:

\begin{theoremEnd}{lemma}[\GL{} monotonicity]\label{lem:monotone}
		For any \GLV game $\gamma$ and 
		any \assignment $\stdasst$ the map $D\mapsto\stdasst\envelope{\gamma}^D$ is monotone.
\end{theoremEnd}
\begin{proofEnd}
	The proof is by structural induction on the game.

\begin{inparaitem}[\noindent- \emph{Case:}]
\item When $\gamma$ is an atomic transition $\stdtr$ or a test $?\varphi$, monotonicity is immediate.

\item For $\gamma_1\cup \gamma_2$ and $D_1\subseteq D_2$ by induction hypothesis
	\begin{IEEEeqnarray*}{rCl}
		\stdasst\envelope{\gamma_1\cup\gamma_2}^{D_1} &=&\stdasst\envelope{\gamma_1}^{D_1} \cup \stdasst\envelope{\gamma_2}^{D_1}\\
		&\subseteq&\stdasst\envelope{\gamma_1}^{D_2} \cup \stdasst\envelope{\gamma_2}^{D_2}\\
		&=&\stdasst\envelope{\gamma_1\cup \gamma_2}^{D_2}
	\end{IEEEeqnarray*}

\item
	For a sequential composition $\gamma_1;\gamma_2$
	consider
	$D_1\subseteq D_2$.
	The inductive hypothesis for $\gamma_2$ implies $\stdasst\envelope{\gamma_2}^{D_1}\subseteq\stdasst\envelope{\gamma_2}^{D_2}$.
	The inductive hypothesis for $\gamma_1$ implies
	\begin{IEEEeqnarray*}{rCl}
		\stdasst\envelope{\gamma_1;\gamma_2}^{D_1} &=&\stdasst\envelope{\gamma_1}(\stdasst\envelope{\gamma_2}^{D_1})
		\subseteq\stdasst\envelope{\gamma_1}(\stdasst\envelope{\gamma_2}^{D_2})\\
		&=&\stdasst\envelope{\gamma_1;\gamma_2}^{D_2}
	\end{IEEEeqnarray*}

\item
	For a game of the form $\gamma^*$ and $D_1\subseteq D_2$
	simply
	\begin{IEEEeqnarray*}{rCl}
		\stdasst\envelope{\gamma^*}^{D_1} &=&
		\bigcap\{Z\subseteq\Sc : D_1\cup \stdasst\envelope{\gamma}^Z \subseteq Z\}\\
		&\subseteq&
		\bigcap\{Z\subseteq\Sc : D_2\cup \stdasst\envelope{\gamma}^Z \subseteq Z\}\\&=&
		\stdasst\envelope{\gamma^*}^{D_2}
	\end{IEEEeqnarray*}

\item
	For the duality consider a game $\gamma^\dem$ and
	sets
	$D_1\subseteq D_2$.
	Then $\Sc\setminus D_2\subseteq \Sc\setminus D_1$ and
	by the inductive hypothesis
	$\stdasst\envelope{\gamma}^{\Sc\setminus D_2}\subseteq\stdasst\envelope{\gamma}^{\Sc\setminus D_1}$
	Hence
	$$\stdasst\envelope{\gamma^\dem}^{D_1} =\Sc\setminus
	\stdasst\envelope{\gamma}^{\Sc\setminus D_1} \subseteq
	\Sc\setminus
	\stdasst\envelope{\gamma}^{\Sc\setminus D_2}
	=\stdasst\envelope{\gamma^\dem}^{D_2}$$
	This completes the proof.
\end{inparaitem}
\end{proofEnd}

Originally \cite{DBLP:conf/focs/Parikh83} game logic was interpreted over neighborhood structures, which do not restrict the interpretation of atomic transitions $\stdtr\in \Act$ to transition relations but allow atomic games $\Pc(\Sc)\to \Pc(\Sc)$. 
For compatibility with \Lmu, \GL is interpreted here over
structures with atomic transitions.

\subsection{Deterministic Assignment and Quantifiers}
\label{sec:assignmentstructures}

This section shows two fundamental kinds of modalities differentiating
the first-order setting from the propositional one.
The first are deterministic assignments, which assign the value 
of a term to an \ovariable.
The second are nondeterministic assignments, which correspond to quantification
over \ovariables.
Those are introduced in the signature $\Lc$ instead of the logic for increased generality.

\emph{Deterministic assignment:}
A modality of the form \m{x := \stdterm} for a variable $x$ and a term $\stdterm$
is a \emph{deterministic assignment} modality.
A signature $\Lc=(\Lcf,\Act)$ \emph{has deterministic assignments}
iff $\Lcf$ contains at least two distinct
constant symbols $0,1$
and the modality $(x := \stdterm)\in \Act$
for every variable $x\in \Xc$ and every $\Lcf$-term $\stdterm$.
Given a signature $\Lc$ with deterministic assignments
an $\Lc$-structure $\structA$
is called an \emph{assignment structure} iff
the interpretations of these constant symbols are distinct $0^\structA \neq 1^\structA$
and deterministic assignments are interpreted as:
\[(x:=\stdterm)^\structA = \{(\stdstate,\stdstate\tfrac{\stdstate\envelope{\stdterm}}{x}) : \stdstate \in \Sc\}.\]

\emph{Nondeterministic assignment:}
A modality of the form \m{\prandom{x}} for a variable $x$ 
is a \emph{nondeterministic assignment} modality.
A signature $\Lc=(\Lcf,\Act)$ \emph{has nondeterministic assignments}
iff $(\prandom{x})\in \Act$
for every variable $x\in \Xc$.
Given a signature~$\Lc$ with nondeterministic assignments,
an $\Lc$-structure $\structA$ is a \emph{quantifier structure} 
iff nondeterministic assignments mean:
\[\structA\envelope{\prandom{x}} = \{(\stdstate,\stdstate\tfrac{a}{x}) : \stdstate \in \Sc , a\in A\}.\]

The usual first-order quantifiers naturally correspond to nondeterministic assignment modalities.
The quantifier $\lexists{x}{\varphi}$ corresponds to
$\ddiamond{\prandom{x}}{\varphi}$ and
$\lforall{x}{\varphi}$ to $\dbox{\prandom{x}}{\varphi}$.

For signatures containing deterministic and nondeterministic
assignments, the notion of validity is defined relative to 
assignment and quantifier structures respectively.
If $\Act$ contains no modalities apart from nondeterministic assignments,
the logic \Lmu is \emph{least fixpoint logic}  \cite{DBLP:journals/bsl/DawarG02}.

\emph{Renaming:}
For any state $\stdstate\in \Sc$ and \ovariables $x,y\in \Xc$ write $\transposestate{\stdstate}{x}{y}$
for the state which agrees with $\omega$ except that $\transposestate{\stdstate}{x}{y}(x) =\stdstate(y)$
and $\transposestate{\stdstate}{x}{y}(y) =\stdstate(x)$.
For renaming \ovariables in a formula
corresponding versions $\transposepv{X}{x}{y}$ of the \pvariables are added to the syntax of \GLV and \LmuV, which are always interpreted as
$\stdasst(\transposepv{X}{x}{y}) = \{\transposestate{\stdstate}{x}{y} : \stdstate\in \stdasst(X)\}$.
Renamed \pvariables can never be bound or substituted.
Similarly $\Act$ is assumed to be closed under renaming, i.e. for every $\stdtr\in\Act$ there is some $\transposeact{\stdtr}{x}{y}\in\Act$
such that
\[\structA\envelope{\transposeact{\stdtr}{x}{y}}=
\{(\transposestate{\stdstate}{x}{y},\transposestate{\newstate}{x}{y}) : (\stdstate,\newstate) \in \structA\envelope{\stdtr}\}.\]
By \rref{lem:transposesubst}{} the renaming extends to formulas
\(\transposefml{\varphi}{x}{y}\).

\section{Equi-Expressivity}
\label{sec:secequiexpressive}

\subsection{Embedding Game Logic Into Modal \mutex-Calculus}

Game logic embeds easily into the modal $\mu$-calculus based on the fact that the semantics of repetition games already is a least fixpoint.
Let $\Lc = (\Lcf, \Act)$ be any signature.
Define by structural induction on the \GLV-formula an embedding $\cdot^\sharp$ from \GLV-formulas to \LmuV[\Vc]-formulas as follows:
	\begin{IEEEeqnarray*}{rClCrClCrCl}
		\stdlit^\sharp &\equiv& \stdlit ~\text{for}~\stdlit\in\FOLL_{\Lcf}
		&\quad&
		X^\sharp &\equiv& X ~\text{for}~X\in\overline\Vc
		\\
		(\lnot \varphi)^\sharp  &\equiv& \overline{\varphi^\sharp}
		&\quad&
		(\varphi_1\lor \varphi_2)^\sharp  &\equiv& {\varphi_1^\sharp\lor \varphi_1^\sharp}
	\end{IEEEeqnarray*}
And for games:
{\allowdisplaybreaks
	\begin{IEEEeqnarray*}{rClCrCl}
		(\ddiamond{\stdtr}{\varphi})^\sharp  &\equiv& {\ddiamond{\stdtr}{\varphi}}^\sharp  & \text{$\stdtr\in \Act$}
		\\
		(\ddiamond{\ptest{\varphi_1}}{\varphi_2})^\sharp  &\equiv& \varphi_1^\sharp \land \varphi_2^\sharp
		\\
		(\ddiamond{\pchoice{\gamma_1}{\gamma_2}}{\varphi})^\sharp  &\equiv& (\ddiamond{\gamma_1}{\varphi})^\sharp \lor (\ddiamond{\gamma_2}{\varphi})^\sharp
		\\
		(\ddiamond{\gamma_1;\gamma_2}{\varphi})^\sharp  &\equiv&
		(\ddiamond{\gamma_1}{\ddiamond{\gamma_2}{\varphi}})^\sharp
		\\
		(\ddiamond{\pdual{\gamma}}{\varphi})^\sharp  &\equiv& \overline{(\ddiamond{\gamma}{\lnot \varphi})^\sharp}\\
		(\ddiamond{\prepeat{\gamma}}{\varphi})^\sharp  &\equiv&  \lfp{X}{(\varphi\lor \ddiamond{\gamma}{X})^\sharp} & \text{$X$ fresh.}
	\end{IEEEeqnarray*}
}%
This recursion is on a well-founded
order (see proof of \rref{prop:embeddingsharp}).
By induction it is easy to see that:
\begin{theoremEnd}{proposition}[Embedding]\label{prop:embeddingsharp}
	Suppose $\stdasst$ is a \assignment in $\structA$ and $\varphi$ is a \GLV-formula, then
		$$\stdasst \envelopegl{\varphi} = \stdasst\envelopemu{\varphi^\sharp}.$$
\end{theoremEnd}
\begin{proofEnd}
	First observe that the translation is 
	well-defined.
	Define a rank map $\rank$ from \GL-formulas
	and \GL-games to the natural numbers by induction
	on the recursive definition.
	\begin{IEEEeqnarray*}{rClCl}
		\rank(\stdlit) &=& 0 & \quad&\text{for $\stdlit\in \FOLL_{\Lcf}$} 
		\\
		\rank(X) &=& 0 &&\text{for $X\in \overline{\Vc}$} \\
		\rank(\lnot\varphi) &=& \rank(\varphi)+1
		\\
		\rank(\varphi_1\lor\varphi_2) &=&
		\rank(\varphi_1) + \rank(\varphi_2) +1
		\\
		\rank(\ddiamond{\gamma}{\varphi}) &=&
		\rank(\gamma) + \rank(\varphi)+1
	\end{IEEEeqnarray*}
	For \GL-games
	\begin{IEEEeqnarray*}{rClCl}
		\rank(\stdtr) &=& 0 & \quad&\text{for $\stdtr\in \Act$} 
		\\
		\rank(\ptest{\psi})
		&=&
		\rank(\psi)
		\\
		\rank(\pchoice{\gamma_1}{\gamma_2})
		&=&
		\max\{\rank(\gamma_1),\rank(\gamma_2)\}+1
		\\
		\rank({\gamma_1};{\gamma_2})
		&=&
		\rank(\gamma_1)+\rank(\gamma_2)+2
		\\
		\rank(\pdual{\gamma})
		&=&
		\rank(\gamma)+2
		\\
		\rank(\prepeat{\gamma})
		&=&
		\rank(\gamma)+2
	\end{IEEEeqnarray*}
	The recursive definition of ${}^\sharp$ is well-founded because 
	${}^\sharp$ is applied in the defining formula 
	only
	to formulas of rank strictly less than the defined formula.

	The claim proves by induction on rank
	$\rank$ and mostly by unraveling the definitions.
	For example for the case $\ddiamond{\prepeat{\gamma}}{\varphi}$:
	\begin{IEEEeqnarray*}{rCl}
		\stdasst\envelopemu{(\ddiamond{\prepeat{\gamma}}{\varphi})^\sharp}
		&=&
		\stdasst\envelopemu{\lfp{X}{(\varphi
		\lor\ddiamond{\gamma}{X})^\sharp}}
		\\
		&=&\capfold\{Z\subseteq\Sc : \stdasst\tfrac{Z}{X}\envelopemu{
			(\varphi
		\lor\ddiamond{\gamma}{X})^\sharp
		}\subseteq Z\}
		\\
		&=&\capfold\{Z\subseteq\Sc : \stdasst\tfrac{Z}{X}\envelopegl{
			\varphi
		\lor\ddiamond{\gamma}{X}
		}\subseteq Z\}
		\\
		&=&\capfold\{Z\subseteq\Sc : \stdasst\envelopegl{
			\varphi}\cup\stdasst\envelopegl{
		\gamma
		}^Z\subseteq Z\}
		\\
		&=&\stdasst\envelopegl{\ddiamond{\prepeat{\gamma}}{\varphi}}
	\end{IEEEeqnarray*}
	The cases $\neg \varphi_1$ and $\ddiamond{\pdual{\gamma}}{\varphi_1}$
	follow with \rref{prop:negation}.
\end{proofEnd}

When $\varphi$ is a \GL-formula, i.e.  has no \pvariables, then $\varphi^\sharp$ is an \Lmu-formula without \emph{free} \pvariables.
In fact, two distinct \pvariables suffice to construct $\varphi^\sharp$, by reusing the same \pvariables~\cite{pauly01}.

\subsection{Embedding Modal \mutex-Calculus Into Game Logic}
\newcommand{\dictsubst}[2][\vartheta]{#1[{#2}]}%

For this section fix an $\Lc$-assignment structure $\structA$.

In this section a converse translation $\varphi^\flat$ is defined from \Lmu-formulas into \GL, i.e. from \Lmu without free \pvariables
to game logic without \pvariables.
For the inductive proof of the 
correctness of the translation, a more general translation needs to be defined for arbitrary \LmuV-formulas
to \GLV-formulas.
The difficulty in translating from the $\mu$-calculus is dealing with bounded \pvariables inductively.
The construction needs to retain the information that a \pvariable should be thought of as bound in some larger (inaccessible) formula.

This is the job of a \emph{dictionary}, a map $\vartheta:\Vc\to \{0,1\}$.
If $\vartheta(X)=1$, occurrences of $X$ are translated as if they were bound, because their handling is controlled by the translation.
If $\vartheta(X)=0$, occurrences of $X$ are translated as if they were a free variable, regardless of whether $X$ is free or bound in $\varphi$.

If $\vartheta$ is a dictionary and $X\in \overline{\Vc}$, then $\dictsubst{X}$ is the dictionary similar to $\vartheta$ but considering $X$ (and $\overline{X}$) bound:
		\[\dictsubst{X}(Z) = \begin{cases}
			1 &\text{if $Z=X$ or $Z=\overline{X}$}\\
			\vartheta(Z) &\text{otherwise.}
		\end{cases}\]
Write
$\vartheta(\overline{X}) = \vartheta({X})$ for all $X\in {\Vc}$.

The bound \pvariables of $\varphi$ require careful attention
when translating from $\mu$-calculus into game logic.
Because initially those have not \emph{yet} been bound,
they are treated like free \pvariables to begin with.
A dictionary $\vartheta$ is called \emph{compatible} with an \LmuV-formula
$\varphi$ iff \(\vartheta(X)= 0\) for all $X\in \overline\Vc$ that are bound in $\varphi$.

First some assumptions are made that do not restrict the generality.
By renaming bound \pvariables,
it suffices to express in \GLV only those
\LmuV-formulas, in which every \pvariable is bound at most once.
Moreover one may fix finite sets $\Xc_0\subseteq \Xc$, $\Vc_0\subseteq\overline{\Vc}$
and $\Lambda\subseteq \Act$ and
consider %
only \LmuV-formulas with \ovariables from $\Xc_0$, \pvariables from $\Vc_0$
and atomic transition symbols from $\Lambda$.
For readability, assume\footnote{%
This is not essential as distinct \ovariables $x_X$
for each $x\in \Xc_0$ can be used to simulate those constant symbols with only the two constants $0,1$ in $\Lc$.
For example one may view $\setcval{X_i}$ as an abbreviation for the \GL-game:
\(
	x_{X_1} := 0; x_{X_2} := 0;\ldots; x_{X_n}:=0;~ x_{X_i}:=1
\)
and think of $\eqcval{X_i}$ as the \GL-formula
\(
	x_{X_1} = 0 \land \ldots\land x_{X_{i-1}} = 0
	\land x_{X_i}=1\land x_{X_{i+1}}=0\land\ldots \land x_{X_n}=0
\).
}
that there are distinct constant symbols $\ctrlval{X}$ in $\Lcf$ for each
\pvariable $X\in\Vc_0$.
Also pick some \ovariable $\ctrlvar{X}\in \Xc\setminus\Xc_0$ 
that is independent of every transition in $\Lambda$. 
Where a variable $y\in \Xc$ is \emph{independent} from a transition $\stdtr$
iff
$(\stdstate\tfrac{b}{y},\newstate\tfrac{b}{y})\in \structA\envelope{\stdtr}$
whenever $(\stdstate,\newstate)\in \structA\envelope{\stdtr}$.

Now for \Lmu-formulas $\varphi$ (as above) and a compatible dictionary~$\vartheta$, define
a \GLV-game $\varphi^\vartheta$ by recursion on the definition of $\varphi$ as follows:
If $\varphi$ is in $\FOLL_{\Lcf}$ or a \pvariable in~$\overline{\Vc}$ define
{\renewcommand{\flat}{\vartheta}%
	\begin{IEEEeqnarray*}{rClCrCl}
		\stdlit^\flat &\equiv& (\ptest{\stdlit});\pdual{(\ptest{\lfalse})} ~\text{for}~\stdlit \in \FOLL_{\Lcf}
		\\
		X^\flat &\equiv&
		\begin{cases} \ptest{X};\pdual{(\ptest{\lfalse})} &\text{if $\vartheta(X)= 0$ and $X\in \overline{\Vc}$} \\
			 \setcval{X} &\text{if $\vartheta(X)= 1$ and $X\in \overline{\Vc}$}
		\end{cases}
	\end{IEEEeqnarray*}
If $\varphi$ is a composite formula define
		\begin{IEEEeqnarray*}{rClCrCl}
		(\varphi_1\lor\varphi_2)^\flat &\equiv& (\pchoice{\varphi_1^\vartheta}{\varphi_2^\vartheta})
		&\quad\quad& 
		(\varphi_1\land\varphi_2)^\flat &\equiv& (\dchoice{\varphi_1^\vartheta}{\varphi_2^\vartheta})
		\\
		(\ddiamond{\stdtr}{\varphi})^\flat &\equiv& (\stdtr;\varphi^\flat)
		&\quad& 
		(\dbox{\stdtr}{\varphi})^\flat &\equiv& (\pdual{\stdtr};\varphi^\flat).
	\end{IEEEeqnarray*}
If $\varphi$ is a fixpoint formula define
	\begin{IEEEeqnarray*}{rCl}
	(\lfp{X}{\varphi})^\flat &\equiv& {\setcval{X}; \prepeat{(\testcval{X};\varphi^{\flat[X]})};\testcvalnot{X}}\\
	(\gfp{X}{\varphi})^\flat &\equiv& {\setcval{X}; \drepeat{(\pdual{(\testcval{X})};\varphi^{\flat[X]})};\pdual{(\testcvalnot{X})}}.
	\end{IEEEeqnarray*}
}%

The translation of fixpoint formulas is well-defined, because renaming ensured that no \pvariable is bound more than once.
Hence $\dictsubst{X}$ is compatible with~$\varphi$.
Whenever~$\varphi$ does not contain free \pvariables, $\varphi^\vartheta$ is a pure \GL-game, i.e. has no \pvariables.

Most of the translation is natural.
The game $\varphi^\vartheta$ can be thought of as a formal game semantics for the first-order $\mu$-calculus.
Angel tries to verify the formula $\varphi$ by winning the game.
Demon wins the game if the play of the game witnesses falsity of the formula $\varphi$.
The formula $\varphi$ is then true iff Angel has a winning strategy in $\varphi^\vartheta$.

This explains the translation of $X$, when treated like a free variable (i.e. when $\vartheta(X)=0$).
Angel wins if $X$ is true, because she passes her test $?X$ and the game ends because Demon loses his next move $\pdual{?\lfalse}$ as he failed to show falsity.
If $X$ is not true, Angel fails her test and the game ends.
It is important that $X$ is \emph{not} a fixpoint variable, bound in some larger formula containing $\varphi$.
Otherwise the fixpoint bound by $X$ may need to be unrolled again.
But since $X$ is a free variable, the play has encountered an unbound \pvariable and ends. Because Demon has not managed to provide
a play of the game witnessing falsity of the formula, Angel wins.

Most subtle is the translation of $X$ when it occurs in a formula $\varphi$ in scope of a fixpoint operator
$\lfp{X}{\varphi}$.
The translation $(\lfp{X}{\varphi})^\vartheta $
can be understood as a game in which Angel can play the game $\varphi^{\dictsubst{X}}$ arbitrarily often,
which unfolds the fixpoint map defined by~$\varphi$.
The control variable $\ctrlvar{}$ ensures that Angel wins only if she can keep winning the subgame until 
she can eventually force it to end on a true atomic formula which is not $X$.
The translation $X^\vartheta$ to $\setcval{X}$ records that at least one more unfolding of the fixpoint for $X$ is necessary before Angel has a chance to win.
The value of the control variable $\ctrlvar{}$ controls which nested fixpoint to unfold next.

Once the correctness of the translation is proven,
one can safely forget about the dictionary for translating \Lmu-formulas without free variables.
In that case any \pvariable $X$ will simply be translated to the assignment $\setcval{X}$.

The following substitution lemma relates 
translations with dictionary $\vartheta$
and translations with dictionary $\dictsubst{X}$.
\begin{theoremEnd}{lemma}[Substitution] \label{lem:substitution}
		Suppose
		$\stdasst$ is a \assignment in an assignment structure $\structA$,
		$\varphi$ an \LmuV-formula, $X\in \overline{\Vc}$
		and $\vartheta$ a dictionary.
		If $\varphi$ is compatible with $\dictsubst{X}$ 
		and does not mention $\overline{X}$
		then
		\[
		\stdasst \envelope{\varphi^{\dictsubst{X}}}^D =
		\modif{\stdasst}{X}{E}\envelope{\varphi^{\vartheta} }^D
		~~\text{where} ~E=\stdasst\envelope{\setcval{X}}^D
		\]
\end{theoremEnd}

\begin{proofEnd}
	By straightforward induction on the complexity of $\varphi$.
	
	\begin{inparaitem}[\noindent- \emph{Case:}]

	\item $\stdlit$ for $\stdlit \in \FOLL_{\Lcf}$.
	Immediate by the semantics.

	\item $Y\in\overline{\Vc}$ with $\vartheta(Y)=1$.
	Then 
	\begin{IEEEeqnarray*}{rCl}
		\stdasst \envelope{Y^{\dictsubst{X}}}^D 
		&=& \stdasst\envelope{\setcval{Y}}^D\\
		&=&\stdasst\tfrac{E}{X}\envelope{\setcval{Y}}^D
		=\stdasst\tfrac{E}{X}\envelope{Y^{\vartheta} }^D
	\end{IEEEeqnarray*}

	\item $Y\in\overline{\Vc}$ with $\vartheta(Y)=0$.
	By assumption $Y \neq \overline{X}$.
	First consider the case that $Y=X$.
	Then by definition of $E$:
	\begin{IEEEeqnarray*}{rCl}
		\stdasst \envelope{X^{\dictsubst{X}}}^D 
		&=&\stdasst \envelope{\setcval{X}}^D =E\qquad\text{and}\\
		\stdasst\tfrac{E}{X}\envelope{X^{\vartheta} }^D
		&=& \stdasst\tfrac{E}{X}\envelope{?X;(?\lfalse)^\dem }^D
		= \stdasst\tfrac{E}{X}\envelope{?X}^\Sc=E.
	\end{IEEEeqnarray*}

	Now consider the case $Y\not\in\{X,\overline{X}\}$.
	Then
	\begin{IEEEeqnarray*}{rCl}
		\stdasst\envelope{Y^{\dictsubst{X}} }^D
		&=& \stdasst\envelope{?Y;(?\lfalse)^\dem }^D
		= \stdasst\envelope{?Y}^\Sc=\stdasst\envelopegl{Y}\\
		&=&\stdasst\tfrac{E}{X}\envelope{?Y}^\Sc
		=\stdasst\tfrac{E}{X}\envelope{?Y;(?\lfalse)^\dem }^D
		\\ &=&\stdasst\tfrac{E}{X}\envelope{Y^{\vartheta} }^D.
	\end{IEEEeqnarray*}

	\item $\varphi_1\lor\varphi_2$, $\varphi_1\land\varphi_2$, $\ddiamond{\stdtr}{\varphi}$ or $\dbox{\stdtr}{\varphi}$. Straightforward by unraveling the definitions and the induction hypothesis.

	\item $\lfp{Y}{\varphi}$. Observe that $X\neq Y$ by compatibility.
	Unravel the definition of the fixpoint translation
		$\stdasst\envelope{({\lfp{Y}{\varphi}})^{\dictsubst{X}} }^D = \stdasst\envelope{\setcval{Y}}^W$
	where
		\[W = \stdasst\envelope{\prepeat{(\testcval{Y};\gamma^{\dictsubst{X,Y}})}}^Q,\]
	$Q= \stdasst\envelopegl{\neqcval{Y}} \cap D$.
	By definition
	$W$ is the least fixpoint of $U\mapsto Q\cup \stdasst\envelope{\testcval{Y};\varphi^{\dictsubst{X,Y}}}^U$.
	Equivalently, by induction hypothesis, $W$ is the least fixpoint of 
		\[\Gamma_1: U \mapsto Q\cup \stdasst\tfrac{\stdasst\envelope{\setcval{X}}^U}{X}\envelope{\testcval{Y};\varphi^{\dictsubst{Y}}}^U,\]
	Write $E = \stdasst\envelope{\setcval{X}}^D$ and define
	a second map
		\[\Gamma_2: U \mapsto Q\cup \stdasst\tfrac{E}{X}\envelope{\testcval{Y};\varphi^{\dictsubst{Y}}}^U.\]
	Next prove that $\Gamma_1$ and $\Gamma_2$ have the same fixpoints.
	Suppose first that $U$ is a fixpoint of $\Gamma_2$ then
		\begin{IEEEeqnarray*}{rCl}
			&&\stdasst\envelope{\setcval{X}}^U
			= \stdasst\envelope{\setcval{X}}^{\Gamma_2(U)}\\
		&=& \stdasst\envelope{\setcval{X}}^Q \cup\;
		\stdasst\envelope{\setcval{X}}(\stdasst\tfrac{E}{X}\envelope{\testcval{Y};\varphi^{\dictsubst{Y}}}^U)\\
		&=&\stdasst\envelope{\setcval{X}}^Q = \stdasst\envelope{\setcval{X}}^D = E.
		\end{IEEEeqnarray*}
	The third equality follows because $\ctrlval{X}\neq \ctrlval{Y}$ ensures that
	Angel can never pass the test $\testcval{Y}$ right
	after playing 
	$\setcval{X}$.
	The penultimate equality is immediate
	from the definition of $Q$ and $\ctrlval{X}\neq \ctrlval{Y}$.
	This shows that $\Gamma_1(U)=\Gamma_2(U)=U$, i.e. that
	$U$ is also a fixpoint of $\Gamma_1$.
	
	Almost identically it follows that
	 $\stdasst\envelope{\setcval{X}}^U = E$
	for any fixpoint $U$ of $\Gamma_1$.
	Analogously conclude
	that $U=\Gamma_1(U)=\Gamma_2(U)$ for fixpoints $U$ of $\Gamma_1$.

	Thus $W$, the least fixpoint of $\Gamma_1$,
	is the least fixpoint of $\Gamma_2$.
	By definition of the semantics of $\prepeat{}$
	  $$W =\stdasst\tfrac{E}{X}\envelope{
		\prepeat{(\testcval{Y};\varphi^{\dictsubst{Y}})}
	}^Q$$
	Finally because $Q=\stdasst\tfrac{E}{X}\envelopegl{\neqcval{Y}}\cap D$ and $\stdasst\envelope{\setcval{Y}}
	=\stdasst\tfrac{E}{X}\envelope{\setcval{Y}}$
	it follows that 
	\begin{IEEEeqnarray*}{rCl}
		&&\stdasst\envelope{({\lfp{Y}{\varphi}})^{\dictsubst{X}} }^D
		= \stdasst\envelope{\setcval{Y}}^W\\
		&=&\stdasst\tfrac{E}{X}\envelope{\setcval{Y}}(
			\stdasst\tfrac{E}{X}\envelope{
				\prepeat{(\testcval{Y};\varphi^{\dictsubst{Y}})}
			}^Q
		)
		\\
		&=&\stdasst\tfrac{E}{X}\envelope{({\lfp{Y}{\varphi}})^{\vartheta} }^D.
	\end{IEEEeqnarray*}

	\item $\gfp{Y}{\varphi}$. This is almost exactly like the case for the least fixpoint except that least fixpoint
	is replaced by greatest fixpoint everywhere and 
	$Q$ is defined $\stdasst\envelopegl{\eqcval{Y}}\cup D$.
	\end{inparaitem}
\end{proofEnd}

Recall that the fresh control variable $\ctrlvar{X}\notin \Xc_0$ does not occur anywhere in the formula $\varphi$ to translate and is moreover independent of all its transition symbols.
To show that the control variable does not have side effects it is necessary to assume that the initial \assignment $\stdasst$ is also independent of the value of $\ctrlvar{}$.
Formally, for the remainder of this section assume for all \assignments $\stdasst$ and $X\in \overline{\Vc}$ and all $a\in \domain{\structA}$:
$$\stdasst(X) = \{\omega : \omega\tfrac{a}{\ctrlvar{}}\in \stdasst(X)\}.$$

\begin{theoremEnd}{lemma}\label{lem:removectrlvar}
	Suppose $\structA$ is an assignment structure,
	$\stdasst$ is a \assignment in~$\structA$ and
	$\varphi$ an \LmuV-formula.
	Then for any $Z\in \overline{\Vc}$:
	\begin{enumerate}
		\item $\stdasst \envelopemu{\varphi} = \stdasst \envelopemu{\ddiamond{\setcval{Z}}{\varphi}}$ and
		\item $ \stdasst\envelope{\varphi^\vartheta}^D =
		\stdasst\envelope{\setcval{Z};\varphi^\vartheta}^D$ for
		any $D\subseteq\Sc$ and any compatible dictionary $\vartheta$.
	\end{enumerate}
\end{theoremEnd}

\begin{proofEnd}
	For the first part prove
	\[\stdasst \envelopemu{\varphi} = \{\omega: \omega\tfrac{a}{\ctrlvar{Z}}\in \stdasst \envelopemu{\varphi}\}\]
	for all $a\in \domain{\structA}$ by induction on the definition of $\varphi$.

	\begin{inparaitem}[\noindent -\emph{Case}: ]
		\item $\stdlit\in \FOLL_{\Lcf}$.
			Immediate because $\stdlit$ does not mention $\ctrlvar{X}$.

		\item $X\in \overline{\Vc}$.
			That
			\(\stdasst(X)
			=\{\omega: \omega\tfrac{b}{\ctrlvar{Z}}\in \stdasst(X)\}\)
			holds by the assumption on $\stdasst$.

		\item $\varphi_1\lor\varphi_2$. Simply by
			\begin{IEEEeqnarray*}{rCl}
				&&
				\{\omega: \omega\tfrac{b}{\ctrlvar{Z}}\in \stdasst\envelopemu{\varphi_1\lor\varphi_2}\}
				\\
				&=&
				\{\omega: \omega\tfrac{b}{\ctrlvar{Z}}\in \stdasst\envelopemu{\varphi_1}\}
				\cup
				\{\omega: \omega\tfrac{b}{\ctrlvar{Z}}\in \stdasst\envelopemu{\varphi_2}\}
				\\
				&=&
				\stdasst \envelopemu{{\varphi_1}}
				\cup
				\stdasst \envelopemu{{\varphi_2}}
				=
				\stdasst \envelopemu{{\varphi_1\lor\varphi_2}}.
			\end{IEEEeqnarray*}

		\item $\varphi_1\land\varphi_2$.
			Similar to the previous case.
	
		\item $\ddiamond{\stdtr}{\varphi}$ for $\stdtr \in \Act$.
			Suppose that $\stdstate\in \stdasst \envelopemu{\ddiamond{\stdtr}{\varphi}}$.
			Then there is $\newstate$ such that $(\stdstate,\newstate)\in \structA\envelope{\stdtr}$
			and $\newstate \in \stdasst\envelopemu{\varphi}$.
			Because $\stdtr$ is independent of $\ctrlvar{X}$ moreover
			$(\stdstate\tfrac{b}{\ctrlvar{}}, \newstate\tfrac{b}{\ctrlvar{}})\in \structA\envelope{\stdtr}$.
			By induction hypothesis $\newstate\tfrac{b}{\ctrlvar{}} \in \stdasst\envelopemu{\varphi}$
			so that $\omega\tfrac{b}{\ctrlvar{}}\in \stdasst \envelopemu{\ddiamond{\stdtr}{\varphi}}$.

			The reverse inclusion is symmetric.
	
		\item $\dbox{\stdtr}{\varphi}$.
			Similar to the previous case.

		\item $\lfp{X}{\varphi}$.
			Let $E=\stdasst \envelopemu{\lfp{X}{\varphi}}$.
			Because $E$ is a fixpoint of the map $D\mapsto\stdasst\tfrac{D}{X}\envelopemu{\varphi}$
			 the induction hypothesis implies:
			\[
			E = \stdasst\tfrac{E}{X}\envelopemu{\varphi}
			= 
			\{\omega: \omega\tfrac{b}{\ctrlvar{Z}}\in 
			\stdasst\tfrac{E}{X}\envelopemu{\varphi}\}	
			=
			\{\omega: \omega\tfrac{b}{\ctrlvar{Z}}\in 
			E\}
			\]
		
		\item $\gfp{X}{\varphi}$.
			Similar to the previous case.
	\end{inparaitem}

	For the second part proceed
	again by induction on $\varphi$

	\begin{inparaitem}[\noindent -\emph{Case}: ]
		\item $\stdlit\in \FOLL_{\Lcf}$.
			Immediate because $\stdlit$ does not mention $\ctrlvar{X}$.

		\item $X\in \overline{\Vc}$ and $\vartheta(X)=0$.
			Unravelling the translation and the semantics the
			statement becomes
			\(\stdasst(X)
			=\{\omega: \omega\tfrac{\ctrlval{Z}}{\ctrlvar{Z}}\in \stdasst(X)\},\)
			which holds by assumption on $\stdasst$.

		\item $X\in \overline{\Vc}$ and $\vartheta(X)=1$.
			The semantics of $\setcval{Z};\setcval{X}$
			and $\setcval{X}$ coincide.

		\item Conjunction and disjunction are straightforward.
	
		\item Modalities are just like modalities in the first part.

		\item $\lfp{X}{\varphi}$.
		Again because the semantics of $\setcval{Z};\setcval{X}$
		and $\setcval{X}$ coincide:
			\begin{IEEEeqnarray*}{rCl}
				&&
				\stdasst\envelope{\setcval{Z};(\lfp{X}{\varphi})^\vartheta}^D
				\\
				&=&
				\stdasst\envelope{\setcval{Z};{\setcval{X}; \prepeat{(\testcval{X};\varphi^{\vartheta[X]})};\testcvalnot{X}}}^D
				\\
				&=&
				\stdasst\envelope{{\setcval{X}; \prepeat{(\testcval{X};\varphi^{\vartheta[X]})};\testcvalnot{X}}}^D
				\\
				&=&
				\stdasst\envelope{(\lfp{X}{\varphi})^\vartheta}^D
			\end{IEEEeqnarray*}
		
		\item $\gfp{X}{\varphi}$.
			Similar to the previous case.
	\end{inparaitem}
\end{proofEnd}

\begin{theoremEnd}[see full proof]{proposition}[Counterembedding] \label{prop:translationmain}
	Suppose $\structA$ is an assignment structure,
	$\stdasst$ is a \assignment in~$\structA$,
	$\varphi$ an \LmuV-formula,
	$\vartheta$ a compatible dictionary and suppose $D\subseteq \Sc$ is such that
		\begin{equation}\label{conditiondagger}
			\mforall{X{\in}\overline{\Vc}} \big(\vartheta(X)=1 \;\Rightarrow\;\stdasst\envelope{\setcval{X}}^D = \stdasst(X)\big)
			\tag{$\ast$}
	 	 \end{equation}
	Then:
	\[
	\stdasst \envelopemu{\varphi} = \stdasst\envelope{\varphi^\vartheta}^D 
	\]
\end{theoremEnd}

The main case of interest when applying the proposition is where $\vartheta(X)=0$
for all variables ${X}\in \overline{\Vc}$. In that case condition \ref{conditiondagger} holds vacuously and is not a restriction.
\iflongversion
\else
\begin{proof}[Proof Sketch]
	The statement is proved by induction on the complexity of the formula $\varphi$
	simultaneously for all valuations $\stdasst$, all compatible dictionaries $\vartheta$ and all sets $D$ satisfying \ref{conditiondagger}.
	A full proof is in \rref{app:proofs}. The case of propositional variables and least fixpoint operators are sketched here:

	\begin{inparaitem}[\noindent- \emph{Case:}]
		\item $X$ for $X\in \overline{\Vc}$ with $\vartheta(X)=0$.
		\begin{IEEEeqnarray*}{rCl}
			\stdasst \envelopemu{X} &=& \stdasst(X) =
		\stdasst\envelope{?X}(\stdasst\envelope{(?\lfalse)^\dem}^D)=
	 \stdasst\envelope{X^\vartheta}^D.
		\end{IEEEeqnarray*}
	
		\item $X$ for $X\in \overline{\Vc}$ with $\vartheta(X)=1$.
			 $$\stdasst \envelopemu{X}
			  = \stdasst(X) \stackrel{\text{\ref{conditiondagger}}}{=} \stdasst\envelope{\setcval{X}}^D
			  = \stdasst\envelope{X^\vartheta}^D.$$

		\item $\lfp{X}{\varphi}$. By applying \rref{lem:substitution}
		the two fixpoint maps defining
		the semantics of repetition and least fixpoint are shown to have 
		the same fixpoints modulo the value of the control variable $\ctrlvar{}$.
		\rref{lem:removectrlvar} then takes care of $\ctrlvar{}$.
	\end{inparaitem}
\end{proof}
\fi
\begin{proofEnd}

	Prove the statement by induction on the complexity of formula $\varphi$ simultaneously for all \assignments $\stdasst$, compatible dictionaries~$\vartheta$ and all sets $D$.
	For the use of the induction hypothesis, note: If~$\vartheta$ is compatible with $\varphi$, then $\vartheta$ is compatible with all subformulas of $\varphi$.
	\begin{inparaitem}[\noindent- \emph{Case:}]

	\item $\stdlit$ for $\stdlit \in \FOLL_{\Lcf}$.
	Immediate by the semantics.

	\item $X$ for $X\in \overline{\Vc}$ with $\vartheta(X)=0$.
	\begin{IEEEeqnarray*}{rCl}
		\stdasst \envelopemu{X} &=& \stdasst(X) =
	\stdasst\envelope{?X}(\stdasst\envelope{(?\lfalse)^\dem}^D)=
 \stdasst\envelope{X^\vartheta}^D.
	\end{IEEEeqnarray*}

	\item $X$ for $X\in \overline{\Vc}$ with $\vartheta(X)=1$.
		 $$\stdasst \envelopemu{X}
		  = \stdasst(X) \stackrel{\text{\ref{conditiondagger}}}{=} \stdasst\envelope{\setcval{X}}^D
		  = \stdasst\envelope{X^\vartheta}^D.$$
	
	\item $\varphi_1\lor \varphi_2$.
	By the inductive hypothesis
		\begin{IEEEeqnarray*}{rCl}
			\stdasst\envelopemu{\varphi_1\lor\varphi_2}&=&
		\stdasst\envelopemu{\varphi_1}\cup 
		\stdasst\envelopemu{\varphi_2}\\ &=&
		\stdasst\envelope{\varphi_1^\vartheta}^D\cup 
		\stdasst\envelope{\varphi_2^\vartheta}^D =
		\stdasst\envelope{\varphi_1^\vartheta\cup\varphi_2^\vartheta}^D\\
		&=&
		\stdasst\envelope{(\varphi_1\lor\varphi_2)^\vartheta}^D
		\end{IEEEeqnarray*}

	\item $\varphi_1\land \varphi_2$.
		By the inductive hypothesis
			\begin{IEEEeqnarray*}{rCl}
			\stdasst\envelopemu{\varphi_1\land\varphi_2}&=&
			\stdasst\envelopemu{\varphi_1}\cap 
			\stdasst\envelopemu{\varphi_2}\\ &=&
			\stdasst\envelope{\varphi_1^\vartheta}^D\cap 
			\stdasst\envelope{\varphi_2^\vartheta}^D =
			\stdasst\envelope{(\varphi_1\land\varphi_2)^\vartheta}^D
			\end{IEEEeqnarray*}

	\item $\ddiamond{\stdtr}{\varphi}$.
	Reading definitions and the inductive hypothesis
	\begin{IEEEeqnarray*}{rCl}
		\stdasst\envelope{(\ddiamond{\stdtr}{\varphi})^\vartheta}^D
		&=& \stdasst\envelope{ \stdtr; \varphi^\vartheta}^D
		= \stdasst\envelope{\stdtr}(\stdasst\envelope{\varphi^\vartheta}^D)\\
		&=& \stdasst\envelope{\stdtr}(\stdasst\envelopemu{\varphi})\\
		&=& \{\stdstate\in \Sc : \mexists{(\stdstate,\newstate)\in \structA\envelope{\stdtr}} \newstate\in \stdasst\envelopemu{\varphi}\}\\ &=& \stdasst\envelopemu{\ddiamond{\stdtr}{\varphi}}
	\end{IEEEeqnarray*}

	\item $\dbox{\stdtr}{\varphi}$.
	Reading definitions and the inductive hypothesis
	\begin{IEEEeqnarray*}{rCl}
		\stdasst\envelope{(\dbox{\stdtr}{\varphi})^\vartheta}^D
		&=& \stdasst\envelope{ \pdual{\stdtr}; \varphi^\vartheta}^D
		= \Sc\setminus\stdasst\envelope{\stdtr}(\Sc\setminus\stdasst\envelope{\varphi^\vartheta}^D)\\
		&=&\Sc\setminus\stdasst\envelope{\stdtr}(\Sc\setminus\stdasst\envelopemu{\varphi})
		= \stdasst\envelopemu{\dbox{{\stdtr}}{\varphi}}
	\end{IEEEeqnarray*}

	\item $\lfp{X}{\varphi}$.
	Recall that $\varphi$ may not mention $\overline{X}$ by \rref{def:Lmu-syntax}.
	Let $E = \stdasst\envelope{(\lfp{X}{\varphi})^\vartheta}^D$
	so that after unraveling the definition of the translation
		$E = \stdasst\envelope{\setcval{X}}^W$
	for $W$ the least fixpoint of
		$$\Gamma(U) = Q\cup \stdasst\envelope{\testcval{X};\varphi^{\dictsubst{X}}}^U$$
	and $Q= \stdasst\envelopegl{\neqcval{X}}\cap D$.
	Also let $\Delta(U) = \stdasst\tfrac{U}{X}\envelopemu{\varphi}$.

	Firstly observe that for any set $R\subseteq\Sc$:
		\[\stdasst\envelope{\setcval{Z}}^{Q\cup R}=\begin{cases}
			\stdasst\envelope{\setcval{X}}^{R}
				& \text{if $Z=X$.}\\
			\stdasst\envelope{\setcval{X}}^{D\cup R}
				& \text{if $Z\neq X$.}
		\end{cases}
		\]

	Secondly note that $\modif{\stdasst}{X}{E}$ satisfies the independence assumption made prior to
	\rref{lem:removectrlvar}.
	For $Z\neq X$, independence of $\modif{\stdasst}{X}{E}$ follows from independence of $\stdasst$.
	And for any $a\in \domain{\structA}$:
		\begin{IEEEeqnarray*}{rCl}
			\modif{\stdasst}{X}{E}(X)&=&E =\stdasst\envelope{\setcval{X}}^W\\
			&=&\{\omega : \omega\tfrac{a}{\ctrlvar{}}\in W\}
		=\{\omega : \omega\tfrac{a}{\ctrlvar{}}\in E\}
		\end{IEEEeqnarray*}
		where the last line uses $\omega\tfrac{a}{\ctrlvar{}}\in E$
		iff $\omega\tfrac{a}{\ctrlvar{}}\in W$ since $E$ and $W$ are related by a change of $\ctrlvar{}$.

	Thirdly note that $\stdasst\tfrac{E}{X}$ satisfies condition \ref{conditiondagger} for $W$.
	Suppose $Z\in \overline{\Vc}$ with $\vartheta(Z)= 1$.
	Then $Z\neq X$,
	because $\vartheta$ is compatible with $\lfp{X}{\varphi}.$ Then
	\begin{IEEEeqnarray*}{rCl}
		\stdasst\tfrac{E}{X}\envelope{\setcval{Z}}^W
		&=&
		\stdasst\tfrac{E}{X}\envelope{\setcval{Z}}^{\Gamma(W)}\\
		&=&
		\stdasst\tfrac{E}{X}\envelope{\setcval{Z}}^D
		\stackrel{\text{\ref{conditiondagger}}}{=} \stdasst\tfrac{E}{X}(Z)
	\end{IEEEeqnarray*}
	The second equality is by the first observation above. 
	The last equality is by \ref{conditiondagger} for $D$.
	In particular, the assumption that the \assignment is independent of $\ctrlvar{}$ does not contradict \ref{conditiondagger}.

	The following computation shows that $E$ is a fixpoint 
	of~$\Delta$:
	\begin{IEEEeqnarray*}{rCl}
		E 
		&=&
		\stdasst\envelope{\setcval{X}}^W
		=
		\stdasst\envelope{\setcval{X}}^{\Gamma(W)}
		\\
		&=&
		\stdasst\envelope{\setcval{X}}(\stdasst\envelope{\testcval{X};\varphi^{\dictsubst{X}}}^W)
		\\
		&=&
		\stdasst\envelope{\setcval{X}}(\stdasst\envelope{\varphi^{\dictsubst{X}}}^W)
		\\
		&=&
		\stdasst\envelope{\setcval{X}}(\stdasst\tfrac{E}{X}
		\envelope{\varphi^{\vartheta}}^W)
		=
		\stdasst\tfrac{E}{X}
		\envelope{\setcval{X};\varphi^{\vartheta}}^W
		\\
		&=&
		\stdasst\tfrac{E}{X}
		\envelope{\varphi^{\vartheta}}^W
		=
		\stdasst\tfrac{E}{X}
		\envelopemu{\varphi}
		=\Delta(E)
	\end{IEEEeqnarray*}
	The first line uses that $W$ is a fixpoint of $\Gamma$.
	The second is by the first observation above with $R\equiv{}\stdasst\envelope{\testcval{X};\varphi^{\dictsubst{X}}}^W$.
	The third line uses that the semantics of 
	$\setcval{X};\testcval{X}$ and $\setcval{X}$ coincide.
	The fourth line is by \rref{lem:substitution},
	since $\varphi$ may not mention $\overline{X}$.
	The fifth line is \rref{lem:removectrlvar}
	and the induction hypothesis.
	Thus $E$ is a fixpoint of $\Delta$.

	To show that $E = \stdasst\envelopemu{\lfp{X}{\varphi}}$,
	it remains to show that $E$ is the least
	fixpoint of $\Delta$.
	Let $F$ be the least fixpoint of $\Delta$ and show $E\subseteq F$.
	By \rref{lem:removectrlvar} note
		\[\stdasst\envelope{\setcval{X}}^{F}
		= \stdasst\envelopemu{\ddiamond{\setcval{X}}{\lfp{X}{\varphi}}}
		= \stdasst\envelopemu{\lfp{X}{\varphi}} = F.\]
	For $U = Q\cup (\stdasst\envelopegl{\eqcval{X}}
	\cap F)$ this equality together with the observation above this implies
		\[\stdasst\envelope{\setcval{X}}^U =
		\stdasst\envelope{\setcval{X}}^F= F.\]
	By \rref{lem:substitution} this
	implies $\stdasst\tfrac{F}{X}\envelope{\varphi^\vartheta}^U
	=\stdasst\envelope{\varphi^{\dictsubst{X}}}^U$.
	And because $F$ is the least fixpoint of $\Delta$, this implies:
	\[F=\Delta(F)
	= \stdasst\tfrac{F}{X}\envelopemu{\varphi}
	=\stdasst\tfrac{F}{X}\envelope{\varphi^\vartheta}^U
	=\stdasst\envelope{\varphi^{\dictsubst{X}}}^U.
	\]
	The third equality uses the induction hypothesis.
	It remains to justify that $U$ satisfies \ref{conditiondagger}. Suppose $Z\in \overline{\Vc}$ with $\vartheta(Z)= 1$.
	Then by compatibility $Z \neq X$ and consequently
		\[\stdasst\envelope{\setcval{Z}}^U
		=\stdasst\envelope{\setcval{Z}}^Q
		=\stdasst\envelope{\setcval{Z}}^D
		\stackrel{\text{\ref{conditiondagger}}}{=}
		\Omega(Z),\]
	where the first equality uses that Angel can never win the game $\setcval{Z}$ into $\eqcval{X}$ as $Z \neq X$.
	The second equality uses that Angel wins win the game $\setcval{X}$ into $Q$ exactly if she wins into $D$, because $\neqcval{X}$
	is always true after playing.
	Thus, $U$ is a fixpoint of $\Gamma$ because
	\begin{IEEEeqnarray*}{rCl}
		\Gamma(U)
	&=& Q \cup \stdasst\envelope{\testcval{X}; \varphi^{\dictsubst{X}}}^U 
	\\
	&=&
	Q\cup (\stdasst\envelope{\eqcval{X}}\cap F)= U.
	\end{IEEEeqnarray*}
	Because $W$ is the least fixpoint of $\Gamma$,
	it follows that $W\subseteq U$.
	Consequently \rref{lem:monotone} shows the following as desired:
		$$E = \stdasst\envelope{\setcval{X}}^W
		\subseteq
		\stdasst\envelope{\setcval{X}}^U
		 =F$$

	\item $\gfp{X}{\varphi}$.
	This is almost exactly like the least fixpoint case
	with a few modifications.
	Let $E = \stdasst\envelope{(\gfp{X}{\varphi})^\vartheta}^D$
	so that after unraveling the definition of the translation
		$E = \stdasst\envelope{\setcval{X}}^W$
	where $W$ is the greatest fixpoint of
		$$\Gamma(U) = Q\cap \stdasst\envelope{\pdual{(\testcval{X})};\varphi^{\dictsubst{X}}}^U$$
	and $Q 
	=\stdasst\envelope{\eqcval{X}}\cup D$.
	Also let $\Delta(U) = \stdasst\tfrac{U}{X}\envelopemu{\varphi}$.

	The useful observation becomes
	\[\stdasst\envelope{\setcval{Z}}^{Q\cap R}=\begin{cases}
		\stdasst\envelope{\setcval{X}}^{R}
			& \text{if $Z=X$.}\\
		\stdasst\envelope{\setcval{X}}^{D\cap R}
			& \text{if $Z\neq X$.}
	\end{cases}
	\]
	for any set $R$.
	The argument proceeds almost identically 
	to the last case with least replaced by
	greatest everywhere.
\end{inparaitem}
\end{proofEnd}

For an \Lmu formula $\psi$ 
define $\psi^\flat$ to be the formula $\ddiamond{\psi^\vartheta}{\ltrue}$
where $\vartheta$ is the dictionary with $\vartheta(X)=0$ for all $X\in\Vc$.
If $\psi$ is an \Lmu formula,
then $\psi^\flat$ is a pure \GL-formula. %
The following semantic correspondence between
first-order game logic and the first-order modal $\mu$-calculus
follows from \rref{prop:translationmain}.

\begin{corollary}\label{cor:generalequi-expressivity}
	The first-order $\mu$-calculus and first-order game logic 
	are equi-expressive over any assignment
	structure $\structA$:
\begin{enumerate}
	\item  $\structA\envelopegl{\varphi} = \structA\envelopemu{\varphi^\sharp}$ for any \GL-formula  $\varphi$
	\item $\structA\envelopemu{\psi} = \structA\envelopegl{\psi^\flat}$ for any \LmuV-formula $\psi$
\end{enumerate}
\end{corollary}

On assignment structures,
$\mu$-calculus embeds into game logic, which embeds into the two-propositional-variable fragment of the modal $\mu$-calculus.
Thus, any formula of the modal $\mu$-calculus is equivalent over assignment structures to one with only two \pvariables.
This is in contrast to the propositional case,
whose variable hierarchy is strict \cite{DBLP:journals/mst/BerwangerGL07}.

Note that the use of the assignment modality is not necessary and nondeterministic assignments could be used as well.
The proof also works with the deterministic assignment $\setcval{X}$ replaced by $\prandom{\ctrlvar{}};\testcval{X}$.

\section{Proof Calculi}
\label{sec:proofcalculi}

This section introduces a Hilbert-style proof calculus for (the logic confusingly called) first-order modal $\mu$-calculus \LmuV
and first-order game logic \GL
with assignments.
In this section assume that $\Lc$ contains deterministic and nondeterministic assignments.

\subsection{Proof Calculus for First-Order Modal \mutex-Calculus}
\label{sec:proofcalcmu}

The proof calculus consists of the proof rules modus ponens, the atomic monotonicity rule, and the least fixpoint rule.
\begin{center}
	\renewcommand{\linferenceRuleNameSeparation}{\hspace{0.5\tabcolsep}}
	\begin{calculuscollection}
		\begin{calculus}
			\cinferenceRule[modusPonens|MP]{modus ponens}{
			\linferenceRule[sequent]{
				\psi\limply \varphi& \psi
			}{
				\varphi
			}}{}
		\end{calculus}\quad
		\begin{calculus}
			\cinferenceRule[atomicmonotonicity|${\text{M}}_\stdtr$]{monotonicity}{
				\linferenceRule[sequent]{
					\psi\limply \varphi
				}{
					\ddiamond{\stdtr}{\psi} \limply \ddiamond{\stdtr}{\varphi}
				}}{}
		\end{calculus}\quad
		\begin{calculus}
			\cinferenceRule[muI|FP$\mu$]{$\mu$ introduction proof rule}{
			\linferenceRule[sequent]{
			\subst[\psi]{X}{\varphi}\limply \varphi
			}{
				(\lfp{X}{\psi}) \limply \varphi
			}\quad}{\text{$X$ free for $\varphi$ in $\psi$}}
		\end{calculus}
	\end{calculuscollection}
\end{center}
The side condition (elaborated in \rref{app:substitutionpvars}) is necessary.

As axioms add all propositional tautologies and equality axioms.
The set of propositional tautologies is the smallest set of \LmuV-formulas
closed under substitution \(\varphi\mapsto \subst[\varphi]{X}{\psi}\)
which contains all formulas consisting only of \pvariables $X\in \overline{\Vc}$ and conjunctions that are valid when interpreted as a formula of propositional logic in 
the usual way (interpreting bars as negation).
Equality axioms are those of standard first-order logic characterizing equality as a congruence relation with respect to the function and predicate 
symbols in $\Lcf$.
Add the unfolding axiom for fixpoints:
\[
	\cinferenceRule[muf|$\mu$]{$\mu$-fixpoint axiom}
   {
   \linferenceRule[viuqe]
   {\subst[\varphi]{X}{\lfp{X}{\varphi}}}
   {\lfp{X}{\varphi}}\quad
   }{\text{$X$ free for $\lfp{X}{\varphi}$ in $\varphi$}}
\]
Renaming of bound \pvariables is permitted in a proof.
Such technicalities will be glossed over.

Additional axioms are required for interpreted modalities.
As least-fixpoint logic is at least as expressive as first-order logic,
the proof calculus should be made complete
for the first-order fragment. Adding the axioms
	\(
			\cinferenceRule[existsintroduction|$\exists$I]{existential quantifier introduction axiom}
		   {
		   \linferenceRule[impl]
		   {\subst[\varphi]{x}{\stdterm}}
		   {\ddiamond{\prandom{x}}{\varphi}}
		   }{}
	\)
		   and
	\(
			\cinferenceRule[vacuousAssignment|V]{vacous assignment}{
			\linferenceRule[impl]
			{\ddiamond{\prandom{x}}{\psi}}
			{\psi}}{}
	\)
if $x$ not free in $\psi$ achieves this by
G\"odel's completeness theorem \cite{Goedel30}.
Finally for deterministic assignments also add the axiom:\footnote{%
The axiom \(\ddiamond{x:=\stdterm}{\varphi}
\lbisubjunct \varphi\tfrac{\stdterm}{x}\)
handling assignments by (free) syntactic substitution suffices for any particular choice of atomic transitions.
It is not pursued here to avoid technicalities for substitution in atomic transitions.
}
\[
	\cinferenceRule[assignment|$\didia{:=}$]{assignment axiom}
   {
   \linferenceRule[viuqe]
   {\ddiamond{x:=\stdterm}{\varphi}}
   {\ddiamond{\prandom{y}}{(y=\stdterm \land \transposefml{\varphi}{x}{y}})}
   }{\text{$y$ not in $\varphi,\stdterm$}}
\]

Given a set of \LmuV-formulas $T$ 
and an \LmuV-formula~$\varphi$, write \(T\infers[\Lmucalc]\varphi\)
iff there is a proof of $\varphi$ in this calculus.
A \emph{proof} is a sequence of \LmuV-formulas such that
each formula is either an axiom, belongs to $T$
or follows from one of the preceding formulas by an application of one
of the three proof rules.
Since it is unclear whether provability is independent of the signature,
the language whose notion of provability is studied is carried
around as a subscript.

\begin{theoremEnd}{theorem}[\Lmucalc soundness]
	The $\Lmucalc$ proof calculus is sound.
	That is any $\LmuV$-formula $\varphi$ with $\infers[\Lmucalc] \varphi$ is valid.
\end{theoremEnd}

\begin{proofEnd}
	It suffices to prove that all axioms are sound,
	i.e. all their instances are valid
	and that all proof rules are sound, i.e. the conclusion
	is valid whenever all premises are valid.
	Consider a structure $\structA$ and a \assignment $\stdasst$.

\begin{inparaitem}[\noindent-]
\item
	The soundness of the propositional tautologies,
	equality axioms and the rule
	\irref{modusPonens} is as usual.
	Note that for propositional tautologies \rref{prop:negation}
	is used.

\item
	For \irref{atomicmonotonicity} suppose \(\psi \limply \varphi\) is valid.
	Then \(\stdasst\envelopemu{\psi}\subseteq\stdasst\envelopemu{\varphi}\).
	Thus,
	\(\stdasst\envelopemu{\ddiamond{\stdtr}{\psi}\limply\ddiamond{\stdtr}{\varphi}} = \Sc\)
	follows from
		\begin{IEEEeqnarray*}{rCl}
			\stdasst\envelopemu{\ddiamond{\stdtr}{\psi}}
			&=&
			\{\stdstate\in \Sc {\with} \mexists{(\stdstate,\newstate)\in \structA\envelope{\stdtr}} \newstate\in \stdasst\envelopemu{\psi}\}
			\\
			&\subseteq&
			\{\stdstate\in \Sc {\with} \mexists{(\stdstate,\newstate)\in \structA\envelope{\stdtr}} \newstate\in \stdasst\envelopemu{\varphi}\}
			\\
			&=&
			\stdasst\envelopemu{\ddiamond{\stdtr}{\varphi}}.
		\end{IEEEeqnarray*}

\item
	For \irref{muI} suppose $\subst[\psi]{X}{\varphi}\limply \varphi$ is valid.
	Then \(\stdasst\envelopemu{\subst[\psi]{X}{\varphi}}\subseteq\stdasst\envelopemu{\varphi}\).
	By \rref{lem:subsitutesemantically}
	\[\stdasst\tfrac{\stdasst\envelopemu{\varphi}}{X}\envelopemu{\psi}=
	\stdasst\envelopemu{\subst[\psi]{X}{\varphi}}\subseteq\stdasst\envelopemu{\varphi}.\]
	Hence:
	\begin{IEEEeqnarray*}{rCl}
		\stdasst\envelopemu{\lfp{X}{\psi}}
		&=&
		\capfold\{D\subseteq \Sc : \stdasst\tfrac{D}{X}\envelopemu{\psi} \subseteq D\}
		\subseteq\stdasst\envelopemu{\varphi}
	\end{IEEEeqnarray*}

\item
	For the unfolding axiom \irref{muf}
	note that $\stdasst\envelopemu{\lfp{X}{\varphi}}$
	is the least fixpoint of the monotone (\rref{lem:monomu}) map $D\mapsto\stdasst\tfrac{D}{X}\envelopemu{\varphi}$.
	Thus by \rref{lem:subsitutesemantically}
	 $$\stdasst\envelopemu{\lfp{X}{\varphi}}=\stdasst\tfrac{\stdasst\envelopemu{\lfp{X}{\varphi}}}{X}\envelopemu{\varphi}=\stdasst\envelopemu{
		{\subst[\varphi]{X}{\lfp{X}{\varphi}}}
	}.$$

\item Soundness of the axioms \irref{existsintroduction} and \irref{vacuousAssignment} are immediate.

\item For the assignment axiom \irref{assignment} note that $y$
not free in $\varphi$ means that $y$ does not occur in $\varphi$ and
all transitions $\stdtr$ in $\varphi$ are such that 
\[\structA\envelope{\stdtr}
=\{(\stdstate\tfrac{\lambda}{y},\newstate\tfrac{\lambda}{y}) : (\stdstate,\newstate)\in\structA\envelope{\stdtr}\}\]
for all $\lambda \in \domain{\structA}$.

It is sufficient to show that for any state $\stdstate$ there is some~$\lambda$ such that
	\[\stdstate\in \envelopemu{\varphi}\quad\text{iff}\quad \stdstate\tfrac{\lambda}{y}\in \envelopemu{y=\stdterm\land\transposefml{\varphi}{x}{y}}.\]
Choosing $\lambda = \stdstate\envelope{\stdterm}$
the right hand side is equivalent to
\(\stdstate\tfrac{\stdstate\envelope{\stdterm}}{y}\in\envelopemu{\transposefml{\varphi}{x}{y}},\)
because $y$ is not free in $\stdterm$.
By \rref{lem:transposesubst} this is equivalent to \(\stdstate\tfrac{\stdstate\envelope{\stdterm}}{x}\tfrac{\stdstate(x)}{y}\in \envelopemu{\varphi}.\)
Because $y$ is not free in $\varphi$ it follows that
\(\stdstate\tfrac{\stdstate\envelope{\stdterm}}{x}\in \envelopemu{\varphi}.\)
\end{inparaitem}
\end{proofEnd}

Substitution of \LmuV-formulas for \pvariables uniformly in an \Lmucalc derivation
is an admissible rule.

\begin{theoremEnd}{proposition}[\Lmucalc substitution]\label{prop:substituteoverproof}
	Let $\varphi,\psi$ be \LmuV-formulas.
	If \(\infers[\Lmucalc] \varphi\) and $X$ is free for $\psi$ in $\varphi$	 then \(\infers[\Lmucalc] \subst[\varphi]{X}{\psi}\).
\end{theoremEnd}
\begin{proofEnd}
By induction on the derivation witnessing \(\infers[\Lmucalc] \varphi\).

\begin{inparaitem}[\noindent-]
\item
	If $\varphi$ is an instance of an equality axiom
	it does not contain any \pvariables.

\item
	If $\varphi$ is a propositional tautology, then $\subst[\varphi]{X}{\psi}$
	is also a propositional tautology, since the set of propositional tautologies is closed under substitution.

\item
	If $\varphi$ is an instance of \irref{existsintroduction},
	\irref{vacuousAssignment} or
	\irref{assignment} then $\varphi\tfrac{\psi}{X}$
	is an instance of the same axiom.
	(For \irref{assignment} note that $\varphi_x^y\tfrac{\psi}{X}
	= (\varphi\tfrac{\psi}{X})_x^y$ )

\item
	If $\varphi$ is an instance of the axiom \irref{muf}, i.e. 
	of the form $\subst[\rho]{Y}{\lfp{Y}{\rho}}\lbisubjunct\lfp{Y}{\rho}$,
	where $Y$ is free for $\lfp{Y}{\rho}$ in $\rho$.
	There are two cases.
	If $Y\in \{X,\overline{X}\}$ then $\varphi$ and
	the substituted formula $\varphi\frac{\psi}{X}$
	are identical by \rref{lem:doublesubstitution}, so there is nothing to do.
	If $Y\notin \{X,\overline{X}\}$ then because $X$ is free for $\psi$ in 
	$\varphi$, the \pvariable $Y$ does not occur freely in $\psi$.
	By \rref{lem:doublesubstitution}
	the substituted formula has the form
		\[(\rho\tfrac{\psi}{X})\tfrac{\lfp{Y}{\rho\tfrac{\psi}{X}}}{Y} \lbisubjunct\lfp{Y}{(\rho\tfrac{\psi}{X})}\]
	and is provable as an instance of \irref{muf}.
	The side condition is satisfied since $Y$ is free for $\lfp{Y}{\rho}$
	in $\rho$ and $Y$ does not occur freely in $\psi$.

\item If $\varphi$ is the conclusion of an application of
	\irref{modusPonens}, i.e. there is $\psi$ such that 
	$\infers[\Lmucalc] \rho$ and $\infers[\Lmucalc] \rho\limply\varphi$.
	Since both derivations are shorter the inductive hypothesis yields
	$\infers[\Lmucalc] \rho\tfrac{\psi}{X}$ and $\infers[\Lmucalc] \rho\frac{\psi}{X}\limply\varphi\frac{\psi}{X}$.
	Thus $\infers[\Lmucalc] \varphi\frac{\psi}{X}$ derives by \irref{modusPonens}.

\item
	The case where $\varphi$ is the conclusion of an application of
	\irref{atomicmonotonicity} is similar.

\item
	Suppose $\varphi$ is the conclusion of an application of the rule \irref{muI}, so of the form $(\lfp{Y}{\rho})\limply \sigma$.
	Then $Y$ is free for $\sigma$ in $\rho$ and
	$\infers[\Lmucalc{}] \rho\tfrac{\sigma}{Y}\limply\sigma$.
	By the inductive hypothesis this implies $\infers[\Lmucalc] (\rho\tfrac{\sigma}{Y}\limply\sigma)\tfrac{\psi}{X}$.
	First assume that $Y\in \{X,\overline{X}\}$.
	Then by \rref{lem:doublesubstitution} also
	$$\infers[\Lmucalc] \rho\tfrac{\sigma\frac{\psi}{X}}{Y}\limply\sigma\tfrac{\psi}{X}.$$
	Because $Y$ is free for $\sigma$ in $\rho$ and $X$ is free for $\psi$ in $\varphi$ rule \irref{muI} is applicable and implies
	$\lfp{Y}{\varphi}\limply\sigma\tfrac{\psi}{X}$ as required.

	Now assume that $Y\notin \{X,\overline{X}\}$. 
	By renaming bound variables, 	
	assume without loss of generality
	that no free
	\pvariable of $\psi$ is bound in $\rho$.
	Because $X$ is free for
	$\psi$ in $\lfp{Y}{\rho}$, the \pvariables $Y$ and $\overline{Y}$
	are not free in $\rho$. By \rref{lem:doublesubstitution}{}
	then 
		\[\infers[\Lmucalc] (\rho\tfrac{\psi}{X})\tfrac{\sigma\tfrac{\psi}{X}}{Y} \limply \sigma\tfrac{\psi}{X}\]
	Because $X$ is free for $\psi$ in $\varphi$, %
	the \pvariable $Y$ does not occur freely in $\psi$.
	So since $Y$ is free for $\sigma$ in $\rho$ and 
	no free \pvariable of $\psi$ is bound in $\rho$, the \pvariable $Y$
	is free for $\sigma\tfrac{\psi}{X}$ in $\rho\tfrac{\psi}{X}$.
	By \irref{muI} it follows that
		\[\infers[\Lmucalc]\lfp{Y}{(\rho\tfrac{\psi}{X})} \limply\sigma\tfrac{\psi}{X}\qedhere\]
\end{inparaitem}
\end{proofEnd}

An important consequence is that free \pvariables
do not increase the deductive strength of the \Lmucalc-calculus:

\begin{theoremEnd}{corollary}
	If $\varphi$ is an \Lmu-formula with $\infers[\Lmucalc]\varphi$
	then there is a derivation of $\varphi$ consisting only of \Lmu-formulas.
\end{theoremEnd}
\begin{proofEnd}
	The proof is by induction on the length of the derivation of $\varphi$.
	All cases except for \irref{modusPonens} are immediate.
	Suppose that~$\varphi$ is the consequence of an application of \irref{modusPonens}.
	That is, there is some \LmuV-formula $\psi$ such that $\infers[\Lmucalc]\psi\limply\varphi$ and $\infers[\Lmucalc]\psi$.
	For readability assume that $X$ is the only free \pvariable in $\psi$.
	By \rref{prop:substituteoverproof} it follows that
	$\infers[\Lmucalc]\psi\tfrac{\lfalse}{X}\limply\varphi$ and $\infers[\Lmucalc]\psi\tfrac{\lfalse}{X}$.
	A close inspection of the proof of \rref{prop:substituteoverproof}
	reveals that those derivations are no longer than those of $\psi \limply\varphi$ and $\psi$ respectively.
	By the induction hypothesis, $\infers[\Lmucalc]\psi\tfrac{\lfalse}{X}\limply\varphi$ and $\infers[\Lmucalc]\psi\tfrac{\lfalse}{X}$
	are witnessed by derivations consisting only of \Lmu-formulas.
	Since $\varphi$ follows with \irref{modusPonens},
	there is a derivation of $\varphi$ without free \pvariables.
\end{proofEnd}

\subsection{Proof Calculus for First-Order Game Logic}
\label{sec:proofcalcgl}

This section introduces a similar Hilbert-style proof calculus for \GL-formulas.
It consists of the proof rule \irref{modusPonens} together with
the following monotonicity and fixpoint rules:
\begin{center}
	\begin{calculuscollection}
		\begin{calculus}
			\cinferenceRule[monotonicity|$\text{M}$]{monotonicity}{
				\linferenceRule[sequent]{
					\psi\limply \varphi
				}{
					\ddiamond{\gamma}{\psi} \limply \ddiamond{\gamma}{\varphi}
				}}{}
		\end{calculus}\qquad
		\begin{calculus}
			\cinferenceRule[fixpoint|FP*]{fixpoint rule}{
			\linferenceRule[sequent]{
				(\psi \lor\ddiamond{\gamma}{\varphi})\limply \varphi
			}{
				\ddiamond{\prepeat{\gamma}}{\psi} \limply \varphi
			}\quad}{}
		\end{calculus}
	\end{calculuscollection}
\end{center}
Here $\gamma$ ranges over all games, not only atomic transitions~$\stdtr$ as is the case for first-order $\mu$-calculus.
As axioms, add all propositional tautologies, equality axioms, the axioms \irref{existsintroduction} \irref{vacuousAssignment},\irref{assignment}
as well as the following
axiom schemata capturing the semantics of games:

\begin{calculus}
	\cinferenceRule[test|$\didia{?}$]{challenge axiom}
   {
   \linferenceRule[viuqe]
   {\ddiamond{\ptest{\psi}}{\varphi}}
   {(\psi\land\varphi)}
   }{}
   \cinferenceRule[choice|$\didia{{\cup}}$]{choice axiom}
   {
	\linferenceRule[viuqe]
   {\ddiamond{\pchoice{\gamma_1}{\gamma_2}}{\varphi}}
   {(\ddiamond{\gamma_1}{\varphi} \lor \ddiamond{\gamma_2}{\varphi} )}
   }{}
   \cinferenceRule[composition|$\didia{{;}}$]{composition axiom}
   {
	\linferenceRule[viuqe]
   {\ddiamond{\gamma_1;\gamma_2}{\varphi}}
   {\ddiamond{\gamma_1}{\ddiamond{\gamma_2}{\varphi}}}
   }{}
   \cinferenceRule[fixpointaxiom|$\didia{{}^*}$]{fixpoint axiom}
   {
	\linferenceRule[viuqe]
   {\ddiamond{\prepeat{\gamma}}{\varphi}}
   {(\varphi\lor\ddiamond{\gamma}{\ddiamond{\prepeat{\gamma}}{\varphi}})}
   }{}
   \cinferenceRule[demon|$ \didia{{^d}}$]{demon}
   {
	\linferenceRule[viuqe]
   {\ddiamond{\pdual{\gamma}}{\varphi}}
   {\lnot\ddiamond{\gamma}{\lnot\varphi}}
   }{}
\end{calculus}

In the axioms and proof rules, games range only over \GL{}-games and formulas 
only over \GL-formulas. Hence derivations consists of a sequence of \GL-formulas
without \pvariables.
Provability $\infers[\GLcalc]$ in \GL
is defined like provability $\infers[\Lmucalc]$ in \Lmu.

\begin{theoremEnd}{theorem}[\GL soundness]
	The \GL proof calculus is sound.
	That is, any \GL-formula $\varphi$ with \(\infers[\GLcalc] \varphi\) is valid.
\end{theoremEnd}

\begin{proofEnd}
It suffices to prove that all axioms and proof rules are sound.
Consider a structure $\structA$ and \assignment $\stdasst$.

\begin{inparaitem}[\noindent-]
\item
	The soundness of the propositional tautologies,
	equality axioms and the rule
	\irref{modusPonens} is straightforward.

\item
	For \irref{monotonicity} suppose $\psi \limply
	\varphi$ is valid.
	Then $\stdasst\envelopegl{\psi}\subseteq\stdasst\envelopegl{\varphi}$.
	Thus, $\Sc = \stdasst\envelopegl{\ddiamond{\gamma}{\psi}\limply\ddiamond{\gamma}{\varphi}}$ follows by \rref{lem:monotone}:
		\begin{IEEEeqnarray*}{rCl}
			\stdasst\envelopegl{\ddiamond{\gamma}{\psi}}
			&=&\stdasst\envelope{{\gamma}}(\stdasst\envelopegl{{\psi}})\\
			&\subseteq&
			\stdasst\envelope{{\gamma}}(\stdasst\envelopegl{{\varphi}})
			= \stdasst\envelopegl{\ddiamond{\gamma}{\varphi}}
		\end{IEEEeqnarray*}

\item
	For rule \irref{fixpoint} suppose $(\psi \lor\ddiamond{\gamma}{\varphi})\limply \varphi$ is valid.
	Then 
		$$\stdasst\envelopegl{\psi}\cup\stdasst\envelope{\gamma}(\stdasst\envelopegl{\varphi})
		\subseteq \stdasst\envelopegl{\varphi}.$$
	Hence
	\begin{IEEEeqnarray*}{rCl}
		\stdasst\envelopegl{\ddiamond{\prepeat{\gamma}}{\varphi}}
		&=&\capfold \{Z\subseteq\Sc : 
		\stdasst\envelopegl{\psi}
		\cup\stdasst\envelope{\gamma}^Z\subseteq Z\}\\
		&\subseteq& \stdasst\envelopegl{\varphi}.
	\end{IEEEeqnarray*}

\item
	For the axiom \irref{test}
	\begin{IEEEeqnarray*}{rCl}
		\stdasst\envelopegl{\ddiamond{\ptest{\psi}}{\varphi}}
		&=&\stdasst\envelope{\ptest{\psi}}(\stdasst\envelopegl{\varphi})
		\\
		&=&\stdasst\envelopegl{\psi}\cap
		\stdasst\envelopegl{{\varphi}}
		=\stdasst\envelopegl{{\psi\land \varphi}}.
	\end{IEEEeqnarray*}

\item
	For the axiom \irref{choice}
	\begin{IEEEeqnarray*}{rCl}
		\stdasst\envelopegl{\ddiamond{\pchoice{\gamma_1}{\gamma_2}}{\varphi}}
		&=&\stdasst\envelope{\pchoice{\gamma_1}{\gamma_2}}(\stdasst\envelopegl{\varphi})
		\\
		&=&\stdasst\envelope{{\gamma_1}}(\stdasst\envelopegl{\varphi})
		\cup 
		\stdasst\envelope{{\gamma_2}}(\stdasst\envelopegl{\varphi})
		\\
		&=& \stdasst\envelopegl{\ddiamond{{\gamma_1}}{\varphi}
		\lor \ddiamond{{\gamma_2}}{\varphi}}.
	\end{IEEEeqnarray*}

\item
	For the axiom \irref{composition}
	\begin{IEEEeqnarray*}{rCl}
		\stdasst\envelopegl{\ddiamond{{\gamma_1};{\gamma_2}}{\varphi}}
		&=&\stdasst\envelope{{\gamma_1};\gamma_2}(\stdasst\envelopegl{\varphi})
		\\
		&=&
		\stdasst\envelope{\gamma_1}(\stdasst\envelope{\gamma_2}\stdasst\envelopegl{\varphi})
		\\
		&=&\stdasst\envelope{\gamma_1}(\stdasst\envelopegl{\ddiamond{\gamma_2}{\varphi}})\\
		&=&
		\stdasst\envelopegl{\ddiamond{\gamma_1}{\ddiamond{\gamma_2}{\varphi}}}.
	\end{IEEEeqnarray*}

\item
	For the axiom \irref{fixpointaxiom}
	note that $\stdasst\envelope{\ddiamond{\prepeat{\gamma}}{\varphi}}$ is the least fixpoint of
	the map $Z\mapsto \stdasst\envelopegl{\varphi}
	\cup\stdasst\envelope{\gamma}^Z$.
	Hence a fixpoint:
		$$\stdasst\envelopegl{\varphi}
		\cup\stdasst\envelope{\gamma}(\stdasst\envelope{\ddiamond{\prepeat{\gamma}}{\varphi}}) = \stdasst\envelope{\ddiamond{\prepeat{\gamma}}{\varphi}}.$$

\item
	For the axiom \irref{demon}
	\begin{IEEEeqnarray*}{+rCl+rx*}
		\stdasst\envelopegl{\ddiamond{\pdual{\gamma}}{\varphi}}
		&=&\Sc \setminus\stdasst\envelope{{\gamma}}(\Sc\setminus\stdasst\envelopegl{\varphi})
		\\
		&=&
		\Sc \setminus\stdasst\envelope{{\gamma}}(\stdasst\envelopegl{\lnot\varphi})\\
		&=&
		\Sc \setminus\stdasst\envelopegl{\ddiamond{\gamma}{\lnot\varphi}}\\
		&=&\stdasst\envelopegl{\lnot\ddiamond{\gamma}{\lnot\varphi}}.
	\end{IEEEeqnarray*}

\item Soundness of the axioms \irref{existsintroduction}
	and	\irref{vacuousAssignment} is immediate.
	
\item Soundness of \irref{assignment}
	is just like for \Lmu. Use \rref{lem:substforreplgl}{} instead of \rref{lem:transposesubst}.
\end{inparaitem}
\end{proofEnd}

This proof calculus for first-order game logic is
essentially the proof calculus for propositional game logic~\cite{DBLP:conf/focs/Parikh83}.
The calculus for first-order $\mu$-calculus is weaker than the propositional one \cite{DBLP:journals/tcs/Kozen83}, which has
additional axioms for modalities.

\subsection{Relating the Proof Calculi}\label{sec:relatingcalculi}

Recall that by the assumption that $\Lc$ 
contains deterministic assignment modalities the translation ${}^\vartheta$ is defined.
The fresh control variable $\ctrlvar{}$ is assumed to be
independent
of all propositional variables and all
modalities not mentioning $\ctrlvar{}$ explicitly.
Formally assume that the proof calculi are extended such that for every $X\in\overline{\Vc}$ the equivalence
$\varphi \lbisubjunct \ddiamond{\setcval{X}}{\varphi}$
is provable in both calculi for every formula $\varphi$
not explicitly mentioning~$\ctrlvar{}$.
This is the syntactic analogue of \rref{lem:removectrlvar}.

Some of the proofs 
\iflongversion
\else
in \rref{app:proofs}
\fi
use auxiliary derived axioms and derived proof rules summarized in \rref{app:auxax}.

\begin{theoremEnd}{lemma}\label{lem:barflat}
	For any \GL-formula $\rho$, any \LmuV-formula~$\psi$ 
	and any compatible dictionary $\vartheta$ such that
	$\vartheta(X)=1$ for all $X$ which are free in $\psi$:
	\[\infers[\GLcalc] \ddiamond{\overline{\psi}^{\vartheta}}{\rho}\lbisubjunct  \ddiamond{\pdual{{\psi^{{\vartheta}}}}}{\rho}.\]
\end{theoremEnd}
\begin{proofEnd}
	An auxiliary definition is needed to prove the relation of the translation $\cdot^\vartheta$ of a formula and its barred version.
	Suppose $\psi$ is an \LmuV-formula and $\vartheta$ a compatible dictionary such 
	that $\vartheta(X)=1$ for all $X$ which are free in $\psi$.
	Write $\psi^{\overline{\vartheta}}$ for the formula $\psi^\vartheta$
	except that any occurrence of $\ctrlval{Y}$ is replaced by
	$\ctrlval{\overline{Y}}.$
	By induction on the definition
	the derivations for the following equivalence are sketched:
	\[\infers[\GLcalc] \ddiamond{\overline{\psi}^{\vartheta}}{\rho}\lbisubjunct  \ddiamond{\pdual{{\psi^{\overline{\vartheta}}}}}{\rho}.\]

	\begin{inparaitem}[\noindent- \emph{Case:}]
		\item $\stdlit\in \FOLL_{\Lcf}$: apply axioms
		\irref{demon+composition+test}.

		\item $X\in \Vc$: If $\vartheta(X)=0$ the derivation is like for literals.
		If $\vartheta(X)=1$, \irref{demon}, \irref{assignment} and quantifier reasoning prove:
			$$\infers[\GLcalc]\ddiamond{\setcval{\overline{X}}}{\rho}\lbisubjunct \ddiamond{{\pdual{(\setcval{\overline{X}})}}}{\rho},$$

		\item $\varphi_1 \land \varphi_2$: apply the \irref{choice} axiom.

		\item $\varphi_1 \lor \varphi_2$: apply the \irref{dchoice} axiom.
		
		\item $\dbox{\stdtr}{\varphi}$: apply axiom \irref{composition}
		and proof rule \irref{monotonicity}.

		\item $\ddiamond{\stdtr}{\varphi}$: apply axiom \irref{composition}
		and proof rule \irref{monotonicity}.
		Then apply the \irref{dcomposition} axiom.
		
		\item $\gfp{X}{\varphi}$:
		The following equivalences are provable in \GLcalc
		\begin{IEEEeqnarray*}{rCl}
			&&\ddiamond{\overline{\gfp{X}{\varphi}}^\vartheta}{\rho}\\
			&\text{iff}\quad&
			\ddiamond{\setcval{\overline{X}};\prepeat{(
			{\testcval{\overline{X}}}	
			;\overline{\varphi}^{\dictsubst{\overline{X}}})};\testcvalnot{\overline{X}}}{\rho}
			\\
			&\text{iff}\quad&
			\ddiamond{\setcval{\overline{X}};\prepeat{(
			{\testcval{\overline{X}}};
			\pdual{{\varphi^{\overline{\dictsubst{{X}}}}}}
			)};\testcvalnot{\overline{X}}}{\rho}
			\\
			&\text{iff}\quad&
			\ddiamond{\pdual{(
				\setcval{\overline{X}};{(
			\pdual{(\testcval{\overline{X}})}	
			;{\varphi}^{\overline{\dictsubst{X}}})}^\demrep;\pdual{(\testcvalnot{\overline{X}})})}}{\rho}
			\\
			&\text{iff}\quad&
			\ddiamond{\pdual{(\gfp{X}{
				{\varphi^{\overline{\vartheta}}}
				})}}{\rho}
		\end{IEEEeqnarray*}
		
		The second equivalence is by the induction hypothesis and by replacing equivalent
		programs in this context. (This follows by applying \irref{composition}
		to decompose and \rref{lem:replaceinloop} together with \irref{monotonicity}.)
		The third equivalence is by \irref{dcomposition} and \irref{assignment}.

		\item $\lfp{X}{\varphi}$: Like the previous case with the
		necessary additional applications of \irref{demon} and \irref{dcomposition}.
	\end{inparaitem}

	Because $\ctrlvar{X}$ is independent
	of modalities, literals and propositional
	variables the statement is derivable.
\end{proofEnd}

\begin{theoremEnd}{lemma}\label{lem:sharpdistribute}
	For \GLV-formulas $\varphi,\psi$
	$(\varphi\tfrac{\psi}{X})^\sharp\equiv\varphi^\sharp\tfrac{\psi^\sharp}{X}$
\end{theoremEnd}
\begin{proofEnd}
	This is proved by a straightforward induction on the rank introduced in the 
	proof of \rref{prop:embeddingsharp}.
	For the cases $\lnot\varphi$ and $\ddiamond{\pdual{\gamma}}{\varphi}$ use \rref{lem:negationsubstitution}.
\end{proofEnd}

\newcommand{\twostepmod}[2][\vartheta]{{#2}_{#1}}%
\newcommand{\ptwostepmod}[2][\vartheta]{{#2}_{#1}}%

\newcommand{\swostepmod}[3][X]{{#2}_{#1}}%

Without loss of generality assume that all dictionaries take 
the value $1$ only finitely often.
\begin{theoremEnd}{lemma}\label{lem:trueiffalse}
	Suppose $\psi$ is an \LmuV-formula and $\vartheta$ a compatible dictionary such that $\vartheta(X)=1$ for all
	free \pvariables in $\psi$. Then
	\[\infers[\GLcalc]
		\ddiamond{\psi^\vartheta}{\ltrue}\lbisubjunct
		\ddiamond{\psi^\vartheta}{(\vartheta(\ctrlvar{})=1)}\]
	where  $\vartheta(\ctrlvar{}) = 1$ is the disjunction of the
	formulas $\eqcval{X}$ for all~$X$ with $\vartheta(X)=1$.

	In particular 
	\(\infers[\GLcalc]
		\ddiamond{\psi^\vartheta}{\ltrue}\lbisubjunct
		\ddiamond{\psi^\vartheta}{\rho}\)
	for any \GL-formula $\rho$ whenever $\psi$ does not contain free \pvariables.
\end{theoremEnd}
\begin{proofEnd}
	The backward implication is immediate by \irref{monotonicity}.
	The forward implication is a straightforward induction.

	The in particular follows by \irref{monotonicity}
	 because $\vartheta(\ctrlvar{}) = 1$ is the empty disjunction, i.e. $\lfalse$ by compatibility.
\end{proofEnd}

Recall that for any \Lmu-formula $\psi$ the formula $\psi^\flat$ abbreviates $\ddiamond{\psi^\eta}{\ltrue}$
where $\eta$ is the dictionary $\eta(X)=0$ for all $X$.

\begin{theoremEnd}{proposition}\label{prop:rco}
	For a \GL{}-formula~$\varphi$ and an \Lmu-formula~$\psi$
		\[\infers[\GLcalc](\varphi^\sharp)^\flat\lbisubjunct \varphi \quad \text{and}\quad  \infers[\Lmucalc] (\psi^\flat)^\sharp\lbisubjunct \psi.\]
\end{theoremEnd}

\begin{proofEnd}
	First prove $\infers[\GLcalc](\varphi^\sharp)^\flat\lbisubjunct\varphi$.
	More generally prove by induction on the definition for all \GL-formulas $\varphi$ and all \GL-games~$\gamma$:	
		\[
			\infers[\GLcalc] {\varphi^\sharp}^\flat\leftrightarrow\varphi
			\quad\text{and}\quad
		\infers[\GLcalc] \ddiamond{{(\ddiamond{\gamma}{\sigma})^\sharp}^{\vartheta}}{\rho}
		\leftrightarrow \ddiamond{\gamma;{\sigma^\sharp}^\vartheta}{\rho}\]
	simultaneously for all \GL-formula $\rho$, 
	all \GLV-formulas~$\sigma$ and
	all compatible dictionaries $\vartheta$ such that $\vartheta(X)=1$ for every free \pvariable $X$ of $\sigma$.

	For formulas:

	\begin{inparaitem}[- \emph{Case}]
		\item $\stdlit$ for $\stdlit\in \FOLL_{\Lcf}$: simply by unraveling the definition of the 
		translations.

		\item $\lnot\varphi$: let $\eta$ be the dictionary taking value $0$ for every \pvariable.
		By \rref{lem:barflat}
		$(\varphi^\sharp)^\flat$ is the formula 
		$\ddiamond{\pdual{{(\varphi^\sharp)^{{\eta}}}}}{\ltrue}.$
		By \irref{demon} and \rref{lem:trueiffalse}
		this is provably equivalent to $\lnot\ddiamond{{{(\varphi^\sharp)^{{\eta}}}}}{\ltrue}$.
		The equivalence follows by the induction hypothesis.
		
		\item $\varphi_1\lor\varphi_2$: by \irref{choice} and the induction hypothesis.
		
		\item $\ddiamond{\gamma}{\varphi}$: apply the induction hypothesis for $\gamma$
		with $\sigma\equiv\varphi$, $\rho \equiv \ltrue$
		and $\vartheta=\eta$
		. This yields 
		\(\infers[\GLcalc] \ddiamond{{(\ddiamond{\gamma}{\varphi})^\sharp}^{\eta}}{\ltrue}
		\leftrightarrow \ddiamond{\gamma;{\varphi^\sharp}^\eta}{\ltrue}\).
		By axiom \irref{composition},
		\(\infers[\GLcalc] \ddiamond{{(\ddiamond{\gamma}{\varphi})^\sharp}^{\eta}}{\ltrue}
		\leftrightarrow \ddiamond{\gamma}{{\varphi^\sharp}^\flat}\).
		Since \({\varphi^\sharp}^\flat\) is 
		provably equivalent to \(\varphi\)
		by the induction hypothesis for $\varphi$
		the desired equivalence follows with~\irref{monotonicity}.
	\end{inparaitem}

	\begin{inparaitem}[- \emph{Case}]
		\item $\stdtr$ for $\stdtr\in \Act$: simply by unraveling the definition of the 
		translations.
		
		\item $\ptest{\varphi_1}$:
		by \irref{test} and \irref{composition} the right hand side of the equivalence is provably equivalent to $\varphi_1\land\ddiamond{{\sigma^\sharp}^\vartheta}{\rho}$.
		By unraveling the definitions the left hand side is provably equivalent to $\ddiamond{\dchoice{{\varphi_1^\sharp}^\vartheta}{{\sigma^\sharp}^\vartheta}}{\rho}$.
		Note that ${{\varphi_1^\sharp}^\vartheta}\equiv{{\varphi_1^\sharp}^\eta}$ because $\varphi_1$ does not contain \pvariables.
		Hence
		\(\infers[\GLcalc]
			\ddiamond{{{\varphi_1^\sharp}^\vartheta}}{\rho}
		\lbisubjunct
		\ddiamond{{{\varphi_1^\sharp}^\eta}}{\ltrue}\) by \rref{lem:trueiffalse}{}.
		By the induction hypothesis for $\varphi_1$ and \irref{dchoice} the
		desired statement follows.
		
		\item $\pchoice{\gamma_1}{\gamma_2}$: by \irref{choice} and the induction hypothesis.
		
		\item ${\gamma_1};{\gamma_2}$: by the induction hypothesis for $\gamma_1$ and \irref{composition} the formula
		\( \ddiamond{{(\ddiamond{\gamma_1}{\ddiamond{\gamma_2}{\sigma}})^\sharp}^{\vartheta}}{\rho}\)
		is \GL-provably equivalent to
		\[ \ddiamond{\gamma_1}{\ddiamond{{(\ddiamond{\gamma_2}{\sigma})^\sharp}^\vartheta}{\rho}}.\]
		By \irref{monotonicity} and inductive hypothesis for $\gamma_2$ this is equivalent to
		\[ \ddiamond{\gamma_1}{\ddiamond{{\gamma_2;({\sigma}^\sharp)}^\vartheta}{\rho}}.\]
		The statement follows with another application of \irref{composition}.

		\item $\prepeat{\gamma}$: The following equivalences are provable in \GLcalc
		\begin{IEEEeqnarray*}{rCl}
			&&
			\ddiamond{{(\ddiamond{\prepeat{\gamma}}{\sigma})^\sharp}^\vartheta}{\rho}
			\\
			&\text{iff}\quad&
			\ddiamond{\setcval{X};\prepeat{(\testcval{X};{(\sigma\lor \ddiamond{\gamma}{X})^\sharp}^{\dictsubst[\vartheta]{X}})};\testcvalnot{X}}{\rho}
			\\
			&\text{iff}\quad&
			\ddiamond{\setcval{X};\prepeat{(\testcval{X};{(\pchoice{{\sigma^\sharp}^{\dictsubst{X}{}}}{\gamma; {X^\sharp}^{\dictsubst{X}{}} )}})};\testcvalnot{X}}{\rho}
			\\
			&\text{iff}\quad&
			\ddiamond{\setcval{X};\prepeat{(\testcval{X};{(\pchoice{{\sigma^\sharp}^{\vartheta}}{{\gamma}; \setcval{X})}})};\testcvalnot{X}}{\rho}
			\\
			&\text{iff}\quad&
			\ddiamond{\prepeat{{\gamma}};{\sigma^\sharp}^{\vartheta}
			}{\rho}
		\end{IEEEeqnarray*}
		The second equivalence is by the induction hypothesis and \rref{lem:replaceinloop}{}.
		The third equivalence is because $X$ is some fresh \pvariable not occurring in $\sigma^\sharp$.
		The last equivalence is proved with the axioms for repetition and uses that if 
		the game ${\sigma^\sharp}^\vartheta$ ends,
		it ends with $\vartheta(\ctrlvar{X}) = 1$, i.e. 
		with $\neqcval{X}$. 

		\item $\pdual{\gamma}$: by \rref{lem:barflat} the induction hypothesis and \irref{dcomposition}.
	\end{inparaitem}

	\bigskip

	For the second equivalence an auxiliary definition is useful.
For any \LmuV-formula $\psi$ and any variable $X\in \overline{\Vc}$
define another \LmuV-formula~$\swostepmod{\psi}{X}$ by replacing free occurrences of any
\pvariable $Z$ with $\vartheta(Z)=1$ by $\ddiamond{\setcval{Z}}{X}$.
By induction prove
	$$ \infers[\Lmucalc] (\ddiamond{\psi^\vartheta}{X})^\sharp\lbisubjunct \swostepmod{\psi}{\rho}.$$
for any compatible dictionary $\vartheta$ and any $X\in \overline{\Vc}$.

\begin{inparaitem}[\noindent - \emph{Case}]
	\item $\stdlit$ for $\stdlit\in \FOLL_{\Lcf}$:
	immediate by \irref{test}, \irref{demon} and \irref{composition}.

	\item $Z$ for $Z\in \overline{\Vc}$: if $\vartheta(Z)=0$ 
	this is exactly the same as the previous case.
	If $\vartheta(Z)=1$ the antecedent
	and the consequent of the formula are both $\ddiamond{\setcval{Z}}{X}$.

	\item $\psi_1\lor\psi_2$ and $\psi_1\land\psi_2$: immediate by unraveling the definitions and
	the induction hypothesis.

	\item $\ddiamond{\stdtr}{\psi}$: immediate by the induction hypothesis
	and \irref{atomicmonotonicity}.

	\item  $\dbox{\stdtr}{\psi}$: immediate by the induction hypothesis
	and \irref{atomicmonotonicity}.

	\item $\lfp{Y}{\psi}$: each of the following equivalences 
	is provable in \Lmucalc
	\begin{IEEEeqnarray*}{rCl}
		&&(\ddiamond{(\lfp{Y}{\psi})^\vartheta}{X})^\sharp
		\\
		&\text{iff}\quad&
		\ddiamond{\setcval{Y}}{\lfp{Y}{(
			(\neqcval{Y}\land X)
			\lor (\eqcval{Y}\land (\ddiamond{\psi^{\dictsubst{Y}}}{Y})^\sharp)
		)}}
		\\
		&\text{iff}\quad&
		\ddiamond{\setcval{Y}}{\lfp{Y}{(
			(\neqcval{Y}\land X)
			\lor (\eqcval{Y}\land \swostepmod[Y]{\psi})
		)}}
		\\
		&\text{iff}\quad&
		\ddiamond{\setcval{Y}}{\lfp{Y}{(
			(\neqcval{Y}\land X)
			\lor (\eqcval{Y}\land \swostepmod[Y]{\psi}{Y}
			\tfrac{\ddiamond{\setcval{Y}}{Y}}{Y})
		)}}
		\\
		&\text{iff}\quad&
		\lfp{Y}{\swostepmod[X]{\psi}{}}
	\end{IEEEeqnarray*}
	The second equivalence is by the induction hypothesis and \irref{fixpointmonotonicity}.
	The last equivalence is proved with \irref{muf} and \irref{muI}.

	\item $\gfp{X}{\psi}$: similar to the previous case.

	For a formula $\psi$ without free \pvariables
	the statement follows with \rref{prop:substituteoverproof}{} by substituting $X$ with $\ltrue$
\end{inparaitem}
\end{proofEnd}

\begin{corollary}
	The first-order modal $\mu$-calculus is complete relative to first-order game logic over assignment structures and vice versa.
	That is any valid formula in one logic is provable
	from the translation of a valid formula into the other logic.
\end{corollary}

\begin{proof}
	Suppose $\varphi$ is a valid \GL-formula.
	Then $\varphi^\sharp$ is a valid \Lmu formula by
	\rref{prop:translationmain}.
	By \rref{prop:rco}
	$\infers[\GLcalc] (\varphi^\sharp)^\flat \limply\varphi$.
	Similarly for relative completeness of \Lmu to~\GL.
\end{proof}

\begin{theoremEnd}{proposition}[Equi-potency]\label{prop:deductivelyequivalent}
	For a \GL-formula $\varphi$ and an \Lmu-formula $\psi$ %
	\begin{enumerate}
		\item $\infers[\GLcalc] \varphi  \quad \text{iff} \quad  \infers[\Lmucalc]\varphi^\sharp$ \;and \label{point1}
		\item $ \infers[\Lmucalc] \psi \quad \text{iff} \quad \infers[\GLcalc] \psi^\flat$\label{point2}
	\end{enumerate}
\end{theoremEnd}
\begin{proofEnd}
	First prove the forward implications.
	The forward implication of \ref{point1} is proved by induction on the length of the derivation witnessing $\infers[\GLcalc]\varphi$.

\begin{inparaitem}[\noindent-]
\item
	Suppose first that $\varphi$ is an equality axiom, an instance of \irref{assignment} or a propositional tautology.
	In all of those cases case $\varphi^\sharp$ is of the same kind by the definition of~${}^\sharp$.

\item
	If $\varphi$ is one of the game axioms \irref{test}, \irref{choice}, \irref{composition}, \irref{demon} then
	$\varphi^\sharp$ is a propositional tautology,
	by the definition of ${\cdot}^\sharp$.
	For example if $\varphi$ is an instance of \irref{choice}
	then $\varphi^\sharp$ has the form
	\begin{IEEEeqnarray*}{rCl}
		\varphi^\sharp&\equiv&(\ddiamond{\pchoice{\gamma_1}{\gamma_2}}{\rho}\lbisubjunct (\ddiamond{{\gamma_2}}{\rho})\lor(\ddiamond{{\gamma_2}}{\rho}))^\sharp\\
		&\equiv&
		(\ddiamond{{\gamma_2}}{\rho})^\sharp\lor(\ddiamond{{\gamma_2}}{\rho})^\sharp\lbisubjunct (\ddiamond{{\gamma_2}}{\rho})^\sharp\lor(\ddiamond{{\gamma_2}}{\rho})^\sharp.
	\end{IEEEeqnarray*}

\item
	If $\varphi$ is an instance of \irref{fixpointaxiom} then by \rref{lem:sharpdistribute} 
	$\varphi^\sharp$ has the form
	\begin{IEEEeqnarray*}{rCl}
		\varphi^\sharp&\equiv&(\ddiamond{\prepeat{\gamma}}{\rho} \lbisubjunct (\rho\lor\ddiamond{\gamma}{\ddiamond{\prepeat{\gamma}}{\rho}}))^\sharp\\
		&\equiv&
		\ddiamond{\prepeat{\gamma}}{\rho}
		\lbisubjunct
		(\rho^\sharp\lor(\ddiamond{\gamma}{\ddiamond{\prepeat{\gamma}}{\rho}})^\sharp)\\
		&\equiv&
		\ddiamond{\prepeat{\gamma}}{\rho}
		\lbisubjunct
		(\rho^\sharp\lor(\ddiamond{\gamma}{{X}})^\sharp\tfrac{(\ddiamond{\prepeat{\gamma}}{\rho})^\sharp}{X})
	\end{IEEEeqnarray*}
	So $\varphi^\sharp$ is provable in \Lmucalc as it is an instance of axiom \irref{muf}.

\item
	Suppose the derivation ends with
	an application of \irref{modusPonens}.
	That is there is a formula $\rho$ such that
	$\infers[\GLcalc]\rho$ and $\infers[\GLcalc]\rho\limply\varphi$.
	By induction hypothesis $\infers[\Lmucalc]\rho^\sharp$ and $\infers[\Lmucalc](\rho\limply\varphi)^\sharp$.
	By definition of ${}^\sharp$ this means
	$\infers[\Lmucalc] \rho^\sharp\limply\varphi^\sharp$.
	Hence $\infers[\Lmucalc]\varphi^\sharp$ derives by \irref{modusPonens}.

\item
	Suppose the derivation ends with an application of \irref{fixpoint} so with a conclusion
	$\ddiamond{\prepeat{\gamma}}{\rho_1}\limply\rho_2$
	and a corresponding premise $\infers[\GLcalc](\rho_1\lor\ddiamond{\gamma}{\rho_2})\limply \rho_2$.
	By the induction hypothesis
		$$\infers[\Lmucalc] (\rho_1\lor\ddiamond{\gamma}{\rho_2})^\sharp\limply \rho_2^\sharp.$$
	Thus by \rref{lem:sharpdistribute} also
	$\infers[\Lmucalc](\rho_1\lor\ddiamond{\gamma}{X})^\sharp\tfrac{\rho_2^\sharp}{X}\limply\rho_2^\sharp$. An application of \irref{muI} then proves~$(\ddiamond{\prepeat{\gamma}}{\rho_1)^\sharp}\limply\rho_2^\sharp$
	as desired.

\item
	Suppose the derivation ends with \irref{monotonicity} so with a conclusion $\ddiamond{\gamma}{\rho_1}\limply\ddiamond{\gamma}{\rho_2}$ and a corresponding premise $\infers[\GLcalc]\rho_1\limply \rho_2$.
	By the induction hypothesis $\infers[\Lmucalc]\rho_1^\sharp\limply \rho_2^\sharp$.
	Prove $\infers[\Lmucalc](\ddiamond{\gamma}{\rho_1})^\sharp\limply(\ddiamond{\gamma}{\rho_2})^\sharp$ simultaneously for all $\rho_1$, $\rho_2$ 
	with $\infers[\Lmucalc]\rho_1^\sharp\limply \rho_2^\sharp$ by induction on the rank $\rank(\gamma)$ from the proof of \rref{prop:embeddingsharp}.
	\begin{inparaitem}[-]

		\item If $\gamma$ is atomic this follows immediately with \irref{atomicmonotonicity}.

		\item If $\gamma$ is a test game $\ptest{\zeta}$ propositional reasoning suffices.

		\item If $\gamma$ is $\gamma_1\cup \gamma_2$ the conclusion
		follows by applying the induction hypothesis separately to
		$\ddiamond{\gamma_i}{\rho_1}\limply\ddiamond{\gamma_i}{\rho_2}$
		for $i=1,2$ and using propositional reasoning.

		\item If $\gamma$ is $\gamma_1; \gamma_2$ 
		then $\infers[\Lmucalc](\ddiamond{\gamma_2}{\rho_1})^\sharp\limply(\ddiamond{\gamma_2}{\rho_2})^\sharp$ by induction hypothesis for
		$\gamma_2$
		and
		$\infers[\Lmucalc](\ddiamond{\gamma_1}{\ddiamond{\gamma_2}{\rho_1}})^\sharp\limply(\ddiamond{\gamma_1}{\ddiamond{\gamma_2}{\rho_2}})^\sharp$
		by induction hypothesis for $\gamma_2$.
		The claim is by definition of ${}^\sharp.$

		\item If $\gamma$ is $\prepeat{\gamma}$
		by the induction hypothesis and the assumption 
		$\infers[\Lmucalc] (\rho_1\lor\ddiamond{\gamma}{X})^\sharp\limply(\rho_2\lor\ddiamond{\gamma}{X})^\sharp$.
		By rule \irref{fixpointmonotonicity} from \rref{lem:derivedaxiomsmu}
		$\infers[\Lmucalc] \lfp{X}{(\rho_1\lor\ddiamond{\gamma}{X})^\sharp}\limply
		\lfp{X}{(\rho_2\lor\ddiamond{\gamma}{X})^\sharp}$
		derives.
 
		\item If $\gamma$ is $\pdual{\gamma}$ the conclusion follows
		from the induction hypothesis for $\gamma$ applied to obtain
		$\infers[\Lmucalc](\ddiamond{\gamma}{\lnot\rho_2})^\sharp\limply(\ddiamond{\gamma}{\lnot\rho_1})^\sharp$
		and propositional reasoning for double negation.
	\end{inparaitem}
\end{inparaitem}

\bigskip

Next prove the forward direction of \ref{point2}.
The proof is in three steps.
Let $\eta$ be the trivial dictionary with $\eta(X) = 0$ for all $X\in \Vc$.

Step 1: Note that $\psi\mapsto \ddiamond{\psi^{\eta}}\ltrue$ 
provably distributes over propositional connectives for \Lmu-formulas $\psi$.
For $\lor$ this is simply because
	$$\infers[\GLcalc]\ddiamond{\psi_1^{\eta}}{\ltrue}\lor \ddiamond{\psi_2^{\eta}}{\ltrue}\lbisubjunct
	\ddiamond{(\psi_1\lor \psi_2)^{\eta}}{\ltrue}$$
is an instance of \irref{choice}.
For bars
	$$\infers[\GLcalc] \ddiamond{\overline{\psi}^\eta}{\ltrue}\lbisubjunct\lnot\ddiamond{\psi^\eta}{\ltrue}$$
follows from \rref{lem:barflat}, \irref{demon} and \rref{lem:trueiffalse}{}.

Step 2: Prove that
	\[
	\infers[\GLcalc{}] 	\ddiamond{(\psi\tfrac{\sigma}{X})^\vartheta}{\rho} \lbisubjunct \ddiamond{\psi^{\vartheta[X]};
	(\testcval{X};\pchoice{\sigma^\eta}{\testcvalnot{X}})}{\rho}
	\]
for any \GL-formula $\rho$, any \LmuV-formula $\psi$,
any \Lmu-formula $\sigma$, 
and a suitably compatible dictionary $\vartheta$,
such that $(\psi\tfrac{\sigma}{X})^\vartheta$ is a \GL-formula.
This is by induction on $\psi$.
In particular by \rref{lem:trueiffalse}{}
\[
	\infers[\GLcalc{}] 	\ddiamond{(\psi\tfrac{\sigma}{X})^\eta}{\ltrue} \lbisubjunct \ddiamond{\psi^{\eta[X]};
	\sigma^\eta}{\ltrue}
	\]
when $\sigma$ is an \Lmu-formula and $\psi$ mentions only $X$ freely.

Step 3: Prove that $\infers[\mu] \psi$ 
implies $\infers[\GLcalc] \ddiamond{\psi^{\eta}}\ltrue$ for any
\Lmu-formula $\psi$.
The proof is by induction on the length of the derivation.

\begin{inparaitem}[\noindent-]

\item If $\psi$ is provable as a propositional tautology,
then $\ddiamond{\psi^\eta}{\ltrue}$ is also a propositional tautology and
therefore provable in \GL by Step 1.

\item Suppose $\psi$ is provable as an instance of an equality axiom.
	This follows from the fact that the translation distributes over propositional connectives and
	the fact that $\infers[\GLcalc] \stdlit\lbisubjunct\ddiamond{\stdlit^{\eta}}\ltrue$
	for any $\stdlit\in \FOLL_{\Lcf}$.

\item If $\psi$ is provable as an instance of the assignment axiom \irref{assignment}, then $\ddiamond{\psi^\eta}{\ltrue}$ is also an instance of the assignment axiom.

\item If $\psi\tfrac{\lfp{X}{\psi}}{X}\leftrightarrow\lfp{X}{\psi}$ is provable as an instance of \irref{muf},
	then each of the following equivalences is provable in \GL:
	\begin{IEEEeqnarray*}{rCl}
		&&
		\ddiamond{(\lfp{X}{\psi})^\eta}{\ltrue}
		\\
		&\text{iff}\quad&
		\ddiamond{\setcval{X};\prepeat{(\testcval{X};\psi^{\dictsubst[\eta]{X}})}
		\testcvalnot{X}}{\ltrue}
		\\
		&\text{iff}\quad&
		\ddiamond{\setcval{X}}{(
		\ddiamond{\testcvalnot{X}}{\ltrue}\lor{}\\&&
		\ddiamond{\testcval{X};\psi^{\dictsubst[\eta]{X}};\prepeat{(\testcval{X};\psi^{\dictsubst[\eta]{X}})}
		;\testcvalnot{X}}{\ltrue})}
		\\
		&\text{iff}\quad&
		\ddiamond{\setcval{X};
		\psi^{\dictsubst[\eta]{X}};\prepeat{(\testcval{X};\psi^{\dictsubst[\eta]{X}})}
		;\testcvalnot{X}
		}{\ltrue}
		\\
		&\text{iff}\quad&
		\ddiamond{\setcval{X};
		\psi^{\dictsubst[\eta]{X}};
		\setcval{X};
		\prepeat{(\testcval{X};\psi^{\dictsubst[\eta]{X}})}
		;\testcvalnot{X}
		}{\ltrue}
		\\
		&\text{iff}\quad&
		\ddiamond{\setcval{X};
		\psi^{\dictsubst[\eta]{X}};
		(\lfp{X}{\psi})^{\eta}
		}{\ltrue}
		\\
		&\text{iff}\quad&
		\ddiamond{(\psi\tfrac{\lfp{X}{\psi}}{X})^\eta)
		}{\ltrue}
	\end{IEEEeqnarray*}
The second equivalence is by \irref{fixpointaxiom} and \irref{monotonicity}.
The fourth equivalence is by \rref{lem:trueiffalse}{} and the last by Step 2.

\item
	If $\psi$ is the conclusion of an application of
	\irref{modusPonens} then $\ddiamond{\psi^{\eta}}{\ltrue}$ is provable by the induction hypothesis, \irref{modusPonens} and the fact that the translation distributes
	over implication.
	
\item If $\psi$ is the conclusion $\ddiamond{\stdtr}{\psi_1}\limply
	\ddiamond{\stdtr}{\psi_2}$ of 
	an instance of \irref{atomicmonotonicity} then 
	$\ddiamond{{\psi_1}^{\eta}}{\ltrue}\limply\ddiamond{{\psi_2}^{\eta}}{\ltrue}$
	is provable by the induction hypothesis.
	Now
	$$\ddiamond{a}{\ddiamond{{\psi_1}^{\eta}}{\ltrue}}\limply
	\ddiamond{a}{\ddiamond{{\psi_2}^{\eta}}{\ltrue}}$$
	derives by \irref{monotonicity}.
	The statement follows
	by axiom \irref{composition}. 

\item Suppose $\lfp{X}{\psi}\limply\sigma$
	does not contain any free \pvariables and is the conclusion of an
	application of \irref{muI}, i.e. $\psi\tfrac{\sigma}{X}\limply\sigma$ is \Lmucalc-provable.
	Then $\psi\tfrac{\sigma}{X}\limply\sigma$ also does not contain any
	free \pvariables and therefore by the induction hypothesis
		$$\ddiamond{(\psi\tfrac{\sigma}{X})^\eta}{\ltrue}\limply\ddiamond{\sigma^\eta}{\ltrue}$$
	is provable in \GLcalc.
	By Step 2 
	$$\ddiamond{\psi^{\dictsubst[\eta]{X}};
	\pchoice{(\testcval{X};\sigma^\eta)}{(\testcvalnot{X})}}{\ltrue}
	\limply\ddiamond{\sigma^\eta}{\ltrue}$$
	is provable in \GLcalc.
	By \irref{monotonicity} provability of 
	\begin{IEEEeqnarray*}{rCl}
		&&\ddiamond{\testcval{X};\psi^{\dictsubst[\eta]{X}};
	\pchoice{(\testcval{X};\sigma^\eta)}{(\testcvalnot{X})}}{\ltrue}\\
	&\limply&\ddiamond{\testcval{X};\sigma^\eta}{\ltrue}
	\end{IEEEeqnarray*}
	follows.
	By \irref{composition}, \irref{monotonicity}, \irref{choice} and \irref{fixpoint} moreover
	$$\ddiamond{\prepeat{(\testcval{X};\psi^{\dictsubst[\eta]{X}})}}{\ddiamond{{\testcvalnot{X}}}{\ltrue}}\limply \ddiamond{\testcval{X};\sigma^\eta}{\ltrue}$$
	Applying monotonicity again
	\begin{IEEEeqnarray*}{rCl}
		&&\ddiamond{\setcval{X};\prepeat{(\testcval{X};\psi^{\dictsubst[\eta]{X}})};\testcvalnot{X}}{\ltrue}\\
		&\limply& \ddiamond{\setcval{X};\testcval{X};\sigma^\eta}{\ltrue}
	\end{IEEEeqnarray*}
	Because $\ctrlvar{}$ does not affect $\sigma$ conclude
	$$\ddiamond{\setcval{X};\prepeat{(\testcval{X};\psi^{\dictsubst[\eta]{X}})};\testcvalnot{X}}{\ltrue}\limply \ddiamond{\sigma^\eta}{\ltrue}.$$
	By the fact that the translation distributes over propositional connectives
	provability of $\ddiamond{(\lfp{X}{\psi}\limply\sigma)^\eta}{\ltrue}$ follows.
\end{inparaitem}

	\bigskip

	For the backward direction of \ref{point1}
	suppose $\infers[\Lmucalc] \varphi^\sharp$.
	Then \ref{point2} implies $\infers[\GLcalc] (\varphi^\sharp)^\flat$.
	By \rref{prop:rco} and \irref{modusPonens} $\infers[\GLcalc]\varphi$ derives. The backward direction of \ref{point2} is similar.
\end{proofEnd}

This shows that the two logics are the same not only in expressive power,
but that the calculi prove exactly the same.

\subsection{Expressivity for First-Order Modal \mutex-Calculus}

The chief advantage of the first-order $\mu$-calculus over first-order game logic
is that \Lmu-formulas contain only atomic modalities.
This allows one to easily replace modalities in context, thereby locally reducing a formula to an equivalent one with fewer kinds of modalities.

\begin{definition}\label{def:provablyexpressiv}
	Let $\Lc =(\Lcf, \Act)$ be a signature and
	$T$ a collection of \Lmu-formulas.
	A syntactic fragment of \Lmu is
	\emph{$T$-provably expressive} (for \Lmu)
	iff for any \Lmu-formula $\varphi$ there is an \Lmu-formula $\rho$ in that fragment such 
	that $T \infers[\Lmucalc] \varphi \lbisubjunct \rho$.
\end{definition}

A local reduction result establishes provable expressivity:

\begin{proposition}[Local reduction] \label{prop:relativecompletenessscheme}
	Let $\Lc =(\Lcf, \Act)$ be a signature,
	$T$ a collection of \Lmu-formulas and
	$\Lambda \subseteq \Act$.
	Assume that for every $\stdtr\in \Lambda$ and every \LmuV-formula $\varphi$ without $\Lambda$-modalities,
	there is some \LmuV-formula $\hat\varphi$ without $\Lambda$-modalities and with no more free \pvariables than $\varphi$, such that
	$T \infers[\Lmucalc] \hat\varphi \lbisubjunct \ddiamond{\stdtr}{\varphi}.$
	Then the $\Lambda$-modality free fragment of \Lmu is $T$-provably expressive.
\end{proposition}

\begin{proof}
	Show more generally that for any \LmuV-formula $\varphi$ there is an
	\LmuV-formula $\rho$ without $\Lambda$-modalities such that
	$T \infers[\Lmucalc] \psi \lbisubjunct \ddiamond{\stdtr}{\varphi}$
	and $\rho$ has no more free \pvariables than $\varphi$.
	Proceed by induction on $\varphi$.
	For first-order literals or \pvariables the statement is clear.

\begin{inparaitem}[\noindent- \emph{Case:}]
\item \(\varphi_1\lor\varphi_2\):
	By the inductive hypothesis there are formulas $\rho_1$ and $\rho_2$
	without $\Lambda$-modalities
	such that $T \infers[\Lmucalc] \varphi_i \lbisubjunct \rho_i$.
	Propositionally derive
	$T \infers[\Lmucalc] (\varphi_1\lor\varphi_2) \lbisubjunct (\rho_1\lor\rho_2)$. So $\rho\equiv\rho_1\lor\rho_2$ is as required.
	Conjunction is similar.
	
\item \(\ddiamond{\stdtr}{\sigma}\):
	By the induction hypothesis pick a formula $\rho$
	without $\Lambda$-modalities and with no more \pvariables than $\sigma$ such that
	$\infers[\Lmucalc] \rho\lbisubjunct \sigma.$
	Rule \irref{atomicmonotonicity}
	derives  $\infers[\Lmucalc] \ddiamond{\stdtr}{\rho}\lbisubjunct\ddiamond{\stdtr}{\sigma}$.
	If $\stdtr\not\in\Lambda$ this is enough.
	Otherwise by assumption there is
	a formula $\hat \rho$ without $\Lambda$-modalities
	such that $\infers[\Lmucalc]\hat\rho\lbisubjunct \ddiamond{\stdtr}{\rho}$. Because $\hat\rho$ has no more free \pvariables than $\rho$,
	this $\hat\rho$ is as required.

\item \(\lfp{X}{\sigma}\):
	By induction hypothesis there is an \LmuV-formula $\rho$ 
	without $\Lambda$-modalities and with no more  free \pvariables than $\sigma$,
	such that $\infers[\Lmucalc] \rho \lbisubjunct\sigma$.
	\rref{prop:substituteoverproof} derives $\infers[\Lmucalc] \subst[\rho]{X}{\lfp{X}{\sigma}} \lbisubjunct\subst[\sigma]{X}{\lfp{X}{\sigma}}$.
	Then
	$\infers[\Lmucalc] \subst[\rho]{X}{\lfp{X}{\sigma}} \rightarrow \lfp{X}{\sigma}$ derives by axiom \irref{muf}.
	Rule \irref{muI} derives $\infers[\Lmucalc] \lfp{X}{\rho} \limply \lfp{X}{\sigma}$.
	The reverse implication is proved similarly.
	Then $\lfp{X}{\rho}$ is an \LmuV-formula with the required properties.

\item Box modalities \(\dbox{\stdtr}{\varphi}\) and greatest fixpoints \(\gfp{X}{\varphi}\) follow by negation and propositional reasoning.
\end{inparaitem}
\end{proof}

\rref{prop:relativecompletenessscheme} shows the utility of 
using the $\mu$-calculus for studying game logic.
Atomic games can be also replaced by equivalent games in game logic (\rref{prop:schematicgameremoval}{}).
However inside a composite game made up of repetitions, sequential composition and dual operators, they can not easily be simplified
to formulas.
Translating through the $\mu$-calculus shows that the ability to remove atomic occurrences of modalities suffices to remove them anywhere in
a \GL-formula. 

\begin{example}[Discrete assignments] \label{ex:discrete}
Consider a signature $\Lc =(\Lcf, \Act)$
containing deterministic and nondeterministic assignment modalities.
Choosing $\Lambda$ as the set of all deterministic assignments,
by \rref{prop:relativecompletenessscheme}, the deterministic-assignment-free fragment of \Lmu is provably expressive using axiom \irref{assignment} to eliminate a single deterministic assignment.
\end{example}

\section{Extended Actions: Differential Equations}
\label{sec:diffeq}

This section extends the first-order modal $\mu$-calculus and first-order game logic with differential equations as extended actions to model hybrid systems combining discrete and continuous dynamics.
Their proof calculi are extended to completely handle the additional expressivity syntactically.

\subsection{Integral Curves}

This section summarizes notations and definitions related 
to differential equations.
A continuous function 
\(\gamma: [a,b]\to\Rb\) is an \emph{integral curve} of the continuous function $F:\Rb^n\to \Rb^n$, if 
$\gamma$ is differentiable on $(a,b)$ with derivative ${\gamma}'(t)= F(\gamma(t))$ for all $t\in (a,b)$.
For readability the notation $\gamma_s =\gamma(s)$ is used synonymously.
Point $y$ is \emph{$t$-reachable} from $x$ along $F$, written \(x\xrightarrow{F}_t y\), iff there is an integral curve $\gamma:[0,t] \to \Rb^n$ of $F$
such that $\gamma_0 = x$ and $\gamma_t = y$.
If additionally $K\subseteq \Rb^n$ is such that $\gamma(s)\in K$ for all $s\in [0,t]$ write \(x \xrightarrow{F,K}_t y\).

\subsection{Differential Game Logic, Differential \mutex-Calculus}

\newcommand{\Langdlf}{\Lcf^{\dL}}%
\newcommand{\Actdl}{\Act^{\dL}}%
\newcommand{\Langdl}{\Lc^{\dL}}%

The \emph{signature of \dL} \cite{DBLP:conf/lics/Platzer12b} $\Langdl =(\Langdlf, \Actdl)$ consists of
the usual first-order signature $\Langdlf$ of ordered fields
and the collection $\Actdl$ of the following transition symbols:
	\begin{enumerate}
		\item nondeterministic assignments $x:=*$,
		\item deterministic assignments $x:=\stdterm$ for any $\Langdlf$-term $\stdterm$,
		\item continuous evolutions $\pevolvein{{x'}=\stdterm}{\psi}$ for ($N$-tuples of) variables $x$, ($N$-tuples of) $\Langdlf$-terms $\stdterm$ and first-order formulas $\psi$ in the language $\Langdlf$ and any $N\in \Nb$.
	\end{enumerate}

For readability, the notation $x'=\theta$ is used in continuous evolution modalities even though vectorial differential equations are also permitted such as
	\(x'=\alpha x - \beta x y, y' =\delta x y - \gamma y.\) %

The intended model for \Lmu[\Langdl] and \GL[\Langdl]
is the structure $\Rb$ with the usual ordered-field structure.
The interpretations of deterministic and nondeterministic assignments are as in \rref{sec:assignmentstructures}{}
to make $\Rb$ into an assignment structure and a quantifier structure.
The interpretation of the continuous evolution is as the evolution of a differential equation in $\psi$:
	\[\structA\envelope{\pevolvein{{x'}=\stdterm}{\psi}}
	= \{(\stdstate, \newstate)\in \Sc\times\Sc : \stdstate \xrightarrow{\theta, \psi} \newstate\}\]
Here $\stdstate \xrightarrow{\theta, \psi} \newstate$ abbreviates
the statement $\stdstate(x)\xrightarrow{F,K}\newstate(x)$
where $F(r) = \stdstate\tfrac{r}{x}\envelope{\theta}$
and $K = \{r : \structA,\stdstate\tfrac{r}{x}\models {\psi}\}$.
The $\Langdl$-structure with this interpretation is written $\Rf$.

Observe that the syntax of the logic \GL[\Langdl]
is the same as that of \emph{differential game logic} \dGL \cite{DBLP:journals/tocl/Platzer15}.
The only difference is in the semantics.
Differential game logic is an interpreted logic and as such
a \GL[\Langdl]-formula $\varphi$ is valid in the \dGL sense iff $\Rf\models\varphi$.
In contrast $\varphi$ is valid in the \GL[\Langdl] sense
iff it is true in \emph{all} suitable $\Langdl$-structures~$\structA$. 
The fragment of \dGL without $\pdual{\cdot}$ operator is \emph{differential dynamic logic} \dL \cite{DBLP:conf/lics/Platzer12b}.

The \emph{differential $\mu$-calculus} \dLmu{} is defined analogously to differential game logic but with the syntax of~\Lmu[\Langdl] without free variables.
Again, the semantics of \dLmu{} is that of \Lmu[\Langdl] restricted to the structure $\Rf$. \rref{cor:generalequi-expressivity} implies:

\begin{corollary}\label{cor:dGL-dLmu-equi-expressivity}
	Differential game logic and differential $\mu$-calculus are equi-expressive.
\end{corollary}
As for every assignment structure, the fixpoint variable hierarchy of the differential $\mu$-calculus collapses.
By coding arguments it can be seen that
the \ovariable hierarchies of \dLmu and \dGL \cite{DBLP:conf/lics/Platzer12b} collapse at some finite stage as well.
This holds generally for the first-order $\mu$-calculus and first-order game logic over acceptable structures \cite{Moschovakis74}.

The next question is what axioms are needed to handle the additional expressivity
brought to the language $\Langdl$ with continuous evolution modalities.
Relative completeness results can subsequently be obtained as corollaries to the above results for the proof calculi of $\mu$-calculus and game logic.

\subsection{Evolution Domain Constraints}

The evolution domain constraint of a differential equation is of fundamental importance for accurately modeling cyber-physical systems \cite{DBLP:conf/lics/Henzinger96}.
Evolution domain constraints however can be handled syntactically with the ``there and back again'' axiom \cite{DBLP:conf/lics/Platzer12b}, where $t$ is a fresh \ovariable:
\[
	\cinferenceRule[thereandbackagain|$\&$]{existential quantifier elimination axiom}
	   {
	   \linferenceRule[equivl2]
	   {\ddiamond{t:=0}{
		   \ddiamond{\pevolve{{x'}=\stdterm, t'=1}}{
			   (\varphi\land\dbox{\pevolve{{x'}=-\stdterm, t'=-1}}{(t{\geq}0\limply\psi)})
		   }
	   }}
	   {\ddiamond{\pevolvein{{x'}=\stdterm}{\psi}}{\varphi}}}{}
\]
This completely reduces differential equations with evolution domain constraints to ones without.
Instead of checking whether there is a continuous evolution along the differential equation to a state where $\varphi$ holds such that $\psi$ remains true throughout, axiom \irref{thereandbackagain} equivalently checks whether there is a continuous evolution along the differential equation to a state in which $\varphi$ holds and from which $\psi$ always holds when following the differential equation backwards in time.

Syntactically axiom \irref{thereandbackagain} is a formula of both
the differential $\mu$-calculus and differential game logic.
Its soundness proof \cite[Appendix B]{DBLP:conf/lics/Platzer12b} transfers directly to \dGL and \dLmu.

\begin{theoremEnd}[normal]{proposition}[{{\cite{DBLP:conf/lics/Platzer12b}}}] \label{prop:thereandback}
	The `there and back again' axiom \irref{thereandbackagain}
	is sound for differential $\mu$-calculus and differential game logic.
\end{theoremEnd}

Recall that evolution domain constraints were
syntactically assumed to be first-order formulas to avoid complications in the 
recursive definition of formulas.
As \rref{prop:thereandback} demonstrates, evolution domain constraints, even if crucial for modeling, do not increase the expressivity of the language.
It can be convenient to write
	$\pevolvein{{x'}=\stdterm}{\psi}$
where $\psi$ is also a \Lmu[\Langdl]-formula
 to abbreviate the \Lmu[{\Langdl}]-formula
\[\ddiamond{t:=0}{
	\ddiamond{\pevolve{{x'}=\stdterm, t'=1}}{
		(\varphi\land\dbox{\pevolve{{x'}=-\stdterm, t'=-1}}{(t\geq0\limply\psi)})
	}
}.\]
Similarly for box modalities and for \GL[\Langdl].

The dual game operator $\pdual{}$ makes eliminating evolution domain constraints in differential game logic  nicer \cite[Lem.\,3.4]{DBLP:journals/tocl/Platzer15}.
One can not only replace formulas $\ddiamond{\pevolvein{x'=\stdterm}{\psi}}{\varphi}$ by an equivalent \emph{formula} without evolution domain constraints, but indeed locally replace any game $\pevolvein{x'=\stdterm}{\psi}$ 
by an equivalent \emph{game} without evolution domain constraints.

\subsection{Differential Equation Modalities}

An axiom $\Delta$ which completely reduces continuous evolution modalities with semialgebraic postconditions was introduced and proved sound
for differential dynamic logic (\dL) \cite{DBLP:conf/lics/Platzer12b}.
Moreover it was proved that the fragment of \dL 
mentioning continuous evolutions only with semialgebraic postconditions
is (provably) as expressive as \dL.
From these facts it followed that differential dynamic logic is complete relative to its differential-equation-free fragment \cite{DBLP:conf/lics/Platzer12b}.

This approach cannot work for differential game logic,
which is strictly more expressive than differential dynamic logic \cite{DBLP:journals/tocl/Platzer15}.
Instead this section introduces an axiom for differential $\mu$-calculus that directly handles differential equations with arbitrary postconditions.
The idea is to characterize differential equations from the point of view of fixpoints.

Write $X= \Rb^N\times\Rb^N \times [0,\infty)$.
For a continuously differentiable function
$F:\Rb^N\to\Rb^N$ and a set $K\subseteq \Rb^N$
define
	$$\Rc_K =\{(x,y,t)\in X : x\xrightarrow{F,K}_t y\}.$$
The aim is to characterize $\Rc_K$ as the greatest fixpoint of a monotone map $\Gamma_K:\Pc(X)\to\Pc(X)$ for any compact set $K$.

Let $\norm[K]{\cdot}$ denote the supremum-norm
on the set $K\subseteq\Rb^N$.
That is, let \(\norm[K]{F} = \sup_{x\in K} \abs{F(x)}\) for any $F:\Rb^N\to\Rb^N$.

The following growth bound follows from the Taylor approximation to the integral curve.

\pgfkeys{/prAtEnd/forcompleteness/.style={
    text link={For a self-contained presentation, a \hyperref[proof:prAtEnd\pratendcountercurrent]{proof} is on page~\pageref{proof:prAtEnd\pratendcountercurrent}.}
  }
}
\begin{theoremEnd}[forcompleteness]{lemma}\label{lem:taylor}
	If $F:\Rb^N\to\Rb^N$ is a continuously differentiable function
	and $\gamma:[0,t]\to K$ an integral curve of $F$, then
	\begin{enumerate}
		\item $\abs{\gamma_t-\gamma_0}\leq t \norm[K]{F}$ and
		\item $\abs{\gamma_t-\gamma_0-t F(\gamma_0)}\leq \tfrac{t^2}{2} \norm[K]{(DF)F}$
	\end{enumerate}
	where $DF$ is the Jacobian matrix of partial derivatives of $F$ and $(DF)F$ is pointwise matrix vector multiplication.
\end{theoremEnd}
\begin{proofEnd}
	By the fundamental theorem
		$\gamma_t - \gamma_0 = \int_0^t \gamma_s'\ud s.$
	Hence $$|\gamma_t-\gamma_o|\leq t\max_{s\in[0,t]}|\gamma_s'|
	\leq t\max_{s\in[0,t]}|F(\gamma_s)|\leq t\norm[K]{F}$$
	since $\gamma(s)\in K$ for all $s\in[0,t]$.
	By partial integration
		\begin{IEEEeqnarray*}{rCl}
			\int_0^t (t-s)\gamma_s''\ud s = - t\gamma'(0)+\int_0^t \gamma_s'\ud s = \gamma_t-\gamma_0- t\gamma_0'.
		\end{IEEEeqnarray*}
	Since \(\gamma\) is an integral curve of $F$,
	the second derivative is \(\gamma_s''=DF(\gamma_s)F(\gamma_s)\).
	So \(\gamma_s'' \leq \norm[K]{(DF)F}\) because $\gamma_s\in K$.
	Now estimate
		\begin{IEEEeqnarray*}{+rCl+rx*}
			|\gamma_t-\gamma_0- t\gamma_0'| &\leq&
			 \int_0^t |(t-s) \gamma_s''|\ud s \\
			 &\leq&  \int_0^t |(t-s) DF(\gamma_s)F(\gamma_s)|\ud s \\
			 &\leq&  \norm[K]{(DF)F} \int_0^t (t-s)\ud s \\
			 &\leq&  \tfrac{t^2}{2} \norm[K]{(DF)F}&\qedhere
		\end{IEEEeqnarray*}
\end{proofEnd}

Define the subset of $X$ within the bounds of \rref{lem:taylor}:
	\begin{IEEEeqnarray*}{rl}
		G_K=\{&(x,y,t)\in X:
	\abs{y-x} \leq t\norm[K]{F}\\
	&\text{and}~
	 \abs{y-x-tF(x)} \leq \tfrac{t^2}{2} \norm[K]{(DF)F}\}
	\end{IEEEeqnarray*}
Note that $X\setminus G_K$ is the collection of $(x,y,t)$ for which, by \rref{lem:taylor}, no integral curve witnessing $x\xrightarrow{F,K}_t y$ can exist.

When $x\xrightarrow{F,K}_t y$,
there is a halfway point, which can be reached from $x$ along the integral 
curve in time $\frac{t}{2}$ and from which $y$ can be reached by continuing along
the curve for time $\frac{t}{2}$.
In other words $\Rc$ is a fixpoint of the map
$Z \mapsto\{(x,y,t)\in X :\mexists{u} \Lp x,u ,\tfrac{t}{2}\Rp, \Lp u, y,\tfrac{t}{2}\Rp \in Z\}$.
This map is not descriptive enough, as its least fixpoint is empty and its greatest fixpoint is $X$.
By imposing the additional condition that the triples be in $G_K$ 
for some compact set $K$, i.e. satisfy the bound above, $\Rc_K$ can be described as a fixpoint.
Formally:

\begin{proposition} \label{prop:ODEgfphelper}
	The map $\Gamma_K :\Pc(X)\to \Pc(X)$
		\begin{IEEEeqnarray*}{rl}
			\Gamma_K(Z) =  \{&(x,y,t)\in G_K : x,y \in K,
			\mexists{u} (x,u ,\tfrac{t}{2}), (u, y,\tfrac{t}{2}) \in Z\}
		\end{IEEEeqnarray*}
	is monotone.
	For compact~$K$, $\Rc_K$ is the greatest fixpoint of~$\Gamma_K$.
\end{proposition}

\begin{proof}
	For readability drop the subscript of $\Gamma$, $G$ and $\Rc$ in this proof.
	Monotonicity of $\Gamma$ is immediate.

	\emph{$\Rc$ is a post-fixpoint}: i.e.\ $\Rc \subseteq \Gamma(\Rc)$.
	Suppose $(x,y,t)\in \Rc$ and $\gamma$ is
	an integral curve of $F$
	in $K$ witnessing this.
	Then $(x,y,t)\in G$ by \rref{lem:taylor}.
	Finally observe that for $u=\gamma_{t/2}$
	the restrictions of $\gamma$
	to the intervals $[0,\tfrac{t}{2}]$ and
	$[\tfrac{t}{2},t]$ witness that
	$(x,u,\tfrac{t}{2})\in \Rc$ and 
	$(u,y,\tfrac{t}{2})\in \Rc$.
	Hence $(x,y,t)\in \Gamma(\Rc)$.

	\emph{Greatest post-fixpoint:} i.e.\ $A\subseteq \Rc$ for any fixpoint $A=\Gamma(A)$.
		Fix $M\geq \Vert F\Vert_K+\Vert (DF)F\Vert_K$
		and let $(x,y,t)\in A$.
		By recursion on $m$ define $(x^m_k)_{k\leq 2^m}\in \Rb^N$
		such
		that $(x^m_k,x^m_{k+1},t2^{-m})\in A$.

		For $m=0$ let $x^0_0 = x$, $x^0_1 = y$.
		For $m>0$, $k\leq 2^{m}$
		pick $u$  such that
		$(x^{m}_{k}, u, t2^{-m-1}), (u,x^{m}_{k+1}, t2^{-m-1}) \in A$.
		This is possible by $(x^m_k,x^m_{k+1},t2^{-m})\in A=\Gamma(A)$.
		Set $x^{m+1}_{2k} = x^m_k$ and $x^{m+1}_{2k+1}=u$.
		
		Because $A=\Gamma(A)\subseteq G$ this construction yields $x^m_k\in K$
		such that for all $m$ and all $1\leq k \leq 2^m$:
		\begin{IEEEeqnarray*}{rCl}
			|x^m_{k}-x^m_{k-1}| &\leq& t2^{-m}M\\
			|x^m_{k}-x^m_{k-1} -t2^{-m}F(x^m_{k-1})| &\leq& t^22^{-2m-1}M
		\end{IEEEeqnarray*}
		
		Define piecewise constant functions $\gamma^m : [0,t] \to \Rb^n$ by
			$$\gamma^m_s = x^m_{\lfloor 2^{m}\frac{s}{t}\rfloor}$$
		and to prove the following:
			\begin{enumerate}
				\item $\gamma^m$ converges uniformly to a function $\gamma:[0,t]\to K$ as $m\to \infty$,\label{p1}
				\item $\gamma$ is continuous\label{p2} and
				\item $|\gamma^m_s -\gamma^m_0 -\int_0^s F(\gamma^m_r) \ud r|\to 0$
					as $m\to \infty$.\label{p3}
			\end{enumerate}
		This suffices to show that $\gamma$ is an integral curve in $F$ witnessing $(x,y,t)\in \Rc$.
		Indeed by uniform convergence
			\[\gamma_s -\gamma_0 \stackrel{\text{\ref{p3}}}{=} 
			\lim_{m\to \infty}\int_0^s F(\gamma_r^m)\ud r
			\stackrel{\text{\ref{p1}}}{=} \int_0^s F(\gamma_r)\ud r\]
		and the fundamental theorem of calculus
		establishes that~\(\gamma'_s = F(\gamma_s)\) for $s\in(0,t)$ as desired.

		For point \ref{p1}, the uniform convergence of $\gamma^m$, consider ${m\geq k}$.
		By definition of $\gamma$ and the sequence $x$:
		\[\gamma^k_s = x^k_{\lfloor 2^k\tfrac{s}{t}\rfloor} = x^{k+1}_{2\lfloor 2^k\tfrac{s}{t}\rfloor}
		=\ldots = x^{k+m}_{2^{m-k}\lfloor 2^k\tfrac{s}{t}\rfloor}\]
		Note\footnote{Observe $\lfloor 2^{m}\frac{s}{t}\rfloor = \max\{n\in\naturals: n2^{-m}\leq \tfrac{s}{t}\}$.
		Then $\lfloor 2^{m}\frac{s}{t}\rfloor\geq 2^{m-k}\lfloor 2^{k}\frac{s}{t}\rfloor$
		is clear from $2^{m-k}\lfloor 2^{k}\frac{s}{t}\rfloor \cdot 2^{-m}
		= 2^{-k}\lfloor 2^{k}\frac{s}{t}\rfloor\leq\tfrac{s}{t}$.\\
		For the second inequality let $a=\lfloor 2^{m}\frac{s}{t}\rfloor$
		and $b=\lfloor 2^{k}\frac{s}{t}\rfloor$.
		Then $2^{-m}a\leq \tfrac{s}{t}\leq (b+1)2^{-k}$.
		Hence $a-2^{m-k}b\leq (b+1)2^{m-k}-b2^{m-k}=2^{m-k}$}
		 $0\leq\lfloor 2^{m}\tfrac{s}{t}\rfloor
		- 2^{m-k}\lfloor 2^{k}\tfrac{s}{t}\rfloor\leq 2^{m-k}$.
		By triangle inequality:
		\begin{IEEEeqnarray*}{rCl}
			|\gamma^m_s-\gamma^k_s| &=& 
			|x^m_{\lfloor 2^{m}\frac{s}{t}\rfloor}-x^m_{2^{m-k}\lfloor 2^{k}\frac{s}{t}\rfloor}|
			\leq \sum^{\lfloor 2^{m}\frac{s}{t}\rfloor-1}_{i=2^{m-k}\lfloor 2^{k}\frac{s}{t}\rfloor} |x_{i+1}^m-x_i^m| \\
			&\leq& (\lfloor 2^{m}\tfrac{s}{t}\rfloor
			- 2^{m-k}\lfloor 2^{k}\tfrac{s}{t}\rfloor) t2^{-m}M\leq t2^{-k}M
		\end{IEEEeqnarray*}
		This is arbitrarily small for large enough $k$.
		That is, the sequence $\gamma^m$ is uniformly Cauchy, hence converges uniformly
		to some $\gamma:[0,t]\to \Rb^n$.
		Since the range of every $\gamma^m$ is a subset of $K$ and
		$K$ is closed,
		the range of $\gamma$ is a subset of $K$.

		For point \ref{p2}, the continuity of $\gamma$, consider $s\in[0,t]$
		and $\varepsilon>0$.
		By uniform convergence pick $m$ large enough such that
		$t2^{-m}M\leq \tfrac{\varepsilon}{6}$ and
		$\Vert\gamma^m-\gamma\Vert_K\leq\tfrac{\varepsilon}{2}$.
		Consider some $r\in [s,t]$.
		Note that $0\leq\lfloor 2^m\tfrac{r}{t}\rfloor-\lfloor 2^m\tfrac{s}{t}\rfloor \leq 2^m\tfrac{|r-s|}{t}+1$.\footnote{Let $a = \lfloor 2^m\tfrac{r}{t}\rfloor$ and $b=\lfloor 2^m\tfrac{s}{t}\rfloor$. Then combine
		$a2^{-m}\leq \tfrac{r}{t}$ and $(b+1)2^{-m}\geq \tfrac{s}{t}$ to
		get $2^{-m}(a-b-1)\leq \tfrac{|r-s|}{t}$.}
		Then
		\begin{IEEEeqnarray*}{rCl}
				|\gamma^m_s-\gamma^m_r|&=& 
				|x^m_{\lfloor 2^{m}\frac{s}{t}\rfloor}-x^m_{\lfloor 2^{m}\frac{r}{t}\rfloor}|
				\leq \sum^{\lfloor 2^{m}\frac{s}{t}\rfloor-1}_{i=\lfloor 2^{m}\frac{r}{t}\rfloor} |x_{i+1}^m-x_i^m| \\
				&\leq& (\lfloor 2^{m}\tfrac{s}{t}\rfloor
				- \lfloor 2^{m}\tfrac{r}{t}\rfloor) t2^{-m}M\\
				&\leq& |r-s|M + t2^{-m}N \leq|r-s|M +\tfrac{\varepsilon}{6}
			\end{IEEEeqnarray*}
		By symmetry the same inequality holds for $r\in [0,s]$.
		Also
			\begin{IEEEeqnarray*}{rCl}
				|\gamma_s-\gamma_r|&\leq&
				|\gamma_s-\gamma^m_s| +|\gamma^m_s-\gamma^m_r|+|\gamma^m_r-\gamma_r|\\
				&\leq& 2\Vert\gamma^m-\gamma\Vert_K + |r-s|N +\tfrac{\varepsilon}{6}
				\leq \tfrac{5}{6}\varepsilon+ |r-s|M.
			\end{IEEEeqnarray*}
		Thus $|\gamma_s-\gamma_r|\leq \varepsilon$
		for all $r$ with $|s-r|\leq \tfrac{\varepsilon}{6N}$, i.e.
		$\gamma$ continuous.

		It remains to prove point \ref{p3}.
		Because $\gamma^m_r$ is piecewise constant, applying the triangle inequality yields:
			\begin{IEEEeqnarray*}{+rCl+rx*}
				|\gamma^m_s - \gamma^m_0 - \int_0^s F(\gamma^m_r) \ud r|
				&\leq& 
				\sum_{k=1}^{\lfloor 2^m \frac{s}{t}\rfloor}
				\left |x^m_{k} - x^m_{k-1} - t2^{-m}F(x_{k-1}^m)
				\right | 
				\\&&+|\gamma^m_s-\gamma^m_{\lfloor 2^m \frac{s}{t}\rfloor}-t 2^{-m}F(x_{\lfloor 2^m \frac{s}{t}\rfloor}^m)|\\
				&\leq& 2^m\tfrac{s}{t}t^2 2^{-2m-1}M 
				+ t2^{-m}\Vert F\Vert_K\\
				&\leq& 2^{-m}Mt(\tfrac{s}{2}+1)
				\xrightarrow{m\to \infty} 0 \quad\quad\qedhere
			\end{IEEEeqnarray*}
\end{proof}

\newcommand{\secondderstdterm}{\hat \theta}%

The fixpoint representation of $\Rc_K$ turns the continuous evolution modality into a fixpoint operator in \dLmu.
Consider a modality $\ddiamond{\pevolve{x'=\stdterm}}{\varphi}$ written in one-dimensional notation (even if it could have $N$-dimensional tuples of variables and terms in
place of $x$ and $\stdterm$).
Now pick $F(r) = \modif{\stdstate}{x}{r}\envelope{\theta}$.
Because $\theta$ is a term in the language of ordered fields, it can
be thought of as a 
vector of multivariate polynomials.
As such $F$ is continuously differentiable and moreover there is another $\Langdlf$-term 
$\secondderstdterm$ such that $\stdstate\tfrac{a}{x}\envelope{\secondderstdterm} = DF(a)F(a)$.
If $\stdterm= (p_1, \ldots, p_n)$ is a vector of multivariate polynomials $p_1, \ldots,p_n$, then
$\secondderstdterm$ is the vector of multivariate polynomials
\begin{equation}\label{eq:gradprodterm}
\secondderstdterm = \begin{pmatrix}
	\sum_{i=1}^n (\partial_ip_1)\cdot p_i,\ldots,
	\sum_{i=1}^n (\partial_ip_n)\cdot p_i
\end{pmatrix}
\tag{\text{$\ast$}}
\end{equation}
where $\partial_i p_j$ is the formal derivative of $p_j$ with respect to $x_i$.

By \rref{prop:ODEgfphelper} the reachability relation of $F$ in the compact set
$K$ is the greatest fixpoint of~$\Gamma_K$.
The following axiom captures this syntactically in differential $\mu$-calculus.
\[
		\cinferenceRule[nabla|$\nabla$]{nabla axiom}
   		{
  		 \linferenceRule[equiv]
		   {\lexists{y}{\left( \ddiamond{x:=y}{\varphi}
		   \land \lexists{t{>}0}{\lexists{M{>}0}{\gfp{X}{\rho}}}\right)}}
   		{\ddiamond{\pevolve{x'=\stdterm}}{\varphi}}
   		}{}
\]
	where $\rho$ is the conjunction of
	\begin{enumerate}
		\item $\abs{x}\leq M \land \abs{y}\leq M$
		\item $\lexists{z}{\big( \abs{z}\leq M \land \abs{y-x} \leq t\abs{\transposeterm{\stdterm}{x}{z}}\big)}$
		\item $\lexists{z}{\big( \abs{z}\leq M \land 2\abs{y-x -t \stdterm} \leq t^2\abs{\transposeterm{\secondderstdterm}{x}{z}}\big)}$ with $\secondderstdterm$ by \rref{eq:gradprodterm}
		\item $\lexists{u}{\ddiamond{t:=\tfrac{t}{2}}{(\ddiamond{y:=u}{X} \land
		\ddiamond{x:=u}{X})}}$.
	\end{enumerate}

\noindent
The variable $y$ in axiom \irref{nabla}
represents the point reachable from $x$ along the differential equation,
which witnesses that $\ddiamond{x'=\stdterm}{\varphi}$.
This explains the condition $\ddiamond{x:=y}{\varphi}$,
i.e. that formula $\varphi$ 
must be true in the final state.
Moreover $y$ must be reachable from $x$ along the differential equation.
Equivalently there must be some time $t$ and some $M$ such that 
the point $y$ is reachable from point $x$ along an integral curve of $F$ in time $t$
while staying in the compact set $K=\{z : |z|\leq M\}$.
(By continuity of integral curves
and compactness of the time interval.)
By \rref{prop:ODEgfphelper} this is equivalent to $(x,y,t)$ belonging
to the greatest fixpoint of~$\Gamma_K$.
Syntactically this is expressed as the formula $\gfp{X}{\rho}$.
The first conjunct of $\rho$ corresponds to $x,y\in K$. 
The second and third capture that $(x,y,t)\in G_K$.
To ensure $\abs{y-x}\leq t\norm[K]{F}$, by compactness, it suffices to
require the existence of some point $z\in K$
at which $t|F(z)|$ (syntactically $t|\transposeterm{\stdterm}{x}{z}|$) is greater than $\abs{y-x}$. Similarly for the second condition of $G_K$.
The fourth conjunct represents the defining condition of the fixpoint map $\Gamma_K$.

\begin{corollary}
	The axiom scheme \irref{nabla} is sound for the
	differential $\mu$-calculus.
\end{corollary}

\subsection{Proof Calculi}
For a \dLmu-formula $\varphi$ write \(\infers[\dLmu] \varphi\)
to mean \(T_\dLmusubst \infers[{\Lmucalc[\Langdl]}] \varphi\) where $T_\dLmusubst$
consists of the set of all $\Lmu[\Langdl]$ instances of the axioms
\irref{thereandbackagain} and \irref{nabla}
together with all first-order formulas valid in~$\Rb$.

Similarly for \dGL-formula $\varphi$ write \(\infers[\dGL] \varphi\)
for \m{T_{\dGL} \infers[{\GLcalc[\Langdl]}] \varphi} where $T_{\dGL}$
consists of the set of all $\GL[\Langdl]$ instances of the axioms
\irref{thereandbackagain}, all first-order formulas that
are valid in~$\Rb$
and all formulas $\psi^\flat$ for instances $\psi$ of \irref{nabla}.
Modulo the axioms \irref{nabla} and \irref{thereandbackagain} this is the calculus for differential game logic \cite{DBLP:journals/tocl/Platzer15}.

These are both computable proof calculi by the decidability of first-order real arithmetic due to Tarski-Seidenberg \cite{tarski_decisionalgebra51}.

The results from Sections~\ref{sec:secequiexpressive}
and~\ref{sec:relatingcalculi} are applicable to these calculi, because an independent control variable can always be chosen and its independence proved syntactically from
the axioms \irref{assignment} and \irref{nabla}.
Hence differential game logic and differential $\mu$-calculus are equi-potent by \rref{prop:deductivelyequivalent}.

Both calculi have essentially the same axioms
stated in their respective syntax.
To be precise \(\infers[\dGL] \varphi^\flat\) for any $\varphi\in T_\dLmusubst$
and \(\infers[\dLmu] \varphi^\sharp\) for any $\varphi\in T_{\dGL}$.
Moreover any valid formula of differential game logic is provable from the translations of valid formulas of the differential $\mu$-calculus and vice versa.

Observe that the natural numbers are definable in \dLmu:
	\[n\in \Nb \quad\lbisubjunct\quad \lfp{X}{(n=0 \lor \ddiamond{n:=n-1}{X})}\]
First-order integer arithmetic is
interpretable in differential $\mu$-calculus.
The proof calculus must therefore be incomplete by G\"odel's incompleteness theorem \cite{Goedel_1931}.
Hence
there is a true \dLmu-formula not provable in that calculus.
The relevant notion to look at is therefore relative completeness.

\subsection{Relative Completeness}

The differential $\mu$-calculus
is \emph{complete relative} to a syntactic fragment, if for any 
valid \dLmu-formula $\varphi$ there is a valid formula $\rho$ in that fragment such that \(\infers[\dLmu] \rho \limply\varphi\).

The main tool to prove relative completeness for the differential $\mu$-calculus
is \rref{prop:relativecompletenessscheme}.
The next proposition yields the stronger property of \dLmu-provable expressivity.
Recall \rref{def:provablyexpressiv}: a syntactic fragment of \dLmu is 
\emph{\dLmu-provably expressive},
if for any \dLmu-formula $\varphi$ there is another
formula~$\rho$ in that fragment such that \(\infers[\dLmu] \rho \lbisubjunct\varphi\).

\begin{theorem}[Relative completeness] \label{thm:dlmuexpressiveimpliesrelativelycomplete}
	The logic \dLmu is complete relative to \emph{any} \dLmu-provably expressive fragment.
\end{theorem}

\begin{proof}
	Let $K$ be a \dLmu-provably expressive fragment.
	Suppose $\varphi$ is a valid \dLmu formula.
	There is a formula $\rho$ in the fragment~$K$
	such that \(\infers[\dLmu] \varphi \lbisubjunct\rho\).
	By soundness of the proof calculus and validity of $\varphi$,
	the formula $\rho$ is also valid.
	Thus $\rho$ is a valid formula of $K$ such that 
	\(\infers[\dLmu] \rho \limply\varphi\).
\end{proof}

\begin{corollary}\label{cor:removeevolutiondomainconstraint}
	The fragment of differential $\mu$-calculus 
	without evolution domain constraints is \dLmu-provably expressive.
\end{corollary}

\begin{proof}
	By axiom \irref{thereandbackagain}, this follows by \rref{prop:relativecompletenessscheme} applied to the set of modalities with evolution domain constraints.
\end{proof}

The \emph{continuous} fragment of \dLmu is the fragment containing no modalities except differential equation modalities. 
The fragment of \dLmu containing only nondeterministic assignment modalities is first-order least-fixpoint logic over~$\Rb$.

The continuous and least-fixpoint fragments have the same expressive power.
As long as one of the modalities allows state change it is immaterial which one.

\begin{theorem}[Expressive] \label{thm:dlmuexpressive}
	The continuous fragment and least-fixpoint logic over $\Rb$
	are \dLmu-provably expressive.
	So \dLmu is complete relative to both fragments.
\end{theorem}

\begin{proof} 
	First prove that the fragment without differential equations is \dLmu-provably expressive.
	By transitivity,
	it suffices to prove that this fragment is \dLmu-provably expressive 
	for the evolution domain
	constraint free fragment, which is itself 
	\dLmu-provably expressive by \rref{cor:removeevolutiondomainconstraint}.
	To see that this is the case
	apply \rref{prop:relativecompletenessscheme} for the evolution domain constraint free fragment,
	where axiom \irref{nabla} eliminates differential equation modalities without evolution domain constraints.

	To see that least-fixpoint logic over $\Rb$ is \dLmu-provably expressive, it suffices to show that it is provably-expressive relative to the differential-equation free fragment.
	This also follows with \rref{prop:relativecompletenessscheme}
	applied to the set of deterministic assignment modalities using axiom \irref{assignment} for elimination.
	
	Using the \irref{nabla} axiom prove\footnote{%
	Using $\rho$ from the instance of \irref{nabla} for $x'=1$, the axioms for fixpoints
	and \irref{fixpointnumonotonicity}, it follows that
	$t>0\land M>|x|+|y| \limply (\gfp{X}{\rho} \leftrightarrow y = x+t)$.
	(For one direction use the third conjunct and the fact that $\secondderstdterm=0$.)
	Thus $\ddiamond{x'=1}{\varphi}\leftrightarrow \lexists{t{>}0}{\ddiamond{x:=x+t}{\varphi}}$.
	Similarly for $x'=-1$.}
	\[\ddiamond{\prandom{x}}{\varphi} \lbisubjunct \ddiamond{x'=1}{\varphi} \lor \ddiamond{x'=-1}{\varphi}.\]
	With this derived axiom and \rref{prop:relativecompletenessscheme} applied to
	first-order least fixpoint logic, it follows that the
	continuous fragment is \dLmu-provably expressive.

	Relative completeness then follows from \rref{thm:dlmuexpressiveimpliesrelativelycomplete}.
\end{proof}

Analogously to the differential $\mu$-calculus there are notions of \dGL-provable
expressivity and relative completeness for fragments of differential game logic defined accordingly.

Almost the same proof as \rref{thm:dlmuexpressiveimpliesrelativelycomplete} shows that \dGL is complete relative
to any \dGL-provably expressive fragment.
\begin{theorem}[Relative completeness]
	The logic \dGL is complete relative to any \dGL-provably expressive fragment.
\end{theorem}
The stronger statement, that \dGL is complete relative to any \emph{differentially expressive}\footnote{
\dGL fragment $K$ is
\emph{differentially expressive} if $K$ is as expressive as \dGL
and for any $K$-formula $\varphi$ there is a $K$-formula $\psi$ such that
\(\infers[\dGL] \ddiamond{\pevolve{x'=\stdterm}}{\varphi}\lbisubjunct\psi\).}
fragment, is true \cite[Thm.\,4.5]{DBLP:journals/tocl/Platzer15}.
This stronger result is not needed for what follows.
Similarly to \dLmu, the continuous fragment of \dGL consists of all formulas containing as atomic games only continuous evolutions.
The discrete fragment of \dGL consists of all formulas containing as atomic games only nondeterministic assignments.
\rref{thm:dlmuexpressive}{} carries over to \dLmu via the translation ${}^\flat$.

\begin{theorem}[Expressive] \label{thm:dgleqfreeprovexpr}
	The continuous and the discrete fragments of \dGL are
	provably expressive for \dGL.
\end{theorem}

\begin{proof}
	Consider a \dGL formula $\varphi$.
	By Theorem \ref{thm:dlmuexpressive} there is a formula
	$\rho$ of the continuous fragment of \dLmu such
	that \(\infers[{\dLmu}] \rho \lbisubjunct\varphi^\sharp\).
	A derivation witnessing this relies only on finitely many $T_\dLmusubst$-formulas, say \(\delta_1,\ldots\delta_k\).
	Thus, 
	\begin{IEEEeqnarray*}{l}
		\infers[{\Lmucalc[\Langdl]}] (\delta_1\land\ldots\land\delta_k)\limply(\rho \lbisubjunct\varphi^\sharp)\quad\quad \text{and}\\
		\infers[{\GLcalc[\Langdl]}] (\delta_1\land\ldots\land\delta_k)^\flat\limply(\rho^\flat \lbisubjunct{\varphi^\sharp}^\flat)
	\end{IEEEeqnarray*}
	by \rref{prop:deductivelyequivalent}.
	The equivalence \(\infers[\dGL] \rho^\flat \lbisubjunct{\varphi^\sharp}^\flat\) derives because \(\infers[\dGL] \delta_i^\flat\).
	With \rref{prop:rco} and \irref{modusPonens} derive 
	\(\infers[{\dGL}] \rho^\flat \lbisubjunct\varphi\).
	Now $\rho$ is differential equation-free and because ${}^\flat$ does not introduce
	any differential equations, $\varphi$ is \dGL-provably
	equivalent to the differential equation free formula~$\rho^\flat$.
	For all \dGL-formulas $\psi$, derive by \irref{assignment+composition+test}:
		\[\infers[\dGL] \ddiamond{x:=\stdterm}{\psi} \lbisubjunct \ddiamond{\prandom{x};\ptest{x=\stdterm}}{\psi}
		\]
	Hence $\infers[\dGL] \rho \lbisubjunct\tilde \rho$
	where $\tilde \rho$ is obtained from $\rho$ by replacing all deterministic assignments
	$x:=\stdterm$ by 
	$\prandom{x};\ptest{x=\stdterm}$.
	\footnote{For details see \rref{prop:schematicgameremoval}{}.}

	The case for the discrete fragment is similar. By the same argument as above, choosing $\rho$ as a nondeterministic assignment free formula instead, the nondeterministic assignment-free fragment of \dGL is provably expressive for \dGL.
	To remove remaining deterministic assignments
	use
	\[\infers[\dGL] \ddiamond{x:=\stdterm}{\psi} \lbisubjunct \ddiamond{(\pchoice{x'=1}{x'=-1});\ptest{x=\stdterm}}{\psi}
		\]
	instead.
	This is provable by \rref{prop:rco} and the fact that the translation of the formula by ${}^\sharp$ is provable from \irref{assignment} and the axiom derived in the proof of \rref{thm:dlmuexpressive}{}.
\end{proof}

As differential game logic is complete relative to any \dGL-provably expressive fragment,
\rref{thm:dgleqfreeprovexpr}
implies:
\begin{corollary}
	Differential game logic is complete relative to its continuous and its discrete fragment.
\end{corollary}

In contrast to the \pvariable hierarchy,
which collapses
over assignment structures,
general results \cite[5.B]{Moschovakis74} imply that the alternation hierarchy does not collapse.
There is a computable coding $\ulcorner \cdot\urcorner$ of 
\dLmu-formulas into natural numbers and a \dLmu-formula 
	$\mathrm{VALID_n}(\ulcorner\varphi\urcorner)$
with $n+1$ fixpoint operator alternations such that $\mathrm{VALID_n}(\ulcorner\varphi\urcorner)$ is valid for an $n$-fixpoint alternation
formula $\varphi$ exactly if $\varphi$ is valid.

Thus, \dLmu is not complete relative to the 
n-fixpoint operator alternation fragment $F$.
Since \dLmu{} provability
$\mathrm{PROV}(\ulcorner\varphi\urcorner)$
is definable in \dLmu, completeness of \dLmu relative to $F$
would imply that truth of a \dLmu formula
is definable in \dLmu{} by
	$$\lexists{\ulcorner\psi\urcorner \in F}{(\mathrm{VALID_n}(\ulcorner\psi\urcorner) \land \mathrm{PROV}(\ulcorner\psi\limply\varphi\urcorner))}.$$
This would contradict Tarski's undefinability theorem \cite{Tarski36}.

By \rref{cor:dGL-dLmu-equi-expressivity} these results carry over to differential game logic. The fragment with $n$-nested $\pdual{{\cdot^*}}$ games is strictly less expressive than the fragment
with $(n+1)$ such games.

\section{Related Work}
The modal $\mu$-calculus and its relation to model checking is well-studied \cite{BradfieldS06,DBLP:reference/mc/BradfieldW18,DBLP:conf/focs/Pratt81,DBLP:journals/tcs/EmersonJS01}.
Both completeness 
\cite{DBLP:journals/iandc/Walukiewicz00,DBLP:journals/tcs/EnqvistSV18, DBLP:journals/sLogica/Studer08,DBLP:conf/lics/AfshariL17} and expressivity 
\cite{DBLP:conf/stacs/Bradfield96,DBLP:conf/concur/JaninW96} of the propositional modal $\mu$-calculus have received significant attention.
Strictness of its fixpoint alternation \cite{DBLP:conf/concur/Bradfield96} and variable \cite{DBLP:journals/mst/BerwangerGL07} hierarchies are major results.

The literature \cite{DBLP:journals/bsl/DawarG02} has a survey on fixpoint logics extending first-order logic with fixpoint operators on finite structures.
Expressive equivalence of inflationary and least fixpoint operators was shown for finite \cite{DBLP:journals/apal/GurevichS86} and 
infinite structures \cite{DBLP:conf/lics/Kreutzer02}.
The sets of integers definable by alternating least and greatest fixpoints were investigated set theoretically \cite{DBLP:conf/lics/Lubarsky89}.

Propositional game logic is due to Parikh \cite{DBLP:conf/focs/Parikh83}.
Pauly \cite{pauly01} gives an exposition of game logic and interpretations over different kinds of transition systems in the propositional case.
General connections between modal $\mu$-calculus and games have been observed \cite{DBLP:conf/focs/EmersonJ91}, but modal $\mu$-calculus is inexpressible in game logic in the propositional case \cite{DBLP:journals/mst/BerwangerGL07}.

A few aspects of hybrid systems can already be cast in propositional modal $\mu$-calculus \cite{DBLP:conf/hybrid/Davoren97}, but first-order modalities are crucial for describing hybrid systems.
Differential dynamic logic is complete relative to its discrete and continuous fragments \cite{DBLP:conf/lics/Platzer12b}.
Its game variant \dGL is complete relative to differentially expressive sublogics \cite{DBLP:journals/tocl/Platzer15}.

\section{Conclusion}

The modal $\mu$-calculus and game logic were shown to have the same expressive power when interpreted over first-order structures.
Proof calculi for both logics were introduced and shown to have the same deductive power.
Augmented with differential equation modalities, this interplay was exploited to prove relative completeness and expressiveness results for differential $\mu$-calculus and differential game logic.

\bibliographystyle{ACM-Reference-Format}
\bibliography{platzer,dLmu}

\clearpage

\appendix
\section{Substitutions}

\subsection{Substitution of \PVariables}\label{app:substitutionpvars}

Let $\varphi,\psi$ be \LmuV-formulas and $X$ a \pvariable.
The formula $\subst[\varphi]{X}{\psi}$ obtained 
by replacing all free occurrences of $X$ by $\psi$ is defined as follows.
An occurrence of $X$ is free in $\varphi$, if it does not occur in the scope of a least or greatest fixpoint operator binding $X$.
Formally:
\begin{IEEEeqnarray*}{rClCrCl}
	\subst[\stdlit]{X}{\psi} &\equiv& 
	 \stdlit ~\text{for}~\stdlit\in\FOLL_{\Lcf}
	 \\
	\subst[Y]{X}{\psi} &\equiv& {Y} ~\text{for}~Y\in\overline{\Vc}\setminus \{X,\overline{X}\}
	\\
	\subst[X]{X}{\psi} &\equiv& {\psi}
	\\
	\subst[\overline{X}]{X}{\psi} &\equiv& \overline{\psi}
	\\
	\subst[(\varphi_1\lor\varphi_2)]{X}{\psi} &\equiv& 
	\subst[\varphi_1]{X}{\psi} \lor \subst[\varphi_2]{X}{\psi}
	\\
	\subst[(\varphi_1\land\varphi_2)]{X}{\psi} &\equiv& 
	\subst[\varphi_1]{X}{\psi}\land\subst[\varphi_2]{X}{\psi}
	\\
	\subst[(\ddiamond{\stdtr}{\varphi})]{X}{\psi} &\equiv& \ddiamond{\stdtr}{(\subst[\varphi]{X}{\psi})}
	\\
	\subst[(\dbox{\stdtr}{\varphi})]{X}{\psi} &\equiv& \dbox{\stdtr}{(\subst[\varphi]{X}{\psi})}
	\\
	\subst[(\sigma{Y}.{\varphi})]{X}{\psi} &\equiv& 
	\sigma{Y}.{(\subst[\varphi]{X}{\psi})}&\text{if $X\notin\{Y,\overline{Y}\}$}
	\\
	\subst[(\sigma{Y}.{\varphi})]{X}{\psi} &\equiv& 
	\sigma{Y}.{\varphi}&\text{if $X\in\{Y,\overline{Y}\}$}
\end{IEEEeqnarray*}
where $\sigma\in \{\mu,\nu\}$.

Note the following technical facts about substitution.
\begin{lemma}\label{lem:negationsubstitution}
	For \LmuV-formulas $\varphi,\psi$:
	\[\overline{\subst[\varphi]{X}{\psi}}\equiv \subst[\overline{\varphi}]{X}{\psi}\]
\end{lemma}
\begin{proof}
	By straightforward induction on the formula $\varphi$.
\end{proof}

\begin{lemma}\label{lem:doublesubstitution}
	For \LmuV-formulas $\varphi,\psi,\rho$ and $Y\in \{X,\overline{X}\}$
		$$(\varphi\tfrac{\psi}{X})\tfrac{\rho}{Y}\equiv\varphi\tfrac{\psi\tfrac{\rho}{Y}}{X}
	$$
	And if $X$ and $\overline{X}$ do not occur freely in $\rho$ and $Y\notin\{ X,\overline{X}\}$
	$$(\varphi\tfrac{\psi}{X})\tfrac{\rho}{Y}\equiv
		(\varphi\tfrac{\rho}{Y})\tfrac{\psi\tfrac{\rho}{Y}}{X}$$
\end{lemma}

\begin{proof}
	The first two identities are by a straightforward induction
	on $\varphi$,
	using \rref{lem:negationsubstitution} for the case where $\varphi$ is $\overline{X}$.

	The third identity is also by a simple induction on $\varphi$.
	The assumption that $X$ and $\overline{X}$ do not occur freely in $\rho$
	are needed for the cases where $\varphi$ is $Y$ or $\overline{Y}$.
\end{proof}

\Pvariable $X$ is \emph{free for $\varphi$ in $\psi$} iff in $\psi$ the variable $X$
does not occur in the scope of a fixpoint operator ($\mu,\nu$) binding a free \pvariable of $\varphi$.

\begin{lemma} \label{lem:subsitutesemantically}
	For \LmuV-formulas $\varphi,\psi$ \assignments $\stdasst$ and $X\in \overline{\Vc}$.
	If $X$ is free for $\varphi$ in $\psi$ then
		$$\stdasst\envelopemu{\varphi\subst{X}{\psi}} = \modif{\stdasst}{X}{\stdasst\envelopemu{\psi}}\envelopemu{\varphi} $$
\end{lemma}
\begin{proof}
	By a straightforward induction on the formula $\varphi$.
	The interesting case is for fixpoint formulas $\lfp{Y}{\varphi}$ with $Y\notin \{X,\overline X\}$.
	\begin{IEEEeqnarray*}{+rCl+rx*}
		\stdasst\envelopemu{(\lfp{Y}{\varphi})\subst{X}{\psi}}
		&=&
		\capfold\{Z :\stdasst\tfrac{Z}{Y}\envelopemu{\varphi\tfrac{\psi}{X}} \subseteq Z\}
		\\
		&=&
		\capfold\{Z :\stdasst\tfrac{Z}{Y}\tfrac{\stdasst\frac{Z}{Y}\envelopemu{\psi}}{X}\envelopemu{\varphi} \subseteq Z\}
		\\
		&=&
		\capfold\{Z :\stdasst\tfrac{Z}{Y}\tfrac{\stdasst\envelopemu{\psi}}{X}\envelopemu{\varphi} \subseteq Z\} \\
		&=&  \modif{\stdasst}{X}{\stdasst\envelopemu{\psi}}\envelopemu{\varphi}
	\end{IEEEeqnarray*}
	For the third equality observe that if $X$ occurs in $\varphi$, then $Y$ can not be free in $\psi$, since
	$X$ is free for $\psi$ in $\varphi$.
\end{proof}

\subsection{Renaming of \OVariables}
\label{app:substitutionovars}

\newcommand{\transposefmla}[3]{{#1}\sfrac{#3}{#2}}

Let $\varphi$ be an \LmuV-formula.
Define the formula $\transposefmla{\varphi}{x}{y}$ obtained
from $\varphi$ by 
renaming the \ovariable $x$ to $y$ and vice versa.
For $\stdlit\in\FOLL_{\Lcf}$ this is as in first-order logic.
For a \pvariable $X$ let $\transposefmla{X}{x}{y}\equiv X$.
Extend the definition recursively to \LmuV by:
\begin{IEEEeqnarray*}{rClCrCl}
	\transposefmla{(\varphi_1\lor\varphi_2)}{x}{y} &\equiv& 
	\transposefmla{\varphi_1}{x}{y} \lor \transposefmla{\varphi_2}{x}{y}
	\\
	\transposefmla{(\varphi_1\land\varphi_2)}{x}{y} &\equiv& 
	\transposefmla{\varphi_1}{x}{y} \land \transposefmla{\varphi_2}{x}{y}
	\\
	\transposefmla{(\ddiamond{\stdtr}{\varphi})}{x}{y} &\equiv& \ddiamond{\transposeact{\stdtr}{x}{y}}{\transposefmla{\varphi}{x}{y}} 
	\\
	\transposefmla{(\dbox{\stdtr}{\varphi})}{x}{y} &\equiv& \dbox{\transposeact{\stdtr}{x}{y}}{\transposefmla{\varphi}{x}{y}} 
	\\
	\transposefmla{(\sigma{X}.{\varphi})}{x}{y} &\equiv& 
	\sigma{X}.{(\transposefmla{\varphi}{x}{y})}
\end{IEEEeqnarray*}
where $\sigma\in \{\mu,\nu\}$.
Finally define $\transposefml{\varphi}{x}{y}$ to be the formula
obtained from $\transposefmla{\varphi}{x}{y}$ by replacing all \emph{free} occurrence of \pvariables $Z$ by $\transposepv{Z}{x}{y}$.
(Note that $\subst[(\transposepv{\varphi}{x}{y})]{X}{\psi} = 
\transposeterm{(\varphi\tfrac{\psi}{X})}{x}{y}$.)

For any set $D\subseteq\Sc$ let $$\transposefml{D}{x}{y}
=\{\transposestate{\stdstate}{x}{y} : \stdstate \in D\}$$

\begin{lemma}\label{lem:transposesubst}
	For any \LmuV-formula $\psi$
		\[\stdasst\envelopemu{\transposefml{\psi}{x}{y}} = \transposefml{\stdasst\envelopemu{\psi}}{x}{y}.\]
\end{lemma}

\begin{proof}
	Let $\transposepv{\stdasst}{x}{y}(Z) = \transposepv{\stdasst(Z)}{x}{y}$ for all \pvariables $Z$.
	Observe that $\stdasst\envelopemu{\transposefml{\psi}{x}{y}}
	=\transposepv{\stdasst}{x}{y}\envelopemu{\transposefmla{\psi}{x}{y}}$
	by the definition of $\transposefml{\varphi}{x}{y}$ and the interpretation of the \pvariables $\transposepv{X}{x}{y}$,
	By a straightforward induction on the formula $\psi$
	prove that $\transposepv{\stdasst}{x}{y}\envelopemu{\transposefmla{\psi}{x}{y}}=\transposefml{\stdasst\envelopemu{\psi}}{x}{y}$.
	If $\psi$ is a first-order literal 
	or a \pvariable this is straightforward.

	\begin{inparaitem}[\noindent- Case:]
	\item
	$\psi$ is a diamond modality the form $\ddiamond{a}{\varphi}$.
	Then  by definition of $\transposeact{\stdtr}{x}{y}$:
	\begin{align*}
		&\stdstate \in 
		\transposepv{\stdasst}{x}{y} \envelopemu{\transposefml{(\ddiamond{\stdtr}{\varphi})}{x}{y}}
		=
		\transposepv{\stdasst}{x}{y} \envelopemu{\ddiamond{\transposeact{\stdtr}{x}{y}}{\varphi\tfrac{y}{x}}}
		\\
		\text{iff} \quad&
		\exists (\stdstate,\newstate)\in \structA\envelope{\transposeact{\stdtr}{x}{y}} \;\newstate\in\transposepv{\stdasst}{x}{y}\envelopemu{\varphi\tfrac{y}{x}}
		\\
		\text{iff} \quad&
		\exists 
		(\transposestate{\stdstate}{x}{y},\transposestate{\newstate}{x}{y}) \in \structA\envelope{\stdtr} \;\transposepv{\newstate}{x}{y}\in\stdasst\envelopemu{\varphi}
		\\
		\text{iff} \quad& 
		\transposestate{\stdstate}{x}{y} \in \stdasst
			\envelopemu{\ddiamond{\stdtr}{\varphi}}
	\end{align*}

	\item
	 $\psi$ is of the form $\lfp{Z}{\varphi}$.
	Let  $\Gamma_1(D) = \stdasst\tfrac{D}{Z} \envelope{\varphi}$
	and
	\[\Gamma_2(D) = \transposepv{\stdasst}{x}{y}\tfrac{D}{Z}\envelopemu{\transposefmla{\varphi}{x}{y}}=(\stdasst\tfrac{\transposepv{D}{x}{y}}{Z}\envelopemu{\varphi})_x^y.\]
	Define $E^i_\gamma$ for $i\in \{1,2\}$ and $\gamma$ any ordinal recursively by $E^i_0=\emptyset$ and $E^i_{\gamma+1}=\Gamma_i(E_{\gamma^i})$. At limit ordinals take unions.
	Then by the Knaster-Tarski fixpoint theorem
	\[\stdasst\envelopemu{\lfp{Z}{\varphi}}= \bigcup_{\gamma<\infty}E^1_\gamma
	\quad\text{and}\quad
		\transposepv{\stdasst}{x}{y}\envelopemu{\transposefmla{\lfp{Z}{\varphi}}{x}{y}} = \bigcup_{\gamma<\infty}E^2_\gamma
	\]
	By induction it is easy to see that $E_\gamma^2 = \transposepv{(E_\gamma^1)}{x}{y}$ for all $\gamma$.
	Hence 
		\[(\stdasst\envelopemu{\lfp{Z}{\varphi}})_x^y= \bigcup_{\gamma<\infty}(E^1_\gamma)_x^y = \transposepv{\stdasst}{x}{y}\envelopemu{\transposefmla{\lfp{Z}{\varphi}}{x}{y}}\]
	as required.
\end{inparaitem}
\end{proof}

Similarly for a \GL-formula $\varphi$ define $\transposefmla{\varphi}{x}{y}$. By a straightforward induction on $\varphi$ prove
\begin{lemma}\label{lem:substforreplgl}{}
	For \GL-formulas $\varphi$ and \GL-games $\gamma$:
	\begin{align*}
		\stdasst\envelopegl{\transposefml{\varphi}{x}{y}} 
		&= \transposefml{(\stdasst\envelopegl{\varphi})}{x}{y}\quad\text{and}\quad\\
	\stdasst\envelope{\transposefml{\gamma}{x}{y}}(D_x^y)& =(\stdasst\envelope{\gamma}(D))_x^y.
	\end{align*}
\end{lemma}

\begin{proof}
	By a straightforward induction on the definition of \GL-formulas and games.
	For $\gamma$ of the form $\prepeat{\beta}$ use the Knaster-Tarski 
	fixpoint theorem.
\end{proof}

\section{Auxiliary Results}
\label{app:auxax}

\subsection{Derived Axioms and Proof Rules}

Some of the syntactic proofs of \rref{sec:relatingcalculi}
require additional axioms that can be derived.
Those are introduced and collected here.
First for \GL. The proofs are straightforward.

\begin{lemma}\label{lem:replaceinloop}
	Suppose $\gamma$ is a \GL game such that for all \GL-formulas $\psi$
	 $$\infers[\GLcalc]\ddiamond{\gamma_1}{\psi}\lbisubjunct\ddiamond{\gamma_2}{\psi}.$$
	Then for all \GL-formulas $\rho$:
	$$\infers[\GLcalc]\ddiamond{\gamma_1^*}{\rho}\lbisubjunct\ddiamond{\gamma_2^*}{\rho}.$$
\end{lemma}

\begin{proof}
	Consider some \GL-formula $\rho$.
	For the left to right direction observe that 
	the assumption used for $\psi\equiv \ddiamond{\gamma_2^*}{\rho}$
	implies that 
	$$\infers[\GLcalc]\ddiamond{\gamma_1}{\ddiamond{\gamma_2^*}{\rho}}\limply\ddiamond{\gamma_2}{\ddiamond{\gamma_2^*}{\rho}}.$$
	By \irref{fixpointaxiom} if follows that
	$$\infers[\GLcalc]\rho \lor\ddiamond{\gamma_1}{\ddiamond{\gamma_2^*}{\rho}}\limply\ddiamond{\gamma_2^*}{\rho}.$$
	Hence the forward direction derives by \irref{fixpoint}.
	The backward implication is symmetric.
\end{proof}

\begin{lemma}\label{lem:derivedaxiomsgl}
	The following axioms are derived in the \GL-calculus.
	\begin{center}
		\begin{calculuscollection}
			\begin{calculus}
				\dinferenceRule[dchoice|$\didia{\cap}$]{demons choice axiom}
				{
				\linferenceRule[impl]
				{\ddiamond{\gamma_1\cap\gamma_2}{\varphi}}
				{\ddiamond{\gamma_1}{\varphi}\land\ddiamond{\gamma_2}{\varphi}}
				}{}
				\dinferenceRule[dcomposition|$ \didia{{;^d}}$]{demons composition axiom}
				{
				\linferenceRule[impl]
				{\ddiamond{\pdual{(\gamma_1;\gamma_2)}}{\varphi}}
				{\ddiamond{\pdual{\gamma_1};\pdual{\gamma_2}}{\varphi}}
				}{}
			\end{calculus}
		\end{calculuscollection}
	\end{center}
\end{lemma}

Some useful derived rules for \Lmu are also introduced.

\begin{lemma}\label{lem:derivedaxiomsmu}
	The following proof rules are admissible in the \Lmucalc-calculus.
	\begin{center}
		\begin{calculuscollection}
			\begin{calculus}
				\dinferenceRule[fixpointmonotonicity|M$\mu$]{monotonicity of least fixpoint}{
					\linferenceRule[sequent]{
						\psi\limply \varphi
					}{
						\lfp{X}{\psi} \limply \lfp{X}{\varphi}
					}}{}
				\dinferenceRule[fixpointnumonotonicity|M$\nu$]{monotonicity of greatest fixpoint}{
					\linferenceRule[sequent]{
						\psi\limply \varphi
					}{
						\gfp{X}{\psi} \limply \gfp{X}{\varphi}
					}}{}
			\end{calculus}\quad
			\begin{calculus}
			\dinferenceRule[atomicboxmonotonicity|{[M$_\stdtr$]}]{monotonicity}{
				\linferenceRule[sequent]{
					\psi\limply \varphi
				}{
					\dbox{\stdtr}{\psi} \limply \dbox{\stdtr}{\varphi}
				}}{}
			\end{calculus}
		\end{calculuscollection}
	\end{center}
\end{lemma}

\subsection{Local Reduction for Game Logic}

\begin{proposition} \label{prop:schematicgameremoval}{}
	Let $\Lc =(\Lcf, \Act)$ be a signature,
	$T$ a collection of \GL-formulas and
	$\Lambda \subseteq \Act$.
	Assume that for every atomic transition $\stdtr\in \Lambda$
	there is some \GL-game $\gamma$ without $\Lambda$-transitions, such that
	$T \infers[\GLcalc] \ddiamond{\stdtr}{\psi}
	\lbisubjunct\ddiamond{\gamma}{\psi}$
	for all \GL-formulas~$\psi$.
	Then the $\Lambda$-transition free fragment of \GL is $T$-provably expressive.
\end{proposition}

\begin{proof}
	Say that two games $\gamma_1$ and $\gamma_2$ are provably equivalent iff $T \infers[\GLcalc] \ddiamond{\gamma_1}{\psi}
	\lbisubjunct\ddiamond{\gamma_2}{\psi}$
	for all \GL-formulas $\psi$.

	Prove by induction on the definition of \GL-games and \GL-formulas that any \GL-formula
	and any \GL-game is provably equivalent to a formula or a game without $\Lambda$-transitions, respectively.
	Most cases are straightforward. 
	
	\begin{inparaitem}[\noindent- Case:]
	\item
	$\ddiamond{\gamma_1}{\psi_1}$.
	By induction hypothesis there are $\gamma_2$ and
	$\psi_2$ without $\Lambda$-transitions
	which are provably equivalent to $\gamma_1$ and $\psi_1$ respectively.
	By \irref{monotonicity} it follows that
	$T \infers[\GLcalc] \ddiamond{\gamma_1}{\psi_1}
	\lbisubjunct\ddiamond{\gamma_1}{\psi_2}$.
	Because $\gamma_1$ and $\gamma_2$ are provably equivalent
	this implies $T \infers[\GLcalc] \ddiamond{\gamma_1}{\psi_1}
	\lbisubjunct\ddiamond{\gamma_2}{\psi_2}$.

	\item For atomic transitions $\stdtr\in \Act$ there are two cases.
	If $\stdtr \in \Lambda$ there is a provably equivalent $\Lambda$-transition free game by assumption. 
	Otherwise there is nothing to do.

	\item For programs of the form $\prepeat{\gamma_1}$, the induction hypothesis for $\gamma_1$
	implies that there is a $\gamma_2$ without $\Lambda$ transitions such that
	$T \infers[\GLcalc] \ddiamond{\gamma_1}{\psi}
	\lbisubjunct\ddiamond{\gamma_2}{\psi}$
	for all $\Lambda$-free \GL-formulas $\psi$.
	Consider any $\Lambda$-free \GL-formula $\rho$.
	Then  \rref{lem:replaceinloop} implies
	that $\prepeat{\gamma_1}$ and $\prepeat{\gamma_2}$ are equivalent.

	\item For programs of the form $\pdual{\gamma_1}$ 	the induction hypothesis implies that $\gamma_1$ is equivalent to some $\Lambda$-transition-free game~$\gamma_2$. Propositional reasoning implies
	\(T \infers[\GLcalc] \lnot\ddiamond{\gamma_1}{\lnot\psi}
	\lbisubjunct\lnot\ddiamond{\gamma_2}{\lnot\psi}\)
	for all \GL-formulas $\psi$.
	Propositional reasoning and \irref{demon} imply
	\(T \infers[\GLcalc] \ddiamond{\pdual{\gamma_1}}{\psi}
	\lbisubjunct\ddiamond{\pdual{\gamma_2}}{\psi}\)
	for all \GL-formulas $\psi$.
	\end{inparaitem}
\end{proof}

\iflongversion
\else
\section{Proofs} \label{app:proofs}
\printProofs
\fi

\end{document}